\documentclass[11pt,onecolumn,oneside]{IEEEtran}

\IEEEoverridecommandlockouts                              % This command is only needed if 
\usepackage{graphics} % for pdf, bitmapped graphics files
\usepackage{color}
\usepackage{bm}
\usepackage{epstopdf}
\usepackage[geometry]{ifsym}
\usepackage{ifsym}
\usepackage{bbm}
\usepackage{amscd}
\usepackage{dsfont}
\usepackage{amsmath}
\usepackage{epsfig}
\usepackage{graphics}
\usepackage{psfrag}
\usepackage{rotating}
\usepackage{url}
\usepackage{lscape}
\usepackage{color}
\usepackage[caption=false,font=footnotesize]{subfig}
\usepackage{epstopdf}
\usepackage{amsthm}
\usepackage{pgfplots}
\usepackage{tikz}
\usepackage{bm}
\usepackage{array}
\usetikzlibrary{arrows,shapes,decorations,backgrounds}
\usepackage[margin=1in,footskip=0.25in]{geometry}
\usepackage{hhline}
\newtheorem{theorem}{Theorem}
\newtheorem{lemma}[theorem]{Lemma}

\newtheorem{definition}[theorem]{Definition}
\newtheorem{proposition}[theorem]{Proposition}

\usepackage{amsfonts}
\usepackage{amsmath}
\usepackage{algorithm}
\usepackage{algorithmicx}
\usepackage{algorithm,algpseudocode}
\usepackage{verbatim}
\usepackage{geometry}
\usepackage{float}
\usepackage{amssymb}
\usepackage{graphicx}
\usepackage{epstopdf}
\usepackage{fancyhdr}
\usepackage{blindtext}
\usepackage{chngcntr}
\usepackage{afterpage}
\usepackage{pifont}
\usepackage{etoolbox}
\usepackage{pdfpages}
\usepackage{mathtools, cuted}
\usepackage{lipsum, color}
\usepackage[justification=centering]{caption}
\usepackage{hyperref}

\ifCLASSINFOpdf
\else
\fi

\begin{document}
\pdfoutput=1
%nohypertex
\title{\LARGE \bf
Stability Analysis of TDD Networks Revisited: \\ A trade-off between Complexity and Precision
}

\bibliographystyle{IEEEtran}
\author{\IEEEauthorblockN{ %\IEEEauthorrefmark{4},
 Rita Ibrahim$^{+*}$, Mohamad Assaad$^*$, Berna Sayrac$^+$, and Anthony Ephremides$^\times$\\
\small{$^+$Orange Labs, France.}\\
\small{$^*$TCL chair on 5G, Laboratoire des Signaux et Systemes, CentraleSupelec, France.}\\
\small{$^\times$ Institute for Systems Research, University of Maryland, USA. }\\
\small{E-mails: rita.ibrahim, berna.sayrac@orange.com, Mohamad.Assaad@centralesupelec.fr, etony@umd.edu}
}}

\maketitle

\thispagestyle{empty}
\pagestyle{empty}

\bibliographystyle{IEEEtran}

%%%%%%%%%%%%%%%%%%%%%%%%%%%%%%%%%%%%%%%%%%%%%%%%%%%%%%%%%%%%%%%%%%%%%%%%%%%%%%%%
\begin{abstract} 
In this paper, we revisit the stability region of a cellular time division duplex (TDD) network. We characterize the queuing stability region of a network model that consists of two types of communications: (i) users communicating with the base station  and (ii) users communicating with each other by passing through the base station. When a communication passes through the base station (BS) then a packet cannot be delivered to the destination UE until it is first received by the BS queue from the source UE. Due to the relaying functionality at the BS level, a coupling is created between the queues of the source users and the BS queues. In addition, contrarily to the majority of the existing works where an ON/OFF model of transmission is considered, we assume a link adaptation model (i.e. multiple rate model) where the bit rate of a link depends on its radio conditions. The coupling between the queues as well as the multiple rate model are the main challenges that highly increase the complexity of the stability region characterization. In this paper, we propose a simple approach that permits to overcome these challenges and to provide a full characterization of the exact stability region as a convex polytope with a finite number of vertices. An approximated model is proposed for reducing the computational complexity of the exact stability region. For the multi-user scenario, a trade-off is established between the complexity and the preciseness of the approximated stability region compared to the exact one. Furthermore, numerical results are presented to corroborate our claims.

\end{abstract}
%%%%%%%%%%%%%%%%%%%%%%%%%%%%%%%%%%%%%%%%%%%%%%%%%%%%%%%%%%%%%%%%%%%%%%%%%%%%%%%%
%%%%%%%%%%%%%%%%%%%%%%%%%%%%%%%%%%%%%%%%%%%%%%%%%%%%%%%%%%%%%%%%%%%%%%%%%%%%%%%%

\section{Introduction} 
%\subsection{The time division duplex systems}
Time division duplex (TDD) systems use a single frequency band for both uplink (UL) and downlink (DL) traffic which offers the flexibility to adjust the UL and DL channels based on their respective traffic demand. TDD provides dynamic UL and DL bandwidth allocation which allows the network to combine spectrum bands and achieve greater spectral efficiency when customers need it most. \iffalse Additionally, massive MIMO has been conceived in such a way that DL beamforming relies on UL pilot measurements based on the use of reciprocity and TDD operation.\fi With the growth of various new applications with high traffic and data rate demand, new techniques have been investigated to fulfill the requirements of future cellular networks. The severity of this situation will increase with fifth generation (5G) cellular networks that demand the support of higher data rates and lower latency. By enabling a seamlessly adaptation of the spectrum bands and an efficient UL/DL load asymmetry, TDD can be used for improving the capacity of dense area with high mobile data demand. Besides, motivated by the use of Massive MIMO, TDD is mainly adopted in 5G systems. \iffalse most of the new spectrum bands for 5G are TDD. Therefore, 5G will rely mostly on TDD.\fi

%\subsection{The stability region as performance metric}
The physical layer study of cellular networks provides the information-theoretic capacity of these scenarios while assuming that the queues are saturated (not empty). In this case, the rate region is defined as the set of the achievable bit rates of the system of saturated queues. However, integrating the network level (e.g. \cite{Rongstability} and \cite{SimeoneStableThroughput}) to the study has demonstrated important gains in terms of throughput and delay. Under bursty traffic arrivals, the stability region becomes a relevant measure of the queues' bit rates. The stability region is defined as the set of arrival rates that can be supported by the network under the stability constraint of the queues (i.e. as defined in section \ref{sec:System-Model}, a queue is considered stable if its mean arrival rate is lower than its mean departure rate). The stability region has several definitions as one can see in \cite{szpankowski1992stabilty}. The difference between capacity region and stability region has been the subject of interests of several studies (e.g. \cite{GeorgiadisStabilityCapacity} and \cite{LuoThroughputCapacity}).

%\subsection{The interacting queues model}
Queuing stability region is the performance metric used in this work. The consideration of bursty traffic and queuing analysis is motivated by the following scenario. In TDD cellular networks, any pair of users communicate with each other by passing through the base station (BS) where the UL and DL parts of this communication compete for the same spectrum. Therefore, the BS plays the role of a relay that receives packets from the first device (source) on the Uplink (UL) stacks them in a buffer and then transmits them to the second device (destination) on the Downlink (DL) using an opportunistic scheduling algorithm. In a TDD system, the UL and DL transmissions are performed on different timeslots, hence a packet cannot be transmitted on the downlink (BS-to-destination) until it is first received by the BS on the uplink (source-to-BS). During some time-slots, even when the channel state of the BS-destination link (DL) is favorable, it may happen that the buffer at the BS is empty (i.e no need to schedule this DL). This coupling cannot be captured by a simple performance analysis at the physical layer. Hence, a traffic pattern must be included in the analysis in order to provide a more realistic evaluation of the network capacity. The particularity of this study is the simultaneous consideration of UL and DL communications which implies a relaying operation at the BS level and a coupling between the system of queues (i.e. contrarily to the majority of the existing works where either UL or DL communications are studied).

\iffalse
\begin{itemize}
\item Work on cellular networks that are restricted on either UL and DL and give some references. The speciality of this work is cosidering both UL and DL in the same time as well as the relay functionality of the BS\\
\item Cooperative relaying (almost all on aloha + differences on/off model/ fading , no centralized protocol, and not detailed results on the characterization of SR of more than 2 users)\\
\item cooperative relays is similar to interacting queues which has been studied etcInteracting queues \\
\item motivation and organization
\end{itemize}
\fi
%\subsection{Related works }
Due to the relaying aspects at the BS level, it is interesting to examine the existing works regarding cooperative and relaying networks. The approach of cooperative relaying and multi-hop wireless networks has been studied at: the physical layer (e.g. \cite{Sendonaris2003} and \cite{Mohajer2009}) and the traffic layer (e.g. \cite{Sadek2007} and \cite{Rongstability}).  However, these works consider only simple scenarios (i.e. three-node network or one relay and one destination scenario). Indeed, the majority of the works in the area of relaying networks consider an ALOHA random access system: (i) \cite{Rongstability} and \cite{Pappas2016} use stochastic dominance technique, (ii) \cite{Kashef2016} considers energy-harvesting capabilities and computes the relaying parameter that maximizes the stable throughput rate of the source and (iii)  \cite{hong2007stability} proposes two queuing strategies for cooperative networks. Furthermore, wireless multiple-access system with probabilistic channel receptions was studied in \cite{Rong2012Cooperative} for a network of $N$ sources communicating with one destination; the impact of a protocol-level cooperation is investigated in such wireless network. Authors in \cite{Sadek2007} have introduced a new cognitive multiple-access protocol in a relay assisted network. Both \cite{Sadek2007} and \cite{Rong2012Cooperative} take into account the time division multiple access (TDMA) as a scheduling policy. Furthermore, the majority of these studies consider simple networks (i.e. consisting of one source to destination communication aided by a relay node) due to the difficulty that one can face while characterizing the stability region for multi-user cooperative networks.

Both the bursty traffic and the relaying role of the BS lead to a coupling between the queues in the system such that the service rate of each queue will depend on the state of the other queues. The stability region characterization of the system of interacting queues has been a challenging problem and has received the attention of researchers (i.e. especially for the case of multiple-access channel networks). For the slotted  ALOHA system, several approximated models were proposed (e.g. \cite{SaadawiAnalysis}, \cite{Sidi1984} and \cite{Bordenave2012}) and some bounds for the ergodicity region have been obtained (e.g. \cite{szpankowski1984ergodicity} and \cite{Szpankowski1988_}). Several tools were proposed in the literature to overcome the challenges of the system of interacting queues. The stochastic dominance technique was introduced in \cite{Rao2006}. It consists of elaborating a simple way for capturing the interaction between the queues by studying simple auxiliary systems of queues that dominate the system of interest. This technique is especially applied to simple scenarios in order to: (i) obtain the stability region (e.g. \cite{Jeon2011}, \cite{Kompella2011} and  \cite{Pappas2013}) and (ii) characterize the stability bounds for the system of interacting queues (e.g. \cite{Pappas2012}). The stability of a coupled system of queues where the service rate of each queue depends on the number of customers in all the other queues was extended in \cite{Borst2008}. Their stability approach consists of deriving marginal drift criteria for multi-class birth and death processes. This analysis was adopted in \cite{vitale2015Perf} for proposing a coupled processors model that was applied for the study of underlay device-to-device communications in \cite{Vitale2015Modeling}.  Furthermore, to deal with the coupling between the queues, authors in \cite{CIARDO19933} and \cite{Mainkar1995} propose a  decomposition model and an iterative approach in Stochastic Petri Nets.

The majority of the previous works consider ALOHA random access (i.e. without a centralized scheduling) contrariwise to our scheme where a TDD cellular system with a centralized TDMA scheduling is investigated. Preceding efforts were limited to simple scenarios (i.e. three-node scenario or single destination scenario) with an elementary ON/OFF model of transmission. However, in this work we analytically characterize the stability region of TDD cellular systems (i.e. both UL and DL communications) with a TDMA scheduling and where both multi-user and multiple rate model are taken into account.

%Stability region analysis 
%\subsection{Related work and motivation}
%\subsection{Can be interesting to see cooperative networks a }
%\subsection{Motivation and organization}
In this paper, we consider the case of a discrete-time slotted TDD system with $2K+U$ users (UE), each of which has a buffer of infinite capacity. $U$ users communicate with the BS (denoted by \textit{UE2BS communications} in the sequel) and $K$ pairs of users communicate with each other by passing through the BS (denoted by \textit{UE2UE communications} in the sequel). Under a TDMA policy, we characterize the stability region of the network. As we mentioned before, what makes the problem challenging is the BS functionality as a relay between the UL and the DL of the \textit{UE2UE communications} (i.e. a packet cannot be transmitted on the DL if it is not received by the BS on the UL). Hence, the performance of the queues will depend on the state of the queues (empty or not) at the BS level.

%\subsection{Several differences with the existing work}
Most of the works concerning the analysis of stability region in cellular network are focused on the strict consideration of either uplink or downlink communications. However, in our scenario we consider \textit{UE2UE communications} that take place on both uplink and downlink in such a way that the BS plays the role of a relay between the source UE and the destination UE. We revisit the queuing analysis of such TDD scenario in order to characterize its stability region as a convex polytope with a limited number of vertices and to derive the existing threshold between the complexity and the precision of this characterization. The key particularities of this work are summarized as follows:
\begin{itemize}
\iffalse
The particularity of our work is the consideration of a queuing analysis approach of a TDD system with a TDMA user scheduling such that only one communication is scheduled per time-slot (depending on the network scheduling policy, the channel states of the links and the queues' states). The full characterization of the stability region for this system is given as a convex polytope with a limited number of vertices. The novelty is to derive the threshold between the complexity of this characterization and its precision. The key differences of this work are summarized as follows:

\item The BS relays the UL and DL traffic of the \textit{UE2UE communications} based on a TDMA scheduling contrary to the majority of the cooperative networks studied in the literature where an ALOHA random access system is considered. For this latter scheme, the complexity of the interacting queues increases in such a way that the explicit form of the stability region was found only for a simple three-users scenario. Hence, considering a TDMA scheduling reduces the complexity of the system such that the performance of the queues depends only on the state (empty or not) of the queues at the BS level. This coupling will be discussed in details later in the paper. 
\fi
\item Motivated by capturing the relaying effect at the BS level, a bursty traffic is considered and the performance of the network is evaluated in terms of stability region based on a queuing theory approach. As shown afterward in this paper, the queuing analysis approach brings important additional outcomes (i.e. in terms of network capacity) compared to the physical layer only approach. We illustrate the difference in terms of performance evaluation between these two approaches: (i) approach that takes into account the impact of coupling between the queues (the service rate of a queue will depend not only on the distance, fading, bit error rate and transmission power but also on the state of the other queues) and (ii) approach that does not take into account the coupling between the queues (i.e. assuming full buffer queues where the bit rate is considered at the physical layer without any bursty traffic). 

\item The relaying approach in this scenario differs from the preceding works regarding the cooperative relaying networks as it follows: (i) the presence of the BS (i.e. as a relay) is mandatory for enabling the communication between two users of the network contrarily to previous works where the relay is added to the network in order to use the spatial diversity and to enhance the performance of the network by enabling cooperative communication, (ii) the majority of the previous works considers simple scenario (i.e. single relay and single destination scenario or three-node models) whereas in this work multi-user scenario is considered where the centralized entity (BS), that has the global state information of the network, plays the role of relay in the network, (iii) a coupling exists between the queues such that the service rate of a queue depend on the state of the other queues (i.e. especially on the queues at the BS level and not on the sources' queues) and (iv) the TDD cellular scenario with a centralized TDMA scheduling differs from the existing studies focused on ALOHA random access system.

\iffalse at the BS level, there are $K$ independent queues such that the choice of the scheduled queue depends only on the TDMA policy hence no need need for a relay-level strategy for choosing the type of traffic (own or relayed traffic) to transmit once the corresponding queue is scheduled.\fi

\item We assume a link adaptation model rather than single rate model (ON/OFF model). It corresponds to the matching of the bit rate to the radio conditions (i.e. SNR) of the link. This realistic assumption makes the analysis more complicated as one can see in the sequel.

\item In the literature, the stability region analysis of interacting queues with relaying functionality was mainly analyzed either: (i) based on stochastic dominance technique for finding the stability region (mainly for aloha systems with three-user scenario and ON/OFF transmission model) or at least describing the necessary and sufficient conditions for queuing stability or (ii) via numerical and simulation results. However, in this paper we approach the problem in a different way and we transform the multidimensional Markov Chain that models the network to multiple one dimensional (1D) Markov Chain models (as detailed in the sequel). After evaluating the different 1D Markov Chains, we characterize the exact stability region as a convex polytope with a limited number of vertices. This exact stability region turns to have a high complexity. Therefore, we propose approximated models characterized by having an explicit analytic form and a low complexity while a high precision is guaranteed (i.e compared to the exact stability region). We start by considering the simple scenario that consists of one \textit{UE2UE} and one \textit{UE2BS communication}. In theorem \ref{th_ss_cellular_real}, we characterize the exact stability region of this simple scenario which turns to have a non explicit complex form. Hence, we propose a closed-form tight upper bound of the exact stability region (see theorem \ref{th3_gen}). Moreover, for the general scenario of multiple \textit{UE2UE} and \textit{UE2BS communications},  we evaluate the exact stability region and we discuss its computational complexity (see theorem \ref{th1_oneRate_MU}). Thus, we propose two techniques to respectively reduce the complexity of the exact stability region (see theorems \ref{th2_oneRate_MU} and \ref{th3_oneRate_MU}). We deduce the trade-off between the complexity and the precision of the stability region computation.
  
\end{itemize}
  
The remainder of this paper is organized as follows. Section \ref{sec:System-Model}
describes the system model. The stability region analysis for a three-UEs simple case is provided in Section \ref{sec:3UE_scenario} in order to have a clear perception of the advantages of overlay D2D compared to cellular communications. The stability region of the multi-UE general case is presented in Section \ref{sec:MUE_scenario} as a simple convex polytope. Numerical results are presented in Section \ref{sec:Numerical}. Section  \ref{sec:Conclusion} concludes the paper whereas the proofs are provided in the appendices. \footnote{Part of the simple scenario case has been presented in the paper: "Overlay D2D vs. Cellular communications: a stability region analysis" published in ISWCS 2017} \\

\section{System Model\label{sec:System-Model}}

In this study, the considered scenario consists of two sets of communications: (i) set of \textit{UE2BS communications} (where users are transmitting packets to the BS e.g. to contact users in other cells or to have internet connections, etc.) and  (ii) set of \textit{UE2UE communications} (between pair of UEs).  \textit{UE2UE communications} are performed by passing through the BS.

We consider a single cell scenario with $K$ \textit{UE2UE communications} and $U$ \textit{UE2BS communications}. In other terms, we suppose $K+U$ communications and $2K+U$ users in the cell. We denote by UE$_{i,s}$ and UE$_{i,d}$ the pair source and destination users corresponding to the  $i^{th}$ \textit{UE2UE communication} (for all $1\leq i\leq K$) and  by UE$_{j,u}$ the user corresponding to the $j^{th}$ uplink \textit{UE2BS communication} (for all $1\leq j\leq U$). Let us describe the cellular scenario illustrated in Figure \ref{fig.SystemModel_MU}.

In this scenario, if two devices want to communicate with each other, they  must exchange their packets through the BS. The communication between UE$_{i,s}$ and UE$_{i,d}$ is performed through the BS (for all $1\leq i\leq K$) such that the BS transmits packets to the destination user UE$_{i,d}$ only if it receives them from the source user UE$_{i,s}$. Therefore, for each \textit{ UE2UE communication} corresponds a buffer at the base station that could be empty during some time slots.  Hence, a coupling exists between the queues such that the service rate of the users' queues depend on the state being empty or not of the BS which makes the queuing stability analysis challenging. In this scenario the following links exist: link$_{i,s}$: UE$_{i,s}$ - BS; link$_{i,d}$: BS$\,$-$\,$UE$_{i,d}$ and link$_{j,u}$: UE$_{j,u}\,$-$\,$BS (with $1\leq i\leq K$ and $1\leq j\leq U$).\\

The considered network consists of two type of communications \textit{UE2UE communications} and \textit{UE2BS communications}. In practice, these two types of communications coexist with downlink communications (i.e. between the BS and the users) that aim to deliver downlink traffic to mobile users  (e.g. internet connections or receive calls from users in other cells etc.). The extension of this work to the scenario where this downlink traffic is taken into account is straightforward. \iffalse Please note that in this work, we consider two types considering the \textit{UE2BS communications} as downlinks (and not uplinks) is a straightforward process. Hence, for clarity reasons we limit this work to the case where the \textit{UE2BS communications} consists of UL communications. \fi
\subsection{Priority Policies} 
We call \textit{priority policy} the sorting of the communications' priorities according to which the users are chosen for transmission. In other terms, among the users that are able to transmit (which means have some packets to transmit and have the required radio conditions to do it), the UE that is chosen to transmit is the one that has the highest priority according to the considered priority policy. Hence, a user is scheduled only when all the more prioritized users are not able to transmit. We denote by $\Omega_\Gamma$ the set of all the possible priority policies. $\Gamma \in \Omega_\Gamma$ denotes a priority policy according to which the users are chosen for transmission. One can see that for $N$ communications the number of possible priority policies is given by the number of existing permutations: $N!$.

Note that any other scheduling policy is nothing but a convex combination of these priority policies. For this reason, our work is based on studying these priority policies that characterize the corner points of the stability region. Any other scheduling corresponds to an interior point of the characterized stability region.

These priority policies allows us to avoid the multidimensional Markov Chain modeling of the interacting queues. Thus, for a given priority policy, each queue can be modeled by a one dimensional (1D) Markov chain. Our approach transforms the multidimensional Markov Chain, that captures the dependency between the queues, to a 1D Markov Chain model for a given priority policy. However, the modeling that we propose remains challenging due to the coupling between the queues. As one can see later, an additional analysis is required for capturing the interaction between the queues.

\iffalse
 at each time-slot the UE chosen to transmit corresponds to the most prioritized user that is able to transmit (has packet to transmit and has the required radio condition). Hence, a user is scheduled only when all the more prioritized users are not able to transmit.

 We suppose that $\Omega_\Gamma$ is the set of all the possible sorting priorities of the users. We denote by $\Gamma \in \Omega_\Gamma$ as a priority policy according to which the users are chosen for transmission. The scheduling is done in such a way that the chosen UE corresponds to the most prioritized user that is able to transmit. Hence, a user is scheduled only when all the more prioritized users are not able to transmit. \\
\fi
\begin{centering}
\begin{figure} 
\centering
\captionsetup{justification=centering}
\includegraphics[width=0.4\textwidth]{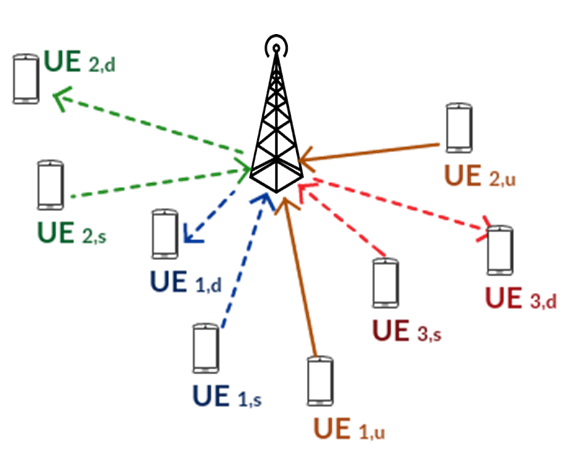}
%\includegraphics[scale=0.5]{../Figures/.jpg} 
%\captionsetup{justification=centering}
\caption{General Cellular Scenario}
\label{fig.SystemModel_MU} 
\end{figure}
\end{centering}

\subsection{System of queues}
We consider a system of queues to describe the studied scenario. UE$_{j,u}$ (for all $1\leq j\leq U$) communicate with the BS through an uplink cellular communication and the queue of user UE$_{j,u}$ is represented by Q$_{j,u}$. The communication between UE$_{i,s}$ and UE$_{i,d}$ (for all $1\leq i\leq K$) is represented by two consecutive queues: the uplink queue Q$_{i,s}$ of UE$_{i,s}$ and the download queue Q$_{i,BS}$ (see figure (\ref{fig.Queues_MU})). The BS does not transmit to UE$_{i,d}$ unless it has received at least one packet from UE$_{i,s}$. This coupling between the queues induces that the service rates of all the queues Q$_{i,s}$ and Q$_{i,u}$ depend on the state (empty / not empty) of each Q$_{i,BS}$ which makes the queuing stability analysis challenging.  The users' queues are assumed saturated. The traffic arriving to the queues Q$_{i,s}$ and Q$_{j,u}$ is time varying, i.i.d. over time and with rate respectively equal to $\lambda_{i,s}$ and $\lambda_{j,u}$ for $1\leq i\leq K$ and $1\leq j\leq U$. The traffic arriving to Q$_{i,BS}$ is nothing but the departure from Q$_{i,s}$.

The traffic departure from the users' queues is also time varying and depends on the queues' states (empty or not), the scheduling allocation decision and the time varying channel conditions. For a given priority policy $\Gamma$, the average service rates of the queues Q$_{i,s}$, Q$_{i,BS}$ and Q$_{j,u}$ are respectively denoted by $\mu_{i,s}\left( \Gamma\right)$ and $\mu_{i,d}\left( \Gamma\right)$ and $\mu_{j,u}\left( \Gamma\right)$. The vector that describes the service rate of the users' queues for a given policy $\Gamma$ is the following: 
\[
\bm{\mu}\left(\Gamma\right)=\left[ \mu_{1,s}\left(\Gamma\right) \,\,\, \mu_{2,s}\left(\Gamma\right) \,\,\,...\,\,\, \mu_{K,s}\left(\Gamma\right) \,\,\, \mu_{1,u}\left(\Gamma\right) \,\,\, \mu_{2,u}\left(\Gamma\right)\,\,\,  ... \,\,\, \mu_{U,u}\left(\Gamma\right) \right]
\]

The vector that describe the arrival rates of the users' queues is given by:
\[
\bm{\lambda}=\left[ \lambda_{1,s} \,\,\, \lambda_{2,s}\,\,\,...\,\,\, \lambda_{K,s} \,\,\, \lambda_{1,u}\,\,\, \lambda_{2,u}\,\,\,  ... \,\,\, \lambda_{U,u} \right]
\]

\begin{figure} [H]
\vspace{-10pt}
\begin{centering} 
\includegraphics[width=0.3\textwidth]{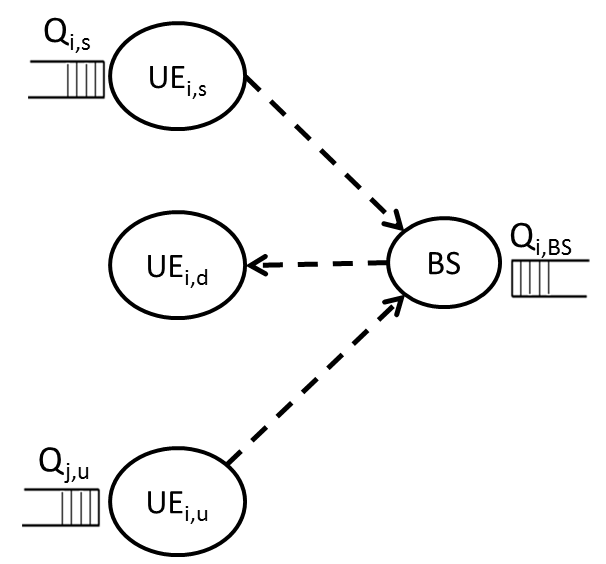}
\captionsetup{justification=centering}
\caption{The queues model of the cellular scenario}
\label{fig.Queues_MU} 
\end{centering}
\vspace{-10pt}
\end{figure}

\subsection{Set of bit rates}

 We consider an adaptive modulation scheme such that the transmission rates is improved by exploiting the channel state at the transmitter. The SNR values are divided into a finite set of intervals $\left[S_{1},...,S_{M}\right]$ where the $j^{th}$ interval $S_j=\left[\gamma_{j},\gamma_{j-1}\right]$  (called hereinafter $j^{th}$ state) is characterized by two SNR thresholds $\gamma_{j-1}$ and $\gamma_{j}$ such that $\gamma_{j-1} > \gamma_{j}$. Here, the adaptive modulation consists on considering a finite set $\left[r_{1},...,r_{M}\right]$ of bit rates as a mapping of the SNR intervals. Therefore, if a link has a SNR within the state $S_j$ than the bit rate if this transmission is $r_j$. It is worth mentioning that to transmit at a rate $r_j$ in the downlink or the uplink, the SNR states and thresholds may be different. In order to deal with that we use $\gamma_{j}^{UL}$ and $\gamma_{j}^{DL}$ to describe the SNR thresholds for respectively the uplinks and downlinks. 
 
 \iffalse
 
The bit rate of each user changes from one timeslot to another depending on fast fading and CSI conditions. The bit rate at each timeslot takes a value that lies within the finite set $\left[r_{1},...,r_{M}\right]$. The scheduled link transmits at bit rate $r_{n}$ if its instantaneous SNR lies within the interval $S_n=\left[\gamma_{n},\gamma_{n-1}\right]$ (called hereinafter $n^{th}$ state) where $\gamma_{n-1}$ and $\gamma_{n}$
are two SNR thresholds with $\gamma_{n-1} > \gamma_{n}$. There exists a matching between the set of SNR  intervals $\left[S_{1},...,S_{M}\right]$ and the set of bit rates $\left[r_{1},...,r_{M}\right]$. The channel quality of a link $i$ is reflected by its probabilities $p^n_{i}$ to have a SNR within the $n^{th}$
SNR interval $S_n$. It is worth mentioning that to transmit at a rate $r_n$ in the downlink or the uplink, the SNR threshold $\gamma_n$ may be different. In order to deal with such a case, we use $\gamma_{n}^{UL}$ and $\gamma_{n}^{DL}$  to describe respectively the uplinks and downlinks SNR thresholds. 
\fi
\subsection{Probabilities for channel quality}
The channel between any two nodes in the network is modeled as a Rayleigh fading channel that remains constant during one time slot and changes
independently from one time slot to another based on a complex Gaussian distribution with zero mean and unit variance. The received SNR for
a link $i$ is given by: \[SNR_{i}=\frac{|h_{i}|^{2}d_{i}^{-\beta}P_{i}}{N_{0}}\] where $h_{i}$ is the fading coefficient, $d_{i}$ is the distance between source and destination, $P{}_{i}$ is the transmission power, $\beta$ is the path loss exponent and $N_{0}$ is the noise. We denote by $p^n_{i}$ the probability that the SNR of the link $i$ is within $n^{th}$ interval (i.e. in state $S_n$):

\begin{equation}
p^n_{i}=\mathbf{\mathbb{P}}\left[ SNR_{i}\in S_n\right]=\mathbf{\mathbb{P}}\left[ \gamma_{n}\leq SNR_{i}\leq \gamma_{n-1}\right]=\mathbf{\mathbb{P}}\left[SNR_{i} \geq \gamma_{n} \right]-\mathbf{\mathbb{P}}\left[ SNR_{i} \geq \gamma_{n-1} \right]
\label{pi_n}
\end{equation}

Given a complex Gaussian distribution of the channel $h_i$ then: 
\[
{\mathbb{P}}\left[SNR_{i}\geq\gamma_{n} \right] = \mathbf{\mathbb{P}}\left[|h{}_{i}|^{2}\geq\frac{\gamma_{n}N_0}{d_{i}^{-\beta}P_{i}}\right]=\exp\left(-\frac{\gamma_{n}N_0}{d_{i}^{-\beta}P_{i}}\right)
\]

Hence
\begin{equation}
p_i^{n}=\exp\left(-\frac{\gamma_{n}N_0}{d_{i}^{-\beta}P_{i}}\right)-\exp\left(-\frac{\gamma_{n-1}N_0}{d_{i}^{-\beta}P_{i}}\right)
\label{pi_n_Gauss}
\end{equation}

Note that $\mathbf{\mathbb{P}}\left[ SNR_{i}\geq +\infty\right]=0$ and $\mathbf{\mathbb{P}}\left[ SNR_{i}\geq -\infty\right]=1$. Let us consider the following notation $\bar{p}^n_{i}=1-p^n_{i}$.\\
\iffalse
The probability that the uplink of the $i^{th}$ \textit{UE2UE communication} transmits at rate $r_n$ is: 
\begin{equation}
p^n_{i,s}=\exp\left(-\frac{\gamma_{n}^{UL}N_0}{d_{i,s}^{-\beta}P_{UL}}\right)-\exp\left(-\frac{\gamma_{n-1}^{UL}N_0}{d_{i,s}^{-\beta}P_{UL}}\right)
\label{pn_is}
\end{equation}
The probability that the downlink of the $i^{th}$ \textit{UE2UE communication} transmits at rate $r_n$ is: 
\begin{equation}
p^n_{i,d}=\exp\left(-\frac{\gamma_{n}^{DL}N_0}{d_{i,d}^{-\beta}P_{DL}}\right)-\exp\left(-\frac{\gamma_{n-1}^{DL}N_0}{d_{i,d}^{-\beta}P_{DL}}\right)
\label{pn_id}
\end{equation}
The probability that the uplink of the $j^{th}$ \textit{UE2BS communication} transmits at rate $r_n$ is:
\begin{equation}
p^n_{j,u}=\exp\left(-\frac{\gamma_{n}^{UL}N_0}{d_{j,u}^{-\beta}P_{UL}}\right)-\exp\left(-\frac{\gamma_{n-1}^{UL}N_0}{d_{j,u}^{-\beta}P_{UL}}\right)
\label{pn_ju}
\end{equation}
The probability that the D2D of the $i^{th}$ \textit{UE2UE communication} transmits at rate $r_n$  is:
\begin{equation}
p^{n}_{i,sd}=\exp\left(-\frac{\gamma_{n}^{UL}N_0}{d_{i,sd}^{-\beta}P_{UL}}\right)-\exp\left(-\frac{\gamma_{n-1}^{UL}N_0}{d_{i,sd}^{-\beta}P_{UL}}\right)
\label{pn_isd}
\end{equation}

 In the aforementioned expressions, $d_{i,s}$ is the distance between UE$_{i,s}$ and BS, $d_{i,d}$ is the distance between UE$_{i,d}$ and BS, $d_{j,u}$ is the distance between UE$_{j,u}$ and BS, $d_{i,sd}$ is the distance between UE$_{i,s}$ and UE$_{i,d}$,  $P_{UL}$ is the user's transmission power and $P_{DL}$ is the  BS transmission power. \\
\fi

\subsection{Stability region definition}

Based on the definition in \cite{szpankowski1992stabilty}, the $i^{th}$ queue is stable if its length Q$_i\left(t\right)$ satisfies:
\[
\lim_{q\rightarrow\infty}\left[\lim_{t\rightarrow\infty}\mathbb{P}\left(Q_i\left(t\right)<q\right)\right]=1
\] 

It follows from Loyne's theorem \cite{loynes_1962} that if the arrival and service process of the $i^{th}$ queue are strictly jointly stationary then this queue is stable if $\lambda_i < \mu_i$ where  $\lambda_i$ and $\mu_i$ denote the mean arrival and service rate of the queue Q$_i$. Hence for a given scheduling policy $\xi$, the stability region $\mathcal{R}\left( \xi \right)$ is given as the closure set of arrival rate vectors for which all the queues are stables under the scheduling policy $\xi$.
\[
\mathcal{R}\left( \xi \right):=\left\lbrace \bm{\lambda}  \in \mathbb{R}_{+}^{K+U} |  \bm{\lambda}  \prec   \bm{\mu}\left(\xi\right) \right\rbrace
\]
where $\prec$ is the component-wise inequality, $\bm{\lambda}$ is the vector of the queues' average arrival rates  and $\bm{\mu}\left(\xi\right)$ is the vector of the queues' average service rates under the scheduling policy $\xi$.

The stability region $\mathcal{R}$ is defined as the union of the stability region for all the feasible scheduling policies (denoted by $\Omega_{\xi}$).
\[
\mathcal{R}:=\bigcup\limits_{\xi \in \Omega_{\xi}} \mathcal{R}\left( \xi \right)
\]

\subsection{Our approach for computing the stability region}

We based our approach on the fact that all the existing scheduling policies can be written as a convex combination of the priority policies. This gives the idea to study the stability region based on the defined priority policies. Thus, the stability region $\mathcal{R}$ is given by the union of the $\mathcal{R}\left( \Gamma \right)$ over all the feasible priority policies (denoted by $\Omega_{\Gamma}$). Note that any region $\mathcal{R}\left( \xi \right)$ for any other scheduling $\xi \not\in  \Omega_\Gamma$ is a subset of the region $\mathcal{R}\left( \Gamma \right)$ due to the fact that this scheduling policy $\xi$ can be written as a convex combination of the set of the priority policies $\Omega_\Gamma$. 
\[
\mathcal{R}:=\bigcup\limits_{\Gamma \in \Omega_{\Gamma}} \mathcal{R}\left( \Gamma \right)
\] 

 Hence, for characterizing the stability region we have to increase as much as possible the service rates of the users' queues. For the $i^{th}$ \textit{UE2UE communications}, this means increasing the service rate of the UL side (UE$_{i,s}$-BS) which implies the increase of the arrival rate at the DL side (BS-UE$_{i,d}$) and risks the loss of the stability of the BS queue Q${_{i,BS}}$. Recall that the service rates of both UL and DL sides are coupled through the scheduling.  Instability of  Q${_{i,BS}}$ means that many packets will not be delivered to UE$_{i,d}$ and the network becomes unstable (it means that average delay is not finite). 
 
 Here, the {\textbf{challenge}} is to characterize the stability region by finding the priority policies that maximize the service rates of the UL side ( hence achieve the border points of the stability region) while guaranteeing the Q${_{BS}}$ stability. This coupling between the stability region and Q${_{BS}}$  is very challenging in general in queuing theory (due to scheduling) and especially in the context of time varying wireless channels (i.e. fast fading) where the allocated bit rate changes according to the time varying channel state.

If we are able to find the policies that achieve the corner points of the stability region and we denote these policies by $\Omega_{\Gamma}^*$ then this set of policies is sufficient to describe the stability region of the system. Therefore, for each priority policy $\Gamma \in \Omega_{\Gamma}^*$, we find the region that capture the coupling between the queues and guarantee their stability. Hence, based on the definition (\ref{Def.Conv}) of the operator $co$, the stability region can be simply characterized by these corner points and given by:
\[
\mathcal{R}=co\left(\bigcup\limits_{\Gamma \in \Omega_{\Gamma^*}} \mathcal{R}\left( \Gamma \right)\right)
\]

\begin{definition}
\label{Def.Conv}
Consider a  finite point set $S\subset\mathbb{R}_+^n$ such that ${S}=\bigcup\limits_{i=1}^{|S|}\lbrace{x_i\rbrace}$, we define $co \left(S\right)$ as it follows:
\[
co\left( {S}\right)=\left\{\vphantom{ \sum\limits_{i=1}^{|{S}|}} x\in \mathbb{R}_+^n \mid \forall x_i \in S: \,\,x \prec x_i \right\} \bigcup \left\{ \sum\limits_{i=1}^{|{S}|}\gamma_i  x_i \mid \left(\forall i:\gamma_i\geq 0\right) \wedge \sum\limits_{i=1}^{|{S}|}\gamma_i\leq 1 \right\}
\]
\end{definition}

\subsection{Organization}
This work is divided into two sections. In the first section, we consider a simple scenario that consists of three users and contains one \textit{UE2UE communication} and one \textit{UE2BS communication}. For this three-UEs scenario, we consider a set of three bit rates $\left[r_1,r_2,r_3\right]$ such that the transmission rate of a link takes a value within this set depending on its channel state. We start by providing the exact stability region of this scenario. This stability region does not have an explicit form and turns be computationally complex (especially due to the consideration of a set of three bit-rates and not a ON/OFF model). In order to reduce complexity, we find an upper bound of this stability region and prove analytically that it is a very close approximation of the stability region. This upper bound that has a simple and explicit form reduces the complexity of the stability region computation. This simple scenario brings several interesting insights on the stability region of these cellular scenarios as well as the complexity of such study. This motivates us to elaborate a simple, explicit and good approximation of this stability region.

 In the second section, we consider the general scenario with $K$ \textit{UE2UE communications} and $U$ \textit{UE2BS communications}. For this general scenario, a set of two bit rates $\left[r_1,r_2\right]$ is considered due to the computation complexity of considering a set of three bit rates for this general case. Similarly to the simple case, we provide the exact stability region of the general scenario and we show the complexity of such computation. Hence, we propose a tight and explicit approximation that decreases the complexity of the stability region. The main novelty of this work is to capture, for such cellular scenarios, the trade-off between the complexity and the precision of the stability region calculation.
%%%%%%%%%%%%%%%%%%%%%%%%%%%%%%%%%%%%%%%%%%%%%%%%%%%%%%%%%%%%%%%%%%%%%%%%%%

\section{Three-UEs scenario\label{sec:3UE_scenario}}

\subsection{3-UEs scenario description}
The simple scenario (SS) consists of considering a single cell wireless network with three users in the cell  UE$_{1,s}$, UE$_{1,d}$, UE$_{1,u}$ respectively denoted by UE$_{s}$, UE$_{d}$, UE$_{u}$ . We assume that user UE$_{s}$ wants to communicate with UE$_{d}$, while UE$_{u}$ is transmitting packets to the BS. The communication between UE$_{s}$ and UE$_{d}$ is performed through the BS. Hence, we consider the following three links: link$_{s}$ UE{}$_{s}$ - BS, link$_{d}$: BS - UE$_{d}$, link$_{u}$: UE$_{u}$-BS. The 3-UEs scenario is illustrated in Figure \ref{fig.SystemModel_TwoUE}. 

\iffalse
\begin{figure} [H]
\begin{centering}
\includegraphics[width=0.9\textwidth]{figure13.png}
\caption{System Model}
\label{fig.SystemModel} 
\par\end{centering}
\end{figure}
\fi 

\begin{figure} [ptb]
\begin{centering}
\includegraphics[width=0.6\textwidth]{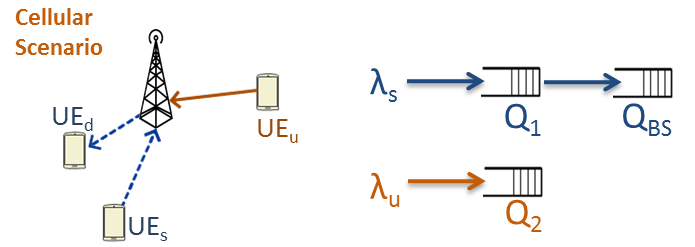}
\caption{System Model of the Three-UEs scenario}
\label{fig.SystemModel_TwoUE} 
\par\end{centering}
\end{figure}

UE$_{u}$ communicates with the BS through an uplink cellular communication and the queue of user UE$_{u}$ is represented
by Q$_{u}$. The traffic arriving to the queues of UE$_s$ and UE$_{u}$ is time varying, i.i.d. over time and with rate respectively equal to $\lambda_{s}$ and $\lambda_{u}$. The communication between UE$_{s}$ and UE$_{d}$ is represented by the cascade of the uplink queue Q$_{s}$ of UE$_{s}$ and the download queue Q$_{BS}$. The BS does not transmit to UE$_d$ unless it has received at least one packet from UE$_{s}$. This coupling between the queues induces that the service rates of Q$_{s}$ and Q$_{u}$ depend on the state (empty / not empty) of Q$_{BS}$ which makes the queuing stability analysis challenging. The traffic arriving to Q$_{BS}$ is nothing but the departure from Q$_{s}$. The traffic departure is also time varying and depends on the scheduling allocation decision and the time varying channel conditions. The departure average rates from queues Q$_s$ and Q$_u$ are respectively denoted by  $\mu_{s}$ and $\mu_{d}$. \\

For this simple scenario, a set of three bit rates is considered $\left[r_{1},r_2,r_{3}\right]$ with $r_1=k r_2$ and $r_3=0$ (with $k \in \mathbb{N}^+$). For this set of bit rate corresponds a set of SNR intervals $\left[S_{1},S_2,S_{3}\right]$. As we mentioned before in equation (\ref{pi_n}), we use the notation $p^n_{i}$ to describe the probability that the SNR of the link $i$ is within the $n^{th}$ SNR interval (i.e. in state $S_n$). We suppose a complex Gaussian distribution of the channel $h_i$; hence these probabilities can be easly derived from equation (\ref{pi_n_Gauss}).\\
For the downlink BS-UE$_d$, these probabilities are the following: 
\[
p^1_{d}=\exp\left(-\frac{\gamma_{1}^{DL}N_0}{d_{d}^{-\beta}P_{DL}}\right)
\]
\[
p^2_{d}=\exp\left(-\frac{\gamma_{2}^{DL}N_0}{d_{d}^{-\beta}P_{DL}}\right)-\exp\left(-\frac{\gamma_{1}^{DL}N_0}{d_{d}^{-\beta}P_{DL}}\right)
\] 
\[
p^3_{d}=1-p^1_{d}-p^2_{d}
\]
For both uplinks UE$_s$-BS and UE$_u$-BS, these probabilities are the following:
\[
p^1_{i}=\exp\left(-\frac{\gamma_{1}^{UL}N_0}{d_{i}^{-\beta}P_{UL}}\right)
\]
\[
p^2_{i}=\exp\left(-\frac{\gamma_{2}^{UL}N_0}{d_{i}^{-\beta}P_{UL}}\right)-\exp\left(-\frac{\gamma_{1}^{UL}N_0}{d_{i}^{-\beta}P_{UL}}\right)
\] 
\[
p^3_{i}=1-p^1_{i}-p^2_{i}
\]
with $i=s$ for the  UE$_s$-BS link and  $i=u$ for the  UE$_u$-BS link.\\
\iffalse
For the uplink UE$_u$-BS, these probabilities are the following:
\[
p^1_{u}=\exp\left(-\frac{\gamma_{1}^{UL}N_0}{d_{u}^{-\beta}P_{UL}}\right)
\]
\[
p^2_{u}=\exp\left(-\frac{\gamma_{2}^{UL}N_0}{d_{u}^{-\beta}P_{UL}}\right)-\exp\left(-\frac{\gamma_{1}^{UL}N_0}{d_{u}^{-\beta}P_{UL}}\right)
\] 
\[
p^3_{u}=1-p^1_{u}-p^2_{u}
\]
\fi

In the aforementioned expressions,  $d_{s}$ is the distance between $UE_{s}$ and the BS,  $d_{u}$ is the distance between $UE_{u}$ and the BS, $d_{d}$ is the distance between $UE_{d}$ and the BS. In the sequel, we refer to the three-UEs scenarios with 3 bit rates ($r_1$, $r_2$ and $r_3$) by the \textit{simple scenario} SS .

\subsection{Organization}
For the 3-UEs scenario study, the organization is as follows: section \ref{subsec:CellSimpleSc} provides a theoretical analysis of the exact stability region (SR) of the system. We note that deriving this SR is computationally complex. Therefore, we propose in section \ref{subsec:CellSimpleSc_Approx} an explicit and simple expression of an upper bound (UB) for this SR. Hence, in section \ref{subsec:SimpleSc_Comp_Approx_Real} we compare the SR to its UB by expressing the maximum difference between both of them. Thus, this UB turns to be a simple and very close approximation of the SR.
\iffalse
Note that the three-UEs scenario was briefly treated in our paper \cite{OnlineArticle} (we will put our conference paper once accepted) but in the aim of being self-contained this paper will include the details of these results.
\fi
%%%%%%%%%%%%%%%%%%%%%%%%%%%%%%%%%%%%%%
%%%%%%%%%%%%%%%%%%%%%%%%%%%%%%%%%%%%%%
\vspace{-10pt}
\subsection{Exact stability region for the 3-UEs scenario\label{subsec:CellSimpleSc}}
\iffalse
We assume a scheduling such that only one communication is possible in each timeslot. Depending on the priority policy, one of these two communications is scheduled: UE$_u$-BS or UE$_{s}$-UE$_{d}$. Due to the fact that Q$_{BS}$ might be empty at some timeslot, when UE$_s$-UE$_d$ communication is scheduled then the choice between uplink (UE$_s$ to BS) or downlink (BS to UE$_d$) depends not only on the SNR states of these two links but also on the state (empty or not) of the BS. In order to take that into account, we introduce in the analysis a new parameter $\alpha=\left(\alpha_1, \alpha_2, \alpha_3, \alpha_4\right)$. The UE$_s$ to UE$_d$ communication, modeled as a chain of queues $Q_s$ - $Q_{BS}$, might be in 4 different SNR states: $\left( S_{1},S_{1}\right),\left( S_{1},S_{2}\right),\left( S_{2},S_{1}\right),\left( S_{2},S_{2}\right)$ where the SNR state $ S_i$ is defined in section \ref{sec:System-Model}. For each state $i$, we consider $\alpha_i$ and $1-\alpha_i$ (with $0\leq\alpha_{i}\leq1$) the fraction of time that the resources are respectively allocated to the uplink (UE$_s$ to BS) and the downlink (BS to UE$_d$). \\
\fi

In this section, we characterize the exact stability region for the 3-UEs cellular scenario. Assuming a user scheduling such that only one communication is possible in each timeslot then for a given timeslot, only one of the following communications is scheduled: UE$_u$-BS or UE$_{s}$-UE$_{d}$. Recall that $\Gamma$ is the priority policy according to which the links are sorted. We know that in order to characterize the stability region it is sufficient to consider its corner points that correspond to the "extreme policies" where the priority is always given to the same communication or when the priority is always given to the communication that has the better channel state. In this paper we denote by $\Omega_{\Gamma}^{ss}$ the set of the priority policies corresponding to the corner points of the stability region. Here, we only indicate that this set contains 6 policies $\lbrace\Gamma_1,\Gamma_2,\Gamma_3,\Gamma_4,\Gamma_5, \Gamma_6 \rbrace$ and that finding the values for this set of policies leads to the characterization of the stability region. We exhaustively describe these policies in the appendix \ref{stab_real_twoRate}. However, the problem is not simply solved by finding the subset $\Omega_{\Gamma}^{ss}$. Actually, for each priority policy $\Gamma \in \Omega_{\Gamma}^{ss}$, the challenge remains in capturing the coupling between the queues due to the relaying functionality of the BS between the UL and DL traffic of the UE2UE communications.

Indeed, the priority policies' approach transforms the multidimensional Markov Chain, that models the interacting queues, to a one dimensional (1D) Markov chain for each priority policy. In order to characterize the stability region, all the priority policies should be considered or at least the priority policies that achieve the corner points of the stability region. Therefore, the queues can be modeled by a 1D Markov Chain per priority policy. Nonetheless, the proposed modeling (1D Markov Chain), for each priority policy, remains challenging since the service rates of the queues still depend on the state of the BS queue due to the coupling between the UL and DL traffic of the \textit{UE2UE communications}. Thus, the following additional analysis is required for capturing the dependency between the queues.

For a given priority policy $\Gamma$, the service rate of the queues depend not only on the SNR states of the links but also on the state (empty or not) of the BS. For this reason, we consider a queuing theory approach in order to capture this coupling between the service rate of the queues (Q$_s$ and Q$_u$) and the state of Q$_{BS}$. In section \ref{sec:Numerical}, we expose examples that show the impact of queuing approach on the performance analysis of the network.

Q$_{BS}$ might be empty at some timeslot, therefore when UE$_s$-UE$_d$ communication is scheduled then the choice between UL  (UE$_s$-BS) or DL (BS-UE$_d$) depends not only on the SNR states of these two links but also on the state (empty or not) of Q$_{BS}$. In order to take that into account, we introduce a new parameter $\alpha$ called fraction vector. It is used to divide the time where the UE$_1$-to-UE$_3$ communication is scheduled between its UL and DL such that the stability of the queue Q$_{BS}$ remains satisfied.

The UE$_s$ to UE$_d$ communication, modeled as a chain of queues $Q_s$ - $Q_{BS}$, might be in 4 different SNR states: $\left( S_{1},S_{1}\right),\left( S_{1},S_{2}\right),\left( S_{2},S_{1}\right),\left( S_{2},S_{2}\right)$ where the SNR state $ S_i$ is defined in section \ref{sec:System-Model}. For each couple of SNR states for the links (UE$_s$-BS, BS-UE$_d$), we define a parameter $\alpha_i$ (with $0\leq\alpha_{i}\leq1$) that corresponds to the fraction of time that the resources are allocated to the UL (UE$_s$ to BS) whereas $1-\alpha_i$ corresponds to the fraction of time that the resources are allocated to the DL (BS to UE$_d$). It gives that $\alpha=\left(\alpha_1, \alpha_2,\alpha_3,\alpha_4\right)$. The fraction vector $\alpha$ is considered only for these 4 couples of SNR states where both UL and DL of the UE$_s$-to-UE$_d$ communication have a SNR state different than $S_3$ thus a non zero rate. Only for these cases, the concept of fraction parameter makes a sense for dividing resources between both UL and DL that are able to transmit. However, for the other combinations of couple SNR states (where at least one SNR state is equal to $S_3$), this concept does not hold because only the link with a positive rate is able to transmit.

For different priority policies of UE$_u$-BS and UE$_{s}$-UE$_{d}$ communications, different fraction of resources will be allocated to the UL and DL of the chain  $Q_s$ - $Q_{BS}$ which corresponds to different values of $\alpha \in \left[0,1\right]^4$. For each priority policy, we find the optimal fraction vector  $\alpha^*=\left( \alpha^*_{1},\alpha^*_{2},\alpha^*_{3},\alpha^*_{4}\right)$ that achieves one corner point of the stability region. Finding $\alpha^*$ allows us to avoid the need to vary $\alpha \in \left[0,1\right]^4$ for each priority policy in order to obtain the corresponding corner point. To make expression simpler, we use the following notation for a given priority policy $\Gamma$:
\begin{itemize}
\item  $U\left(\Gamma\right)$ and $V\left(\Gamma\right)$ as the probabilities that UE$_s$ transmits respectively at rate $r_1$ and $r_2$ when $Q_{BS}$ is empty.\\♠
\item $N\left(\Gamma\right)=p_{s}^1U\left(\Gamma\right)+\bar{p}_s^{1}V\left(\Gamma\right)$ and $M\left(\Gamma\right)=\left(k p_d^{1}+p_d^{2}N\left(\Gamma\right)\right)-\left(k p_s^{1}U\left(\Gamma\right)+p_s^{2}V\left(\Gamma\right)\right)p_d^{3}$.\\
\item  $W\left(\Gamma\right)$ and $X\left(\Gamma\right)$ as the probabilities that UE$_u$ transmits respectively at rate $r_1$ and $r_2$ when $Q_{BS}$ is empty.\\
\item $Y\left(\Gamma\right)$ and $Z\left(\Gamma\right)$ the probabilities that UE$_u$ transmits respectively at rate $r_1$ and $r_2$ when $Q_{BS}$ is not empty.
\end{itemize}

\begin{theorem}
\label{th_ss_cellular_real}
The stability region for the 3-UEs cellular scenario is given by $ \left(\lambda_s,\lambda_u\right)  \in \mathcal{R}_c^{ss}$ s.t.:
\[
\mathcal{R}_{c}^{ss}=co\left(\bigcup\limits_{\Gamma \in \Omega^{ss}_{\Gamma}}\bigcup\limits_{\alpha \in \left[0,1 \right]^4}\left \lbrace  \mu_s\left(\alpha,\Gamma\right) , \mu_u\left(\alpha,\Gamma\right)  \right\rbrace\right)
\]
\iffalse
\[
\mathcal{R}_{c}^{ss}=co \lbrace \bigcup\limits_{\Gamma \in \Omega^{ss}_{\Gamma}}  \mu_s\left(\Gamma\right) , \mu_u\left(\Gamma\right)  \rbrace
\]
\fi
where  $\Omega^{ss}_{\Gamma}$ is the set of the priority policies for the simple scenario that achieve the corner points of the stability region (with $|\Omega^{ss}_{\Gamma}|=6$). The queues' service rates $\mu_s \left(\alpha,\Gamma\right) $ and $\mu_u\left(\alpha,\Gamma\right) $ are respectively given by (\ref{eqn_mu1c_real}) and (\ref{eqn_mu2c_real}) for all the priority policies $\Gamma \in \Omega^{ss}_{\Gamma}$.\\

\[\mu_{s}:=\mu_{s}\left(\alpha,\Gamma \right)=\left[r_1p_s^{1}U\left(\Gamma\right) + r_2p_s^{2}V \left(\Gamma\right) \right]\times\varPi_{BS}^0\left(\Gamma\right)\]
\begin{equation}
\label{eqn_mu1c_real}
\resizebox{0.9\hsize}{!}{$
 + \left[r_1p_s^1U\left(\Gamma\right)\left[ \alpha_{1}p_{d}^1+\alpha_{2}p_{d}^2 +p_{d}^3 \right] + r_2p_s^{2}\left[ \alpha_{3}p_d^{1}U\left(\Gamma\right) +\alpha_{4}p_d^{2}V\left(\Gamma\right) +p_d^{3}V\left(\Gamma\right)\right) \right]\times \left[1-\varPi^0_{BS}\left(\Gamma\right)\right]$}
\end{equation}
\vspace{-10pt}
\begin{equation}
\label{eqn_mu2c_real}
\resizebox{0.9\hsize}{!}{$\mu_{u}:=\mu_{u}\left(\alpha,\Gamma \right)=\left[r_1p_u^{1}W\left(\Gamma\right)  + r_2p_u^{2}X \left(\Gamma\right) \right]\times \varPi_{BS}^0\left(\Gamma\right)+ \left[r_1p_u^{1}Y\left(\Gamma\right)  + r_2p_u^{2}Z\left(\Gamma\right) \right]\times \left[1-\varPi_{BS}^0\left(\Gamma\right)\right]$}
\end{equation}
With $\varPi_{BS}^0\left(\Gamma\right)$ the probability that the queue Q$_{BS}$ is empty under the priority policy $\Gamma$ and the six values of the parameters $U,\,V,\,W,\,X,\,Y,\,Z$  corresponding to the six border priority policies $\Gamma \in \Omega^{ss}_\Gamma$ are given in table \ref{table.Policy-Cell}. 
\end{theorem}

\begin{proof}
See Appendix-\ref{stab_real_twoRate}.
\end{proof}

\begin{table}[H]
\centering
  \begin{tabular}{|c|c|c|c|c|c|c|c|c|c|}
\hline 
Corner Point $\Gamma$ & $U$ & $V$ & $W$ & $X$ & $Y$ & $Z$\tabularnewline
\hline
\hline 
$1$ & $1$ & $1$ & $p{}_s^{3}$ & $p{}_s^{3}$ & $p{}_s^{3}p{}_d^{3}$ & $p{}_s^{3}p{}_d^{3}$\tabularnewline
\hline 
$2$ & $p{}_u^{3}$ & $p{}_u^{3}$ & $1$ & $1$ & $1$ & $1$\tabularnewline
\hline 
$3$ & $1$ & $p{}_u^{3}$ & $\bar{p}_s^{1}$ & $\bar{p}_s^{1}$ & $\bar{p}_s^{1}\bar{p}_d^{1}$ & $\bar{p}_s^{1}\bar{p}_d^{1}$\tabularnewline
\hline 
$4$ & $\bar{p}_u^{1}$ & $\bar{p}_u^{1}$ & $1$ & $p_s^{3}$ & $1$ & $p{}_s^{3}p{}_d^{3}$  \tabularnewline
\hline 
$5$ & $1$ & $\bar{p}_u^{1}$ & $\bar{p}_s^{1}$ & $p_s^{3}$ & $\bar{p}_d^{1} \bar{p}_d^{1}$ & $p{}_s^{3}p{}_d^{3}$ \tabularnewline
\hline 
$6$ & $\bar{p}_u^{1}$ & $p{}_u^{3}$ & $1$ & $\bar{p}_s^{1}$ & $1$ & $\bar{p}_s^{1}\bar{p}_d^{1}$\tabularnewline
\hline 
\end{tabular}\caption{Service rate parameters}
\label{table.Policy-Cell}
\end{table}

We can resume the procedure of finding the stability region by the algorithm below. The complexity of the following computation as well as the non existence of an explicit form of the exact stability region comes from the consideration of a set of three bit rates which makes complicated the Markov Chain model of the queue Q$_{BS}$. Thus, finding the probability that this queue is empty $\varPi_0$ should be deduced from the solution of a system of linear equation. Due to that, we do not have an explicit form of the service rates of the queues. Hence, we are not able to find the optimum fraction vector $\alpha^*$ for each priority policy that achieves the corner point corresponding to this policy and avoids us to consider all the values $\alpha \in \left[ 0,1\right]^4$.

\begin{algorithm*}
\caption{Procedure of the exact stability region}\label{Proc_Real_MC}
\begin{algorithmic}[1]
\For {all $\Gamma \in \Omega_{\Gamma}^{ss}$}
\For {all $\alpha \in \left[0,1 \right]^4 $}
\State Find the probability $\varPi_0$ by solving the system of linear equations
\State Deduce the service rates of the queues
 \EndFor
 \EndFor
\end{algorithmic}
\end{algorithm*}

The computational complexity of the above form of the stability region comes from two factors: (i) the need of varying the fraction vector $\alpha$ within $\left[0,1 \right]^4$ in order to find the stability region and (ii) the calculation of the probability $\varPi_{BS}^0$ that the queue Q$_{BS}$ is empty for each value of $\alpha$. In fact, considering a set of three bit rates makes the Markov Chain modeling of the BS queue Q$_{BS}$ complicated as shown in the proof of the above theorem. It yields to a complicated computation of the  probability $\varPi_{BS}^0$ which is deduced by solving a linear system of equation for each value of the fraction vector $\alpha$. On the aim of reducing the complexity and having an explicit form, we propose a simple analytic expression of an approximation of this stability region.

\subsection{Approximated stability region \label{subsec:CellSimpleSc_Approx}}
As we mentioned before, finding the exact stability region needs a lot of computation (for finding $\Pi^0_{BS}$) and an exhaustive variation of $\alpha \in \left[ 0,1 \right]^4$. These are the main two causes of the computational complexity of the exact stability region. Actually, the exact stability region has not an explicit analytic form due to the fact that it is deduced from the numerical solution of a system of equations (for finding $\Pi^0_{BS}$) for each value of $\alpha \in \left[ 0,1 \right]^4$. Therefore, it is important to propose an epsilon-close upper bound ($\epsilon$-UB) that has an explicit form and that is simple to compute. An epsilon-close upper bound presents an upper bound of the exact stability region that has epsilon as the maximum distance between the exact region and its upper bound.

Therefore, we start by simplifying the Markov Chain modeling of the queue Q$_{BS}$ in order to derive an approximation of the stability region that has a close-form analytic expression and that is simple to compute.  After computing the small gap between the exact and approximated stability region, the latter one turns to be a close upper bound of the exact stability region.

Considering three different bit rates generates a complicated Markov Chain for modeling the queue Q$_{BS}$ (Upper side of Fig. \ref{fig.QBS_real_vs_approx_gen}). As one can see in the proof of the theorem \ref{th_ss_cellular_real}, finding the probability that Q$_{BS}$ has no packet to transmit ($\Pi^0_{BS}$ which is mandatory for finding the stability region) is challenging for this Markov Chain model and has not an explicit form. The reader is invited to refer to the proof of theorem \ref{th_ss_cellular_real} for more details on this Markov Chain. 

Therefore, we propose a simple approximated Markov Chain for modeling the queue Q$_{BS}$ (Lower side of Fig. \ref{fig.QBS_real_vs_approx_gen}). This approximated model is a simple birth and death Markov Chain where passing from state $x_i$ to $x_{i+1}$ (receiving a packet) corresponds to the average probability of receiving a packet at both rates $r_1$ and $r_2$. Moreover, passing from state $x_i$ to $x_{i-1}$ (transmitting a packet) corresponds to the average probability of transmitting a packet at both rates $r_1$ and $r_2$. We recall that the multiple rate model is still taken into consideration at the transition probabilities level of the approximated Markov Chain. This approximated model has mainly two advantages: (i) an explicit form of the probability of being at state 0 and (ii) an explicit form of the fraction optimal fraction vector $\alpha$ that achieves the corner point of the stability region. Hence, compared to the exact stability region, this approximation permits to avoid the two main reasons of complexity as it follows: computation of $\Pi^0_{BS}$ by applying an explicit formula and not by solving a system of equations and the lack of the need of varying $\alpha$ within $\left[ 0,1 \right]^4$ due to the explicit expression of the optimal fraction vector that achieves the corner point of this approximated stability region. The details of the explicit expression of the approximated model are given in the proofs of lemma \ref{th2_gen} and theorem \ref{th3_gen}.

\begin{figure}[H]
\centering
\captionsetup{justification=centering}
\includegraphics[width=1\textwidth]{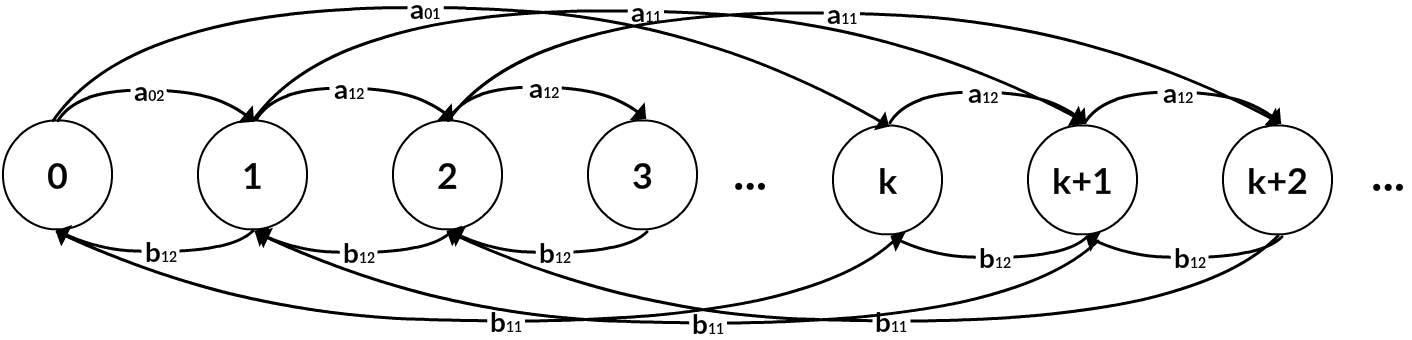}
\includegraphics[width=0.7\textwidth]{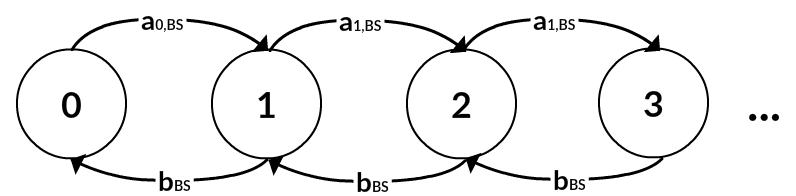}
\par
\caption{The upper Markov Chain is the exact modeling of the queue Q$_{BS}$ and the lower Markov Chain is the approximated modeling of the queue Q$_{BS}$}
\label{fig.QBS_real_vs_approx_gen}
\end{figure}

In the sequel we characterize the approximated stability region by proceeding in two steps: (i) we derive in lemma (\ref{th2_gen}) the condition that $\alpha$ should satisfy for each priority policy then (ii) we provide in theorem (\ref{th3_gen}) the corner points that correspond to the priority policies $\Gamma \in \Omega_{\Gamma}^{ss}$ and that describe this approximated stability region.\\

We distinguish the results of the approximation from that of the exact stability region by adding the following indication $\tilde{}$ over all the used notation. Recall that the set of bit rates is given by $\left[r_1=kr_2,r_2,r_3=0\right]$.\\

\iffalse
\textbf{MTOVIATION FOR APPPROXIMATION (same mean arrival rate and mean departure rate in each state of the markov chain) }

The optimum vector $\tilde{\alpha}^{*}$d $\alpha^*=\left( \alpha^*_{1},\alpha^*_{2},\alpha^*_{3},\alpha^*_{4}\right)$ corresponds to the value, within the set (\ref{Set_alpha_gen}) below, that achieves the boundary of the stability region is achieved. In other terms, finding $\alpha^*$ allows us to avoid the need to vary $\alpha \in \left[0,1\right]^4$ in order to obtain the stability region.\\
The optimum $\tilde{\alpha}^{*}$ is computed by finding $\alpha \in \left[0,1\right]^4$ that expands as much as possible the stability region which means that maximizes $\tilde{\mu}_1$ and $\tilde{\mu}_2$ under the constraint (\ref{eq:QBSstabConstraint_gen}). We prove in Appendix-\ref{B_gen} and we verify it numerically in section \ref{sec:Numerical} that $\alpha^{*}$ corresponds to the vector from the finite set below that expands the most the stability region:

\fi

\begin{lemma}
\label{th2_gen}
For a given priority policy $\Gamma$, the fraction vector $\left( \tilde{\alpha}_{1}, \tilde{\alpha}_{2}, \tilde{\alpha}_{3}, \tilde{\alpha}_{4} \right)$ should verify:
\begin{equation}
\label{eq:QBSstabConstraint_gen}
2k\tilde{\alpha}_{1}p_s^{1}p_d^{1}U+\left(k+1\right)\tilde{\alpha}_{2}p_s^{1}p_d^{2}U+\left(k+1\right)\tilde{\alpha}_{3}p_s^{2}p_d^{1}U+2\tilde{\alpha}_{4}p_s^{2}p_d^{2}V
\leq k p_d^{1}+p_d^{2}N-\left(k p_s^{1}U+p_s^{2}V\right)p_d^{3}
\end{equation}

with $U$, $V$, $W$, $X$, $Y$, $Z$ depend on the priority policy $\Gamma$ and their values for the border priority policies $\Gamma \in \Omega^{ss}_\Gamma$ are given in table \ref{table.Policy-Cell}. Recall that $N=p_s^{1}U+\bar{p}_s^{1}V$.
\end{lemma}

\begin{proof}
See Appendix-\ref{A_gen}.
\end{proof}
\newpage
\begin{theorem}
\label{th3_gen}
The approximated stability region for the 3-UEs cellular scenario is the set of $ \left(\lambda_s,\lambda_u\right)  \in \tilde{\mathcal{R}}_c^{ss}$ such that:

\[
\tilde{\mathcal{R}}_{c}^{ss}=co\left(\bigcup\limits_{\Gamma \in \Omega^{ss}_{\Gamma}}\bigcup\limits_{\tilde{\alpha} \in \mathbb{S}_{\alpha} }\lbrace  \tilde{\mu}_s\left(\tilde{\alpha},\Gamma\right) , \tilde{\mu}_u\left(\tilde{\alpha},\Gamma\right)  \rbrace \right)
\]

where  the queues' service rates $\tilde{\mu}_s \left(\Gamma\right) $  and $\tilde{\mu}_u\left(\Gamma\right) $ are respectively given by (\ref{eqn_mu1c_gen}) and (\ref{eqn_mu2c_gen}), $ \mathbb{S}_{\alpha}$ is a limited subset that is simply computed by (\ref{Set_alpha_gen}). $\Gamma \in \Omega_{\Gamma}^{ss}$ where $\Omega_{\Gamma}^{ss}$ is the set of priority policies that achieve the border of the stability region (with $|\Omega_{\Gamma}^{ss}|=6$). The values of the parameters  $U\left(\Gamma\right)$, $V\left(\Gamma\right)$, $W\left(\Gamma\right)$, $X\left(\Gamma\right)$, $Y\left(\Gamma\right)$ and $Z\left(\Gamma\right)$ for the six priority policies $\Gamma \in \Omega_{\Gamma}^{ss}$ are given in table \ref{table.Policy-Cell}. Recall that $A=\left(k p_d^{1}+p_d^{2}N\right)-\left(k p_s^{1}U+p_s^{2}V\right)p_d^{3}$.

\[\tilde{\mu}_s:=\tilde{\mu}_s\left(\Gamma\right)=\frac{1}{2}\left(r_{1}p_s^{1}U+r_{2}p_s^{2}V\right)\times\]
\begin{equation}
\label{eqn_mu1c_gen}
\resizebox{0.9\hsize}{!}{$\left(1+\frac{\left(1-k\right)\tilde{\alpha}_{2}^{*}p_s^{1}p_d^{2}U+\left(k-1\right)\tilde{\alpha}_{3}^{*}p_s^{2}p_d^{1}U-\left(k p_d^{1}+p_d^{2}N\right)+\left(k p_s^{1}U+p_s^{2}V\right)\bar{p}_d^{3}}{2k\tilde{\alpha}_{1}^{*}p_s^{1}p_d^{1}W+\left(k+1\right)\tilde{\alpha}_{2}^{*}p_s^{1}p_d^{2}U+\left(k+1\right)\tilde{\alpha}_{3}^{*}p_s^{2}p_d^{1}U+2\tilde{\alpha}_4^{*}p_s^{2}p_d^{2}U-\left(k p_d^{1}+p_d^{2}N\right)-\left(k p_s^{1}U+p_s^{2}V\right)\bar{p}_d^{3}}\right)$}
\end{equation}
\[\tilde{\mu}_d:=\tilde{\mu}_d\left(\Gamma\right)=r_{1}p_u^{1}W+r_{2}p_u^{2}X\]
\begin{equation}
\label{eqn_mu2c_gen}
\resizebox{0.9\hsize}{!}{$+\frac{\left(r_{1}p_u^{1}\left(W-Y\right)+r_{2}p_u^{2}\left(X-Z\right)\right)\left(k p_s^{1}U+p_s^{2}V\right)}{2k\tilde{\alpha}_{1}^{*}p_s^{1}p_d^{1}U+\left(k+1\right)\tilde{\alpha}_{2}^{*}p_s^{1}p_d^{2}U+\left(k+1\right)\tilde{\alpha}_{3}^{*}p_s^{2}p_d^{1}U+2\tilde{\alpha}_{4}^{*}p_s^{2}p_d^{2}V-\left(k p_d^{1}+p_d^{2}N\right)-\left(k p_s^{1}U+p_s^{2}V\right)\bar{p}_d^{3}}$}
\end{equation}

\begin{equation}
\begin{split}
\resizebox{1\hsize}{!}{$\mathbb{S}_{\alpha}=
\left(\begin{array}{ccc}
{\tilde{\alpha}_2,\tilde{\alpha}_3}\epsilon\left\{ 0,1\right\}^{2}  & s.t. & \left(k+1\right)\left({\tilde{\alpha}}_{2}p_s^1p_d^2+{\tilde{\alpha}}_{3}p_s^{2}p_d^{1}\right)U\leq M\,\,, \tilde{\alpha}_1=\tilde{\alpha}_4=0\,\,\,\,\,\,\,\,\,\,\,\,\,\,\,\,\,\,\,\,\,\,\,\,\,\,\,\,\,\,\,\,\,\,\,\,\,\,\,\,\,\,\,\,\,\,\,\,\,\,\,\,\,\,\,\,\,\,\,\,\,\,\,\,\,\,\,\,\,\,\,\,\,\,\,\,\,\,\,\,\,\,\,\,\,\,\,\,\,\,\,\,\,\,\,\,\,\,\,\,\,\,\,\,\,\,\,\\
& &\\
{\tilde{\alpha}}_{2},{\tilde{\alpha}}_{3}\epsilon\left\{ 0,1\right\} ^{2} & \,s.t.  & {\tilde{\alpha}}_{1},{\tilde{\alpha}}_{4}\epsilon\left[0,1\right]^{2} \text{ verify } 2k{\tilde{\alpha}}_{1}p_s^{1}p_d^{1}U+2{\tilde{\alpha}}_{4}p_s^{2}p_d^{2}V=M-\left(k+1\right)\left({\tilde{\alpha}}_{2}p_s^{1}p_d^{2}U+{\tilde{\alpha}}_{3}p_s^{2}p_d^{1}U\right) \,\,\,\,\,\,\,\,\,\,\,\,\,\,\,\,\, \\
& &\\
{\tilde{\alpha}}_{2}\epsilon\left\{ 0,1\right\}  & s.t. & {\tilde{\alpha}}_{3}=\left[{M-\left(k+1\right){\tilde{\alpha}}_{2}p_s^{1}p_d^{2}U}\right]/\left[{\left(k+1\right)p_s^{2}p_d^{1}U}\right], {\tilde{\alpha}}_{1}={\tilde{\alpha}}_{4}=0\,\,\,\,\,\,\,\,\,\,\,\,\,\,\,\,\,\,\,\,\,\,\,\,\,\,\,\,\,\,\,\,\,\,\,\,\,\,\,\,\,\,\,\,\,\,\,\,\,\,\,\,\,\,\,\,\,\,\,\,\,\,\,\,\,\,\,\,\,\,\,\,\,\,\,\\
& &\\
{\tilde{\alpha}}_{2}\epsilon\left\{ 0,1\right\}  & s.t. & {\tilde{\alpha}}_{3}=\left[{M-2k p_s^{1}p_d^{1}U-2p_s^{2}p_d^{2}V-\left(k+1\right){\tilde{\alpha}}_{2}p_s{1}p_d^{2}U}\right]/\left[{\left(k+1\right)p_s^{2}p_d^{1}U}\right],  {\tilde{\alpha}}_{1}={\tilde{\alpha}}_{4}=1\,\,\,\,\,\,\,\,\,\,\,\,\,\,\,\,\,\\
& &\\
{\tilde{\alpha}}_{3}\epsilon\left\{ 0,1\right\}  & s.t. & {\tilde{\alpha}}_{2}=\left[{M-\left(k+1\right){\tilde{\alpha}}_{3}p_s^{2}p_d^{1}U}\right]/\left[{\left(k+1\right)p_s^{1}p_d^{2}U}\right], {\tilde{\alpha}}_{1}={\tilde{\alpha}}_{4}=0\,\,\,\,\,\,\,\,\,\,\,\,\,\,\,\,\,\,\,\,\,\,\,\,\,\,\,\,\,\,\,\,\,\,\,\,\,\,\,\,\,\,\,\,\,\,\,\,\,\,\,\,\,\,\,\,\,\,\,\,\,\,\,\,\,\,\,\,\,\,\,\,\,\,\\
& &\\
{\tilde{\alpha}}_{3}\epsilon\left\{ 0,1\right\}  & s.t. & {\tilde{\alpha}}_{2}=\left[{M-2k p_s^{1}p_d^{1}U-2p_s^{2}p_d^{2}V-\left(k+1\right){\tilde{\alpha}}_{3}p_s^{2}p_d^{1}U}\right]/\left[{\left(k+1\right)p_s^{1}p_d^{2}U}\right], {\tilde{\alpha}}_{1}={\tilde{\alpha}}_{4}=1\,\,\,\,\,\,\,\,\,\,\,\,\,\,\,\,\,\,\,\\
\end{array}\right)$}\\
\end{split}
\label{Set_alpha_gen}
\end{equation} 
\hrulefill
%\end{figure*}
\end{theorem}

\begin{proof}
See Appendix-\ref{B_gen}.
\end{proof}

%%%%%%%%%%%%%%%%%%%%%%%%%%%%%%%%%%%%%%%%%%%%%%%%%%%%%%%%%%%%%%%%%%%%%%%%%%%%%%%%%%%%%%

\subsection{Comparison real and approximation \label{subsec:SimpleSc_Comp_Approx_Real}}
The exact stability region has not an explicit form because of the complicated Markov Chain of the queue Q$_{BS}$ that has not an explicit expression of $\varpi^0_{Bs}$ (but a solution of a system of equations). Moreover, this complexity is accentuated by the need of varying $\alpha$ within all the interval $\left[0,1 \right]^4$ in order to elaborate the exact stability region. However, the approximated stability region has an explicit expression of  both parameters: (i) probability that the BS queue is empty and (ii) the optimal fraction vector that achieves the corner point of this region. Hence, the importance of this approximated model lies on the proposition of a simple analytic form of the stability region. Furthermore, we verify that the approximated stability region is an epsilon-close upper bound of the exact stability region. To do so, we demonstrate analytically that the relative error between these two regions is positive and bounded.\\

\begin{theorem}
\label{th_comparison_gen}
For the 3-Users cellular scenario, the approximated stability region $\tilde{\mathcal{R}}_{c}^{ss}$ is a close upper bound of the exact stability region ${\mathcal{R}}_{c}^{ss}$ with a maximum relative error $\epsilon^*_R$. Therefore, ${\mathcal{R}}_{c}^{ss}$ is bounded as it follows:
\[
 \left(1-\epsilon_R^*\right)\tilde{\mathcal{R}}_c^{ss} \subseteq  {\mathcal{R}}_c^{ss} \subseteq \tilde{\mathcal{R}}_c^{ss}
\]
with 
\[
\epsilon^*_R= \max \limits_{\Gamma \in  \Omega^{ss}_{\Gamma}} \frac{kb_{11}\left(\tilde{\alpha}^*\right)\sum\limits_{i=1}^{k-1}\varPi_i\left(\tilde{\alpha}^*\right)\left(a_{02}+ka_{01}-a_{12}\left(\tilde{\alpha}^*\right)-ka_{11}\left(\tilde{\alpha}^*\right)\right)}{\left(a_{02}+ka_{01} \right)\left(b_{12}\left(\tilde{\alpha}^*\right)+kb_{11}\left(\tilde{\alpha}^*\right)\right)}
\]
\iffalse
\begin{equation}
\delta_R= \max \limits_{}\left|\frac{\mu -\tilde{\mu}}{\mu}\right|\leq \epsilon_R^* = \max \limits_{\Gamma \in  \Omega^{ss}_{\Gamma}}\epsilon^*\left(\Gamma\right)
\label{eq.delta_rel_gen}
\end{equation}
\frac{\epsilon^*\left(\Gamma\right)}{1+\epsilon^*\left(\Gamma\right)}

with 
\[
\epsilon^*\left(\Gamma\right)=\frac{kb_{11}\left(\tilde{\alpha}^*\right)\sum \limits_{i=1}^{k-1}\varPi_i\left(\tilde{\alpha}^*\right)\left(a_{02}+ka_{01}-a_{12}\left(\tilde{\alpha}^*\right)-ka_{11}\left(\tilde{\alpha}^*\right)\right)}{\left(a_{02}+ka_{01} \right)\left(b_{12}\left(\tilde{\alpha}^*\right)+kb_{11}\left(\tilde{\alpha}^*\right)\right)}
\]
\fi
and the following parameters depend on the priority policy $\Gamma$:
\begin{itemize}
\item $\tilde{\alpha}^*$ the optimum fraction vector computed analytically in $\mathbb{S}_\alpha$ (given by (\ref{Set_alpha_gen})) and that achieves the corner point of the approximated stability region.
\item $a_{01}$ and $a_{02}$ are the arrival probabilities at Q$_{BS}$ when this queue is empty with a respective rates $r_1$ and $r_2$  (see equations (\ref{QBS_a01_gen}) and (\ref{QBS_a02_gen})).
\item $a_{11}\left(\tilde{\alpha}^*\right)$ and $a_{12}\left(\tilde{\alpha}^*\right)$ are the arrival probabilities at Q$_{BS}$ when this queue is not empty with a respective rates $r_1$ and $r_2$ (see equations (\ref{QBS_a11_gen}) and (\ref{QBS_a12_gen})).
\item $b_{11}\left(\tilde{\alpha}^*\right)$ and $b_{12}\left(\tilde{\alpha}^*\right)$ are the departure probabilities from Q$_{BS}$ at rate $r_1$ and $r_2$ respectively. They depend on the fraction vector $\alpha$ (see equations (\ref{QBS_b1_gen}) and (\ref{QBS_b2_gen})).
\end{itemize}
\end{theorem}

\begin{proof}
See Appendix-\ref{th_comp_gen}.
\end{proof}

Hence, for each priority policy we find $\epsilon\left(\Gamma\right)$ which correspond to the deviation between both the exact corresponding corner point and its approximation. We verified that $\epsilon\left(\Gamma\right) \geq 0$ for all  $\Gamma \in  \Omega^{ss}_{\Gamma}$, hence $\tilde{\mathcal{R}}_c^{ss}$ presents an upper bound of the exact stability region $\mathcal{R}_c^{ss}$. As one can see in the numerical section, the epsilon difference between the exact stability region and its upper bound is small. Hence, we highly reduce the complexity by finding an explicit and close upper bound of the exact stability region. \\
\section{Multi-UE scenario \label{sec:MUE_scenario}}

For the general scenario, we consider $K$ \textit{UE2UE communications} between pair of UEs (UE$_{i,s}$ and UE$_{i,d}$) and $U$ \textit{UE2BS communications} between BS and UE$_{i,u}$. That means that in total $2K+U$ users are considered in the cell. For the multi-user case, the study is applied for a set of two bit rates $[r_1,r_2]$ with $r_2=0$. Unfortunately, considering a set of 3 bit rates as in section \ref{sec:3UE_scenario} is complex and hard to compute and to simplify. Thus, we study the case of two bit rates. Even for this case, the characterization of the exact stability region, given by theorem \ref{th1_oneRate_MU}, remains computationally complex for the multi-user case. Therefore, we limit the complexity of the exact stability region in theorem \ref{th2_oneRate_MU}). However, the complexity remains important. Hence, we propose in theorem \ref{th3_oneRate_MU} a simple approximation of the exact stability region that is characterized by the following: (i) highly reducing the complexity of the exact stability region and (ii) being an $\epsilon$-close approximation of the exact stability region (it means that the maximum distance between the approximated and exact stability regions is equal to a small number $\epsilon$). The importance of this result lies on shifting a very complex problem to a simple approximated one with a low complexity and a high preciseness. A trade-off exists between the precision of the stability region and the complexity of characterizing this region. \\

In order to to simplify the presentation of the results in the sequel, we use the following notation of the transmission probabilities at rate $r_1$: $p_{i,s}=p^1_{i,s}$, $\bar{p}_{i,s}=1-p^1_{i,s}$, $p_{i,d}=p^1_{i,d}$, $\bar{p}_{i,d}=1-p^1_{i,d}$, $p_{i,u}=p^1_{i,u}$, $\bar{p}_{i,u}=1-p^1_{i,u}$, $p_{i,sd}=p^1_{i,sd}$ and $\bar{p}_{i,sd}=1-p^1_{i,sd}$.\\

We denote by $\Gamma$ the priority policy under which the users are sorted. The scheduling policy consists on choosing a UE if and only if all the more prioritized UEs are not able to transmit. We denote by $\Omega_\Gamma$ the set of all the possible sorting of the users. the number of existing priority policies consists of the number of the possible permutation of $K+U$ communications. Hence, $|\Omega_\Gamma|=\left(K+U\right)!$. We define the following two sets for each communication $i$ for all $1 \leq i \leq K+U$: 
\begin{itemize}
\item $\mathbf{U}_{\Gamma}\lbrace i \rbrace = \lbrace$ \textit{UE2BS communications} more prioritized than the communication i under a priority sorting $\Gamma \rbrace$
\item $\mathbf{K}_{\Gamma}\lbrace i \rbrace = \lbrace$ \textit{UE2UE communications} more prioritized than the communication i under a priority sorting $\Gamma \rbrace$
\end{itemize}

\begin{figure}[ptb]
\begin{centering}
\includegraphics[scale=0.75]{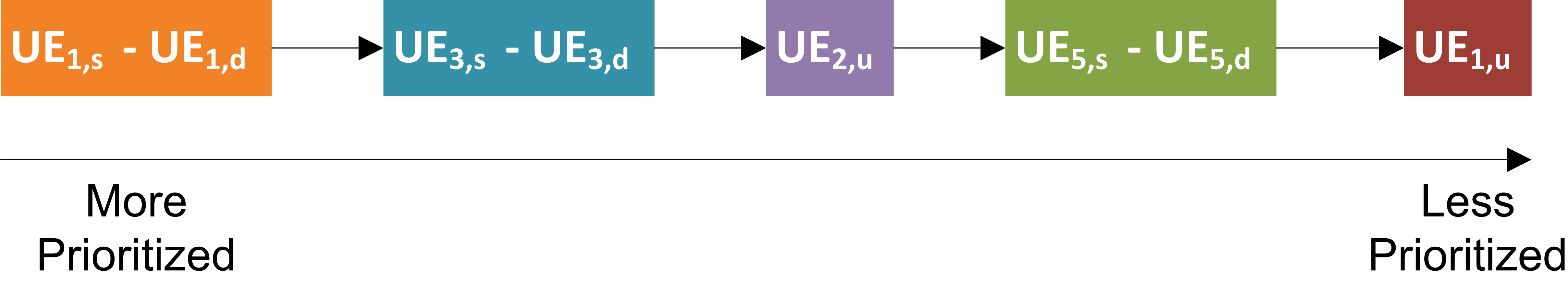}
\par\end{centering}
\caption{Example of one priority policy for $K=3$ and $U=2$. }
\label{fig.UE_sorting} 
%\par\end{centering}
\end{figure}

In Fig. \ref{fig.UE_sorting}, we present an example that illustrate one priority policy for a scenario of $K=3$ \textit{UE2UE communications} and $U=2$ \textit{UE2BS communications}. In this example, the considered priority policies gives the highest priority to the first \textit{UE2UE communication} ($UE_{1,s}-UE_{1,d}$)  and the lowest priority to the first \textit{UE2BS communication} ($UE_{1,u}-BS$). In this example, the sets $\mathbf{U}_{\Gamma}\lbrace 4 \rbrace$ and $\mathbf{K}_{\Gamma}\lbrace 4 \rbrace$ of the $UE_{1,u}$ are the following: $\mathbf{U}_{\Gamma}\lbrace 4 \rbrace=\lbrace 2 \rbrace$ and $\mathbf{K}_{\Gamma}\lbrace 4 \rbrace=\lbrace 1,3,5 \rbrace$. The sets $\mathbf{U}_{\Gamma}\lbrace 5 \rbrace$ and $\mathbf{K}_{\Gamma}\lbrace 5 \rbrace$ of the $UE_{5,s}-UE_{5,d}$ communication are the following: $\mathbf{U}_{\Gamma}\lbrace 2 \rbrace=\lbrace 2 \rbrace$ and $\mathbf{K}_{\Gamma}\lbrace 5 \rbrace=\lbrace 1,3 \rbrace$.

The vector that describe the arrival rates of the users' queues is given by:
\[
\bm{\lambda}=\left[ \lambda_{1,s} \,\,\, \lambda_{2,s}\,\,\,...\,\,\, \lambda_{K,s} \,\,\, \lambda_{1,u}\,\,\, \lambda_{2,u}\,\,\,  ... \,\,\, \lambda_{U,u} \right]
\]

 Here, we derive the stability region of the multi-UE cellular scenario where the devices that want to communicate with each other must exchange their packets through the BS.
 
We consider the following notation of the service rate:
\begin{itemize}
\item $\mu_{i,s}\left( \Gamma\right)$: service rate of the uplink of the $i^{th}$ \textit{UE2UE communication}
\item $\mu_{j,u}\left( \Gamma\right)$: service rate of the downlink of the $j^{th}$ \textit{UE2BS communication}\\
\end{itemize}

The traffic departure is time varying and depends on the scheduling allocation decision and the time varying channel conditions. The departure average rates from all the users queues for the cellular scenario is denoted by:
\[
\bm{\mu}\left(\Gamma\right)=\left[ \mu_{1,s}\left(\Gamma\right) \,\,\, \mu_{2,s}\left(\Gamma\right) \,\,\,...\,\,\, \mu_{K,s}\left(\Gamma\right) \,\,\, \mu_{1,u}\left(\Gamma\right) \,\,\, \mu_{2,u}\left(\Gamma\right)\,\,\,  ... \,\,\, \mu_{U,u}\left(\Gamma\right) \right]
\]

We assume a user scheduling such that only one communication is possible in each timeslot. Q$_{i,BS}$ might be empty at some timeslot, therefore when the UE$_{i,s}$-UE$_{i,d}$ communication is scheduled then the choice between uplink (UE$_{i,s}$ to BS) or downlink (BS to UE$_{i,d}$) depends not only on the SNR states of these two links but also on the state (empty or not) of the corresponding queue at BS Q$_{i,BS}$. In order to take that into account, we introduce in the analysis a new parameter $\alpha_i \in \left[0,1\right]$ for each \textit{UE2UE communication} ($1\leq i\leq K$) that describes the  fraction of time that the resources are respectively allocated to the uplink (UE$_{i,s}$ to BS). In other terms, $1-\alpha_i$ is the fraction of time that the resources are allocated to the downlink (BS to UE$_{i,d}$).  The fraction vector is mainly introduced for guaranteeing the stability region of the queues at the BS level. We denote by $\bm{\alpha}$ the fraction vector that considers all the \textit{UE2UE communications}.
\[
\bm{\alpha}=\left[\alpha_1, \alpha_2 , ..., \alpha_K \right]
\]
\subsection{Organization}
In the sequel, we start by characterizing the exact stability region of the multi-UE scenario in section \ref{ssec:MUE_exact}. This region is given by considering all the feasible priority policies $\Gamma \in \Omega_\Gamma$  and all the values of $\bm{\alpha}\in \left[0,1\right]^K$ are considered. Since computing this region is highly complex. Thus in section \ref{ssec:MUE_approx}, we start by reducing the complexity by limiting the number of fraction vector $\bm{\alpha}$ that should be considered for characterizing the stability region.  Furthermore, an epsilon-close approximation with a simple explicit form and a very low complexity is proposed. A trade-off is elaborated between the complexity and the precision of the stability region computation. 

\subsection{Exact stability region\label{ssec:MUE_exact}}
The theorem below gives the exact stability region of the multi-UE scenario. The complexity of the exact stability region is discussed directly below the theorem \ref{th1_oneRate_MU}.
\begin{theorem}
\label{th1_oneRate_MU}
For the multi-UE cellular scenario, the exact stability region is given by the set of $ \bm{\lambda} \in \mathcal{R}_c$ such that:
\[
\mathcal{R}_{c}=co\left(\bigcup\limits_{\Gamma \in \Omega_{\Gamma}}\bigcup\limits_{\alpha \in \left[\bm{0},\alpha^* \right]}\lbrace \bm{\mu}\left(\bm{\alpha},\Gamma \right) \rbrace\right)
\]
where $\bm{0}$ is the vector zero in $\mathbb{R}^{K}$, $\alpha^*=\left(\alpha_1^*, \alpha_2^*, ... ,\alpha_K^*\right)$ with $\forall 1\leq i \leq K: \alpha_i^*=\frac{p_{i,d}-p_{i,s}+p_{i,s}p_{i,d}}{2p_{i,s}p_{i,d}}$. $\Omega_{\Gamma}$ is the set of all the possible priority policies (with $|\Omega_\Gamma|=\left(K+U\right)!$). The elements of $\bm{\mu} \left( \Gamma \right)$ which are  $\mu_{i,s}\left( \Gamma\right)$ and $\mu_{j,u}\left( \Gamma\right)$ (for $1 \leq i\leq K$, $1 \leq j\leq U$) are respectively given by (\ref{eqn_mu1c_oneRate_MU}) and (\ref{eqn_mu2c_oneRate_MU}).

\[
\mu_{i,s}\left(\Gamma\right)=\frac{1}{2}r_1p_{i,s}\left(1+\frac{\bar{p}_{i,s}p_{i,d}}{-2\alpha_i p_{i,s}p_{i,d}+\left(1+p_{i,s}\right)p_{i,d}}\right)\prod_{m \in \mathbf{U}_{\Gamma}\lbrace i \rbrace}\bar{p}_{m,u}
\]
\begin{equation}
\label{eqn_mu1c_oneRate_MU}
\times \prod_{n \in \mathbf{K}_{\Gamma}\lbrace i \rbrace}\bar{p}_{n,s}\left[1+ \frac{{p}_{n,s}p_{n,d}}{2\alpha_n p_{n,s}p_{n,d}-\left(1+p_{n,s}\right)p_{n,d}} \right]
\end{equation}
\begin{equation}
\label{eqn_mu2c_oneRate_MU}
\mu_{j,u}\left(\Gamma\right)=r_1p_{j,u}\prod_{m \in \mathbf{U}_{\Gamma}\lbrace j+K \rbrace}\bar{p}_{m,u}\prod_{n \in \mathbf{K}_{\Gamma}\lbrace j+K\rbrace}\bar{p}_{n,s}\left[1+ \frac{{p}_{n,s}p_{n,d}}{2\alpha_n p_{n,s}p_{n,d}-\left(1+p_{n,s}\right)p_{n,d}} \right]
\end{equation}
\end{theorem}

\begin{proof}
See Appendix-\ref{proof_oneRate_MU}.\\
\end{proof}
The exact stability region is computed by following the instructions below:
\begin{itemize}
\item For all $\Gamma \in \Omega_{\Gamma}$ (with $|\Omega_\Gamma|=\left(K+U\right)!$)
\begin{itemize}
\item For all $\bm{\alpha} \in \left[0,\bm{\alpha}* \right] $
\begin{itemize}
\item Compute the service rates of the queues 
\item Deduce a point of the exact stability region
\end{itemize}
\end{itemize}
\item Exact stability region as the convex hull of all these points\\
\end{itemize}

Based on the steps above for the exact stability region computation, we deduced the high complexity of this approach. Indeed, the complexity of the exact stability region computation comes from two factors: the high number of priority policies (depending on the number of communications) as well as the large number of the fraction vector $\bm{\alpha}$ values that should be considered for each priority policy. Actually, the number of the considered  computing the exact stability region pushes us to consider all the the possible priority policies $\Gamma \in \Omega_{\Gamma}$. Having $K+U$ communications means that there exists  $\left(K+U\right)!$ possible policies to sort them. Further, for each priority $\Gamma$, we should vary $\bm{\alpha} \in \left[ \bm{0},\bm{\alpha}^*\right]$. Suppose that for each $\alpha_i$ (for $1\leq i \leq K$) we consider $L$ different values within $\left[0,\alpha_i^*\right]$. Thus, for each priority policy we consider $L^{K}$ values of the fraction vector $\bm{\alpha}$. It is clear that bigger is $L$ higher is the precision of the stability region (in numerical section $L\geq 10^3$). we deduce that even by considering only a set of two bit rates $\left[r_1,r_2=0 \right]$, the complexity of the stability region computation remains high: $L^{K}\left(K+U\right)!$. \\

Therefore, we start reducing the complexity of the exact stability region by verifying that only the border points of $\alpha_i \in \lbrace 0,\alpha_i^*\rbrace$ (for $1 \leq i \leq K$) are sufficient for characterizing the exact stability region. Furthermore, we propose an approximation of the exact stability region that highly decreases the complexity with a tight loss of precision. This study show a trade-off between the precision of the stability region computation and its complexity.

\subsection{Precision versus complexity \label{ssec:MUE_approx}}
The main challenge is to reduce the complexity of the stability region computation while guaranteeing its precision. As we have discussed before, the complexity of the exact stability region is mainly caused by the high number of the following two parameters that should considered for the characterization of this region: (i) priority policies and (ii) fraction vector $\bm{\alpha}$ values. These two factors of complexity are respectively studied in this section. In the first part, we reduce the complexity by limiting the values of the fraction vector $\bm{\alpha}$ that should be considered for each priority policy. We prove that a limited set of values of the fraction vector $\bm{\alpha}$ is sufficient for characterizing the exact stability region. However, the complexity remains high due to the large number of the priority policies that should be considered $\left(K+U\right)!$. Hence, in the second part we propose an epsilon-close approximation of the exact stability region of the symmetric case (see definition (\ref{def:symCase}). This approximation highly reduces the complexity by limiting the number of the considered priority policies while guaranteeing a high precision.

\begin{theorem}
\label{th2_oneRate_MU}
For the multi-UE cellular scenario, the exact stability region is the set of $\bm{\lambda} \in \mathcal{R}_c$ such that $\mathcal{R}_c$ can be simplified as it follows:
\[
\mathcal{R}_{c}=co\left(\bigcup\limits_{\Gamma \in \Omega_{\Gamma}}\bigcup\limits_{\alpha \in \mathbb{S}_{\alpha}}\lbrace \bm{\mu}\left(\bm{\alpha},\Gamma \right) \rbrace\right)
\]
with 
%{\color{red}
\[
\mathbb{S_{\alpha}}=\lbrace \bm{\alpha} \, | \, \alpha_i \in \lbrace 0, \alpha_i^*  \rbrace \,\,\,\forall\, 1 \leq i \leq K \rbrace
\]
%}
where the elements of $\bm{\mu} \left( \Gamma \right)$ which are  $\mu_{i,s}\left( \Gamma\right)$ and $\mu_{j,u}\left( \Gamma\right)$ (for $1 \leq i\leq K$, $1 \leq j\leq U$) are respectively given by (\ref{eqn_mu1c_oneRate_MU}) and (\ref{eqn_mu2c_oneRate_MU}). Where the limited set $\mathbb{S_{\alpha}}$ of  $2^K$ elements reduce the complexity.
\end{theorem}

\begin{proof}
See Appendix-\ref{proof_oneRate_MU_th2}.
\end{proof}

In theorem \ref{th2_oneRate_MU}, the complexity of the exact stability region is decreased due to the fact that for each priority policy $\Gamma \in \Omega_\Gamma$ , the border values of each $\alpha_i$ are sufficient for characterizing the corner point of the exact stability region that corresponds to $\Gamma$. Indeed, for a given priority policy $\Gamma \in \Omega_\Gamma$, the service rates vector $\bm{\mu\left(\bm{\alpha}, \Gamma\right)}$ for any $\bm{\alpha} \in \left[\bm{0},\bm{\alpha}^* \right]$ can be written as a convex combination of the service rates vectors $\bm{\mu\left(\bm{\alpha}, \Gamma\right)}$ with $\bm{\alpha} \in  \mathbb{S}_\alpha$. Actually, the $L$ values of each $\alpha_i$ that we should consider in theorem \ref{th1_oneRate_MU} for each priority policy $\Gamma$ is reduced to $2$ values in theorem \ref{th2_oneRate_MU} ($0$ and $\alpha_i^*$). The proof of the theorem \ref{th2_oneRate_MU} shows that varying $\bm{\alpha}$ within the finite set $\mathbb{S}_\alpha$ is sufficient for the characterization of the exact stability region.  Hence, the computational complexity is reduced from $L^{K}\left(K+U\right)!$ to $2^{K}\left(K+U\right)!$.\\

However, the complexity of the exact stability region remains high because of the large number of the priority policies that should be considered ($\left(K+U\right)!$). Hence, for the symmetric case (defined in (\ref{def:symCase})), we  reduce further the complexity by limiting the number of the considered priority policies for the characterization of the stability region. An $\epsilon$-approximation, such that the maximum distance between the approximated and exact stability region is limited to $\epsilon$, is proposed. The main importance of this approximation is that it highly reduces the complexity (in terms of number of the considered priority policies and number of the fraction vector values) while guaranteeing the precision of the result.

\begin{definition}
\label{def:symCase}
The symmetric case consists of considering all the UEs of the cell at the same distance d from the BS. Hence, this symmetric case is defined by the following values: $p_s=p_{i,s}=p_{j,u}$ and $p_d=p_{i,d}$ ($\forall  1 \leq i \leq K$ and $1\leq j \leq U$). 
\end{definition}

\begin{definition}
\label{def:epsilonApprox}
$\tilde{\mathcal{R}}$ is called an $\epsilon$-approximation of a stability region ${\mathcal{R}}$ (with $0\leq \epsilon \leq 1$) iff the following is verified:
\iffalse
\[
\left(1-\epsilon\right)\tilde{\mathcal{R}} \subseteq  {\mathcal{R}} \subseteq \tilde{\mathcal{R}}
\]
\fi
\[
\tilde{\mathcal{R}} \subseteq  {\mathcal{R}} \subseteq \tilde{\mathcal{R}}
+\epsilon
\]
\end{definition}

The symmetric case is defined as the case where all the UEs are at the same distance $d$ from the base station. For this case, we prove in theorem  \ref{th3_oneRate_MU} that we can additionally reduce  the computation complexity by limiting the number of policies $\Gamma$ that should be considered in order characterize the stability region. Thus, we avoid the computation complexity that comes from the consideration of all the policies $\Gamma \in \Omega_{\Gamma}$ (i.e. for a network of $K+U$ communications, it exists$|\Omega_\Gamma|=\left(K+U\right)!$ priority policies overall). Therefore, we propose an $\epsilon$-approximation of the exact stability region that highly reduces the complexity while guaranteeing a high precision. 
\begin{theorem}
\label{th3_oneRate_MU}
For the symmetric case  of the multi-UE cellular scenario with a given couple $\lbrace p_s,p_d \rbrace$, an $\epsilon$-approximation  of the stability region (for a given $\epsilon$) is the set of $ \bm{\lambda} \in\tilde{{\mathcal{R}}}_c$ such that:
\iffalse
\[
\tilde{\mathcal{R}}_{c}=co\left(\bigcup\limits_{\Gamma \in \Omega_{\Gamma}^{K_0}}\bigcup\limits_{\alpha \in \mathbb{S}_{\alpha}^{K_0}}\lbrace \bm{\mu}\left(\bm{\alpha},\Gamma \right) \rbrace
\right) \text{  with  } K_0\left(\epsilon\right)=\min\left\{K+U,\left\lceil 1+\frac{\log\left(\frac{\epsilon}{r_1p_s}\right)}{\log\left(\bar{p}_s\right)}\right\rceil\right\}
\]
\fi
\[
\tilde{\mathcal{R}}_{c}=co\left(\bigcup\limits_{\Gamma \in \Omega_{\Gamma}^{K_0}}\bigcup\limits_{\alpha \in \mathbb{S}_{\alpha}^{K_0}}\lbrace \bm{\mu}\left(\bm{\alpha},\Gamma \right) \rbrace
\right) \text{  with  } K_0\left(\epsilon\right)=\left\lceil 1+\frac{\log\left(\frac{\epsilon}{r_1p_s}\right)}{\log\left(\bar{p}_s\right)}\right\rceil
\]
Hence, the exact stability region can be bounded as it follows:
\[
\tilde{\mathcal{R}}_{c} \subseteq  {\mathcal{R}}_{c} \subseteq \tilde{\mathcal{R}}_{c}+\epsilon
\]
where $\Omega_{\Gamma}^{K_0}$ is the set of the feasible priority policies of the subsets of $K_0$ communications among all the $K+U$ communications (with $|\Omega_{\Gamma}^{K_0}|=\frac{(K+U)!}{\left(K+U-K_0\right)!}$ ), $\mathbb{S}_{\alpha}^{K_0}$ is the set of the values $\alpha_i \in \lbrace 0,\alpha_i^* \rbrace$ where $i$ corresponds to the \textit{UE2UE communications} within these subsets of $K_0$ elements. The elements of $\bm{\mu} \left( \Gamma \right)$ which are  $\mu_{i,s}\left( \Gamma\right)$ and $\mu_{j,u}\left( \Gamma\right)$ (for $1 \leq i\leq K$, $1 \leq j\leq U$) are respectively given by (\ref{eqn_mu1c_oneRate_MU}) and (\ref{eqn_mu2c_oneRate_MU}). Note that the value of ${K_0}$ is limited to $\left(K+U\right)$ which corresponds to the total number of communications.
\end{theorem}

\begin{proof}
See Appendix-\ref{proof_oneRate_MU_th3}.\\ 
\end{proof}

The proposed approximation in theorem \ref{th3_oneRate_MU} is defined by considering all the following cases: (i) the priority policies of all the subset of $K_0$ elements within the $K+U$ communications, (ii) for each priority policies, the \textit{UE2UE communications} within these $K_0$ elements have a fraction element that varies within its two border values ($0$ and $\alpha_i^*$ ). The value of $K_0$ depends on the precision $\epsilon$ of the approximation; it increases when $\epsilon$ decreases. Even for a tight value of $\epsilon$, $K_0$ stays too much lower than $K+U$ which reduces the computational complexity to:  $2^{K_0}\frac{(K+U)!}{\left(K+U-K_0\right)!}$. We can see that $K_0$ is constant for a given triplet $\lbrace p_s, p_d, \epsilon \rbrace$ and does not depend on the number of communications. Therefore, the importance of this approximation is to reduce the exponential complexity O$\left(2^{K} \right)$ for characterizing the exact stability region to a polynomial complexity O$\left(K^{K_0} \right)$.

\section{Numerical Results \label{sec:Numerical}}

Simulations are performed in order to validate our theoretical results and numerically validate the trade-off between the precision and the complexity of the stability region computation. We consider a cell of radius $R=500\,m$. We consider a $10\,$MHz LTE-like TDD system. When a user is scheduled, it transmits over $50$ subcarriers called Resource Blocks (RB). The considered bit rates per RB are $\left\{ r_{1}=400,\,r_{2}=200,\,r_{3}=0\right\} \,kbps/RB$. The unit of the shown results is the arrival rate per RB [kbps/RB] from which the total bit rate can be deduced by multiplying by the number of allocated RB. The total transmission power (i.e. over all RBs) of the BS is $P_{DL}=40\,W$ and that of the mobiles is $P_{UL}=0.25\,W$ (see \cite{3GPP_TR36814}). The noise power is equal to $-195\,dB/Hz$ for the downlink $DL$ and $-199\,dB/Hz$ for the uplink $UL$. \iffalse We consider a standard distance-based path loss model with path loss exponent equal to $\beta=3.76$\fi. We consider the pathloss model specified in \cite{3GPP_TR36814}. The SNR thresholds are respectively $\gamma_{1}^{UL}=9.5\,dB$, $\gamma_{2}^{UL}=2.5\,dB$, $\gamma_{1}^{DL}=7.5\,dB$ and $\gamma_{2}^{DL}=1.5\,dB$. These values are practical values based on a throughput-SNR mapping for a $10\,$MHz E-UTRA TDD network (see \cite{3GPP_TR36942}). 

\subsection{Three-UEs scenario}
We start by presenting the results of three-UEs network in order to validate the theoretical results corresponding to this scenario. For clarity of the presentation, we consider that the three UEs UE$_s$, UE$_u$ and UE$_d$ are at the same distance $d$ from the BS. This assumption is just for simplifying the illustration of the results however the theoretical results can be applied for any distribution of the users in the cell.

\subsubsection{Approximated stability region} 
We illustrate the approximated stability region of  the three-UEs scenario for different distances $d$ between the users and the BS. This study is done for two purposes: (i) validating that the computed set of $\alpha$ in equation (\ref{Set_alpha_gen}) as well as the considered set of priority policies $\Omega_\Gamma^{ss}$ are sufficient for characterizing the approximated stability region and no need to consider all the values $\alpha \in \left[0,1 \right]$ and all the priority policies $\Gamma \in \Omega_\Gamma$ and (ii) elaborating the effect of the distance $d$ on the performance of the users.

In Fig. \ref{fig.NumStability}, for different distances $d$ between UEs and BS $d=\left\{100, 350,\,500\right\}$, we plot in Red the  stability region for the approximated three-UEs scenario obtained by exhaustive search (all $\Gamma \in \Omega_{\Gamma}$ and all $\alpha \in \left[0,1\right]^4$) and we compare it to the one plotted in Green and obtained from Theorem \ref{th3_gen}. We find that both curves coincides which verifies that the latter theorem presents the stability region for the approximated three-UEs scenario. The importance of this theorem is the fact of providing a simple and close form expression of this stability region which is a close approximation of the computational complex exact stability region (as proved in theorem \ref{th_comparison_gen}and numerically showed in the next subsection \ref{ssNum.ApproVsEx}). Fig. \ref{fig.NumStability} gives also insights on the effect of the distance $d$ on the three-UEs stability region and shows how the performance of the system degrades when the distance $d$ increases. \\

\begin{figure}[ptb]
\begin{centering}
\includegraphics[scale=0.15]{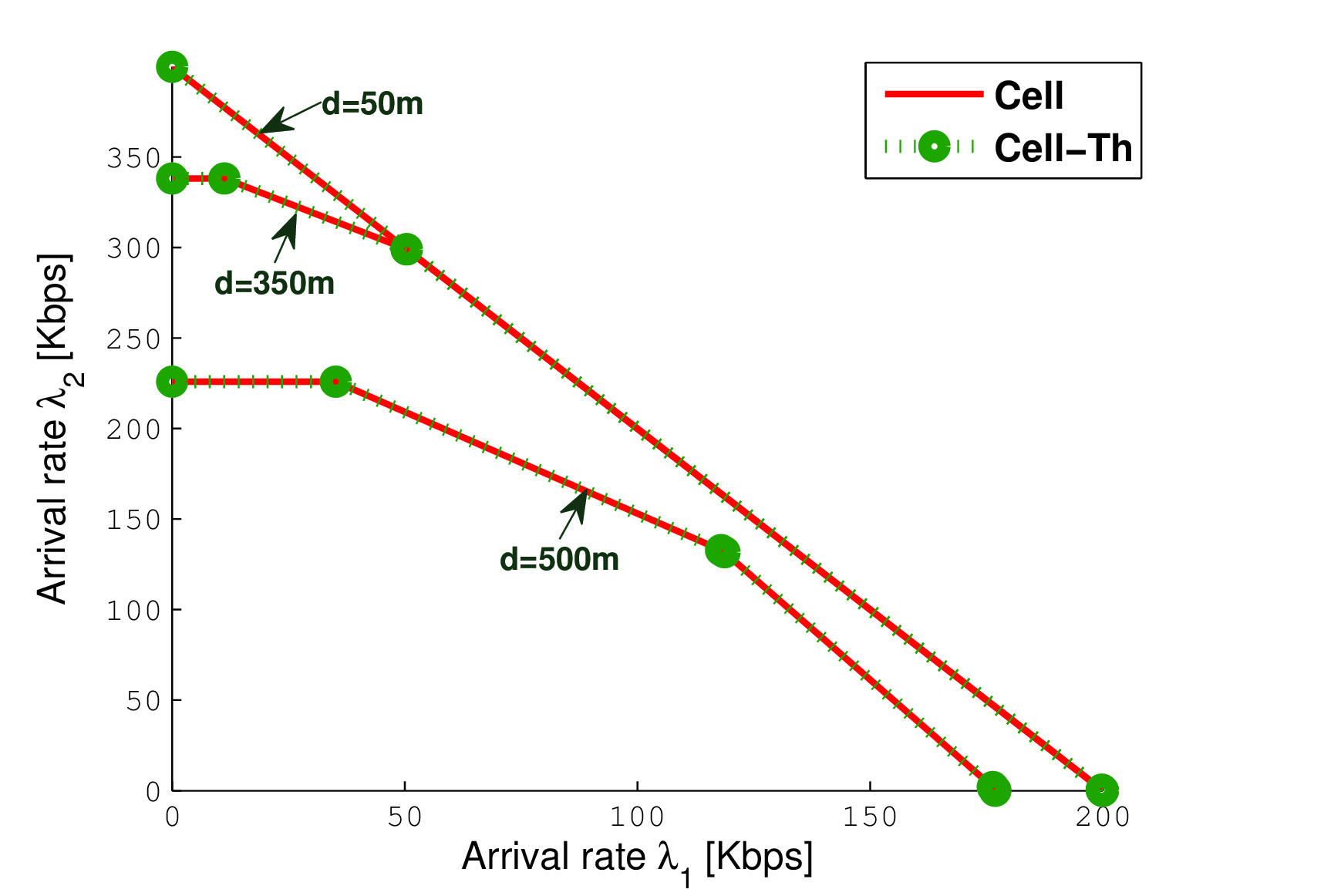}
\par\end{centering}

\caption{Stability regions $\tilde{\mathbb{R}}_c^{ss}$ for the approximated model of the three-UEs scenario where the distance UE to BS distance varies within $d\,\epsilon\,\left\{ 100,\,350,\,500\right\}$m}.
\label{fig.NumStability} 
%\par\end{centering}
\end{figure}

\subsubsection{Exact vs. approximated stability region \label{ssNum.ApproVsEx}} In this section we numerically validate that $\tilde{\mathcal{R}}_c^{ss}$ (given by (\ref{th3_gen})) is a close approximation of the stability region for the three-UEs scenario ${\mathcal{R}}_c^{ss}$ (given by (\ref{th_ss_cellular_real}). We present in Fig.\ref{fig.ApproxReal} that the stability region and its approximation coincide for different distances $d$ between UEs and BS. Simulations considering different users positions verify that the approximated stability region is a good approximation for the exact stability for the three-UEs scenario. \\

\begin{figure}[ptb]
\vspace{-10pt}
\begin{centering}
\includegraphics[scale=0.15]{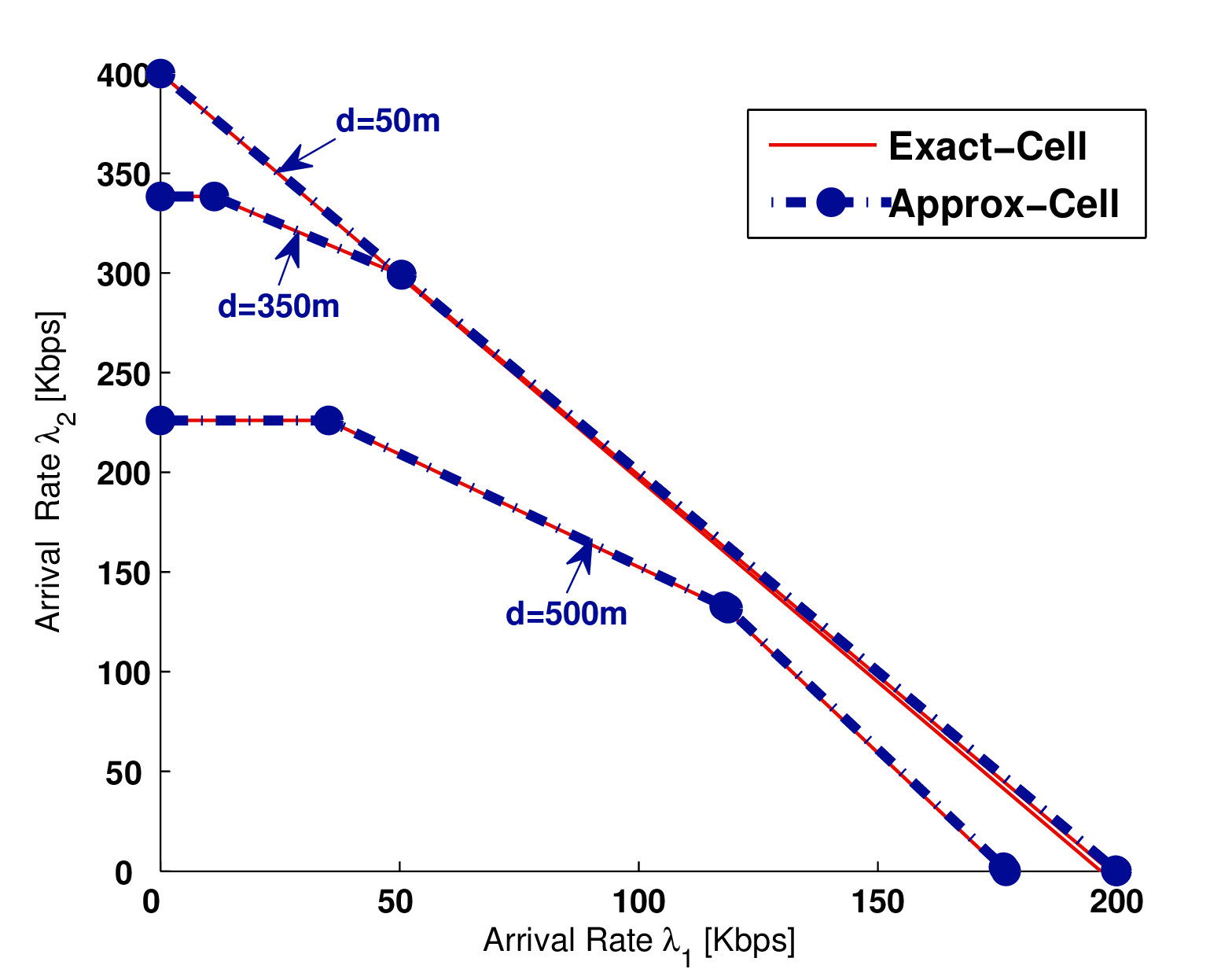}
\caption{Comparison between the stability region for the cellular scenario ${\mathcal{R}}_c^{ss}$ and its approximation $\tilde{\mathcal{R}}_c^{ss}$ for three different UE to BS distance $d\,\epsilon\,\left\{ 100,\,350,\,500\right\}$m.}
\label{fig.ApproxReal}
\par\end{centering}
\vspace{-10pt}
\end{figure}

\iffalse 
For small distances $d$ (i.e. $50$m), we remark that D2D mode is more advantageous than cellular mode for all considered D2D-distances (between $d_{min}$ and $2d$). This is due to the fact that for small d, the distance between D2D peer is also small  in such a way that the performance of direct communication (D2D) remains good.

As we increase the distance $d$, the quality of the cellular communication decreases; in parallel, the range of the studied D2D-distances increases in such a way that the quality of D2D degrades with the increase of $c$. Figure  \ref{fig.NumStability} shows that for each distance $d$, there exists a different  $c*$ under which D2D is more advantageous than cellular communications. For example, $c=1,25$ for $d = 100m$ whereas $c=0,42$ for $d= 350m$. 
\fi
\subsubsection{Queuing impact} We start by showing the advantages of the queuing theory approach in our analysis. For this purpose, we consider a random positioning of 3 UEs in the cell and we apply the two performance evaluation approaches on the cellular scenario: (i) taking into account the queuing aspects and the coupling between the queues (ii) without considering any coupling between the queues. For the first approach, we compute the stability region as the set of the arrival rate vectors to the sources that can be stably supported by the network considering all the possible policies, hence we consider the coupling between the queues (w/ coupling). For the second approach, we assume that the BS queues have a full buffer and we compute the rate region that describes the achievable data rates depending on the the channel states of the links and without any bursty traffic neither coupling between the queues (w/o coupling).  Comparing these two results in Fig. \ref{fig.Causality} verifies that introducing the traffic pattern and the queues' coupling have an effect on the performance of the cellular scenario in terms of stability region.

Considering that the queues have a full buffer, the DL (BS to UE$_3$) has always packets to transmit and it will always be considered by the scheduler. This causes less scheduling for the UL thus a lower bit rate over the UL. However, when coupling between the queues is considered, the DL can be scheduled only if at least one packet is received by the BS from UE$_1$. Due to that, when no packet are buffered at the DL, this latter will not be scheduled and the UL will be able to transmit more packets and will improve its bit rate. This explains the gain that the queuing approach offers to the performance evaluation of the scenario.\\
\iffalse
compute the stability region for the cellular scenario: (i) without any coupling between the queues thus the stability region is only analyzed at the physical layer and depend only on the channel states of the links (w/o coupling) (ii) with consideration of the coupling between the queues hence the service rate of queues Q$_1$ and Q$_2$ depend also on the state (empty or not) of Q$_{BS}$ (w/ coupling). Comparing these two results in Fig. \ref{fig.Causality} verifies that introducing this coupling has an effect on the performance of the cellular scenario in terms of stability region. 
\fi

\begin{figure}[H]
\vspace{-20pt}
\begin{centering}
\includegraphics[scale=0.1]{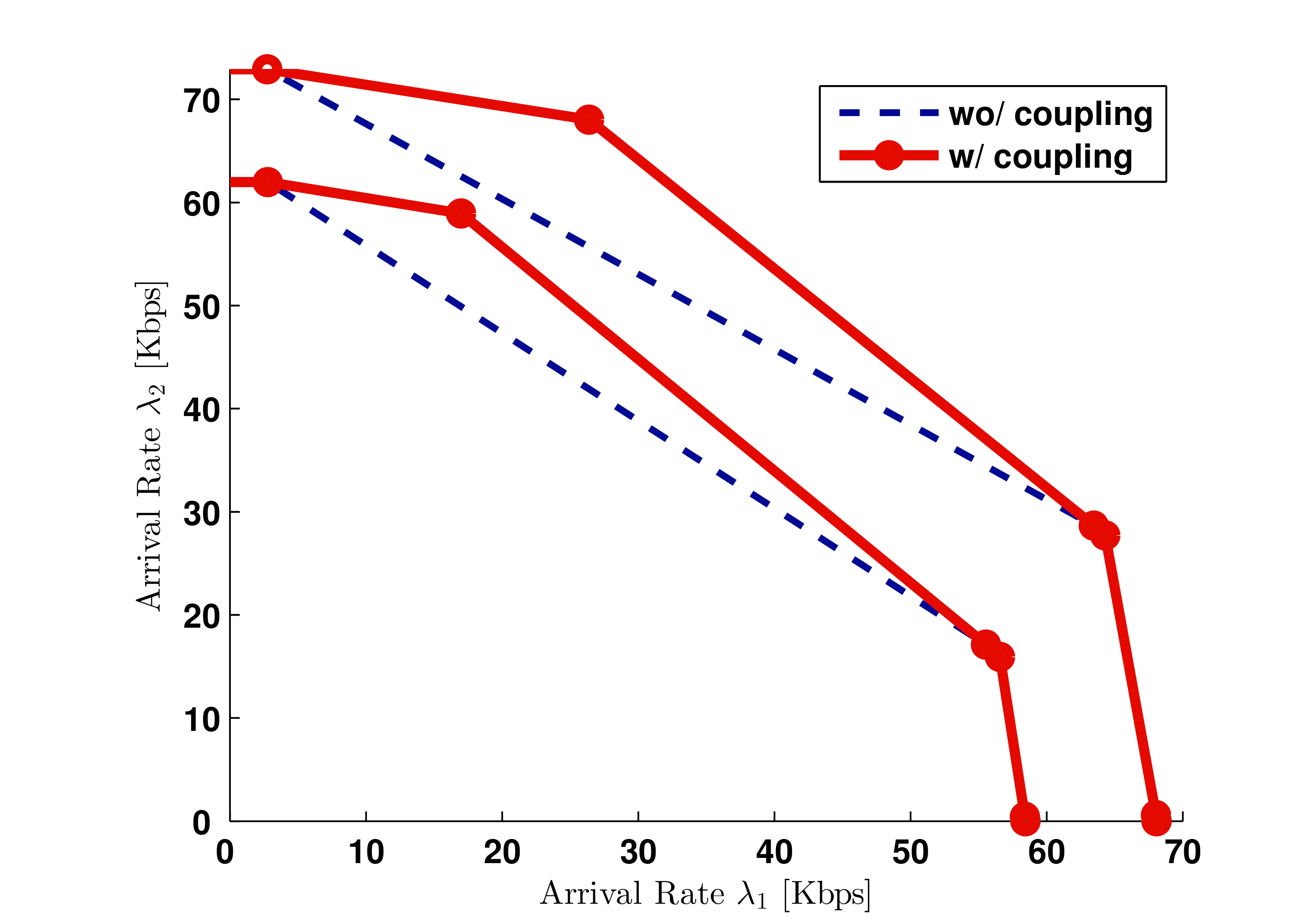}
\par\end{centering}
\caption{Stability region of the cellular scenario (with coupling between the users' queues and Q$_{BS}$ is not saturated) vs. rate region (without coupling hence Q$_{BS}$ is  never empty)}
\label{fig.Causality} 
%\par\end{centering}
%\vspace{-20pt}
\end{figure}

\iffalse
\textbf{Causality }
We can see the effect of the causality (considering probability of being empty or not) for the large distances and for large rates $r_1=1000$ and $r_2=500$ for ex or $r_1=800$ and $r_2=400$ 

\begin{enumerate}
\item causality of cell scenario
\item cellular scenario verify: real , approx, finding alpha
\item compare real and approx: exhaustive research + formula
\item compare D2D and cell for some distances and 4 values of c
\item compare D2D and cell for all distance 50--500m
\item identify for different set on the map where D2D is more advantageous than cellular
\end{enumerate}
{\color{red}
\textbf{\large{Choose metrics $r_1$ and $r_2$}}

\textbf{\large{1. Stability region  of 3-UEs scenario}}

\textbf{\large{1.a Stability region cellular and D2D of 3-UEs scenario}}
Verify approximation and alpha algo and causality

\textbf{\large{1.b Stability region cellular and D2D of 3-UEs scenario}}
Verify approximation and causality

\textbf{\large{2. Comparison results of 3-UEs scenario}}

2.a  D2D dist threshold

2.b set on the map where D2D is better}
\fi

%%%%%%%%%%%%%%%%%%%%%%%%%%%%%%%%%%%%%%%%
%%%%%%%%%%%%%%%%%%%%%%%%%%%%%%%%%%%%%%%%
%%%%%%%%%%%%%%%%%%%%%%%%%%%%%%%%%%%%%%%
%%%%%%%%%%%%%%%%%%%%%%%%%%%%%%%%%%%%%%%
\subsection{Multi-UE scenario}
For the multi-UE scenario, without loss of generality and for clarity reasons, we illustrate the results for the symmetric case where all the users are at the same distance $d$ from the BS. \\

\subsubsection{Exact vs approximated stability region} 
For the multi-UE scenario, we consider a cell containing 50 \textit{UE2UE communications} with a uniform random drop in a cell of radius $R=500 m$ such that all the UEs are at distance $d=350m$. The performance metric that we use to compare between the stability regions of the exact symmetric case (expressed in theorem \ref{th2_oneRate_MU}) and the approximated symmetric case  (expressed in theorem \ref{th3_oneRate_MU}) is the average service rate per user. It is equal to the sum of the service rates of all the users divided by the number of users. For $d=350$m and $\epsilon=0.1$ we find $K_0=3$. Fig. \ref{fig.AvgMu_WWO_D2D} illustrates that the variation of this performance metric as function of the number $K$ of \textit{UE2UE communications} in the network. It shows that the $\epsilon$-approximated stability region is a very tight approximation of the stability region. We deduce that the complexity is reduced from $2^{50}\times 50! \simeq 3\times 10^{79} $ for the exact stability region (from theorem \ref{th2_oneRate_MU}) to $2^{3} \times\frac{50!}{47!}=9 \times 10^5$ for the approximated stability region (from theorem \ref{th3_oneRate_MU}). This illustrates how the exponential complexity O$\left(2^{50} \right)$ is reduced to a polynomial one O$\left({50}^3 \right)$.\\

\subsubsection{Trade-off complexity vs precision} 
For the symmetric case, where all the users are considered at the same distance $d$ from the BS, we show in theorem \ref{th3_oneRate_MU} that the value of $K_0$ characterizes the number of priority policies that is sufficient to be studied in order to describe the stability region, hence the complexity of the approximated stability region. For three different $\epsilon\in \lbrace 10^{-3},10^{-2},10^{-1}\rbrace$, we show in Fig.\ref{fig.K0_func_d} the variation of $K_0$ as function of the distance $d$ between the users and the BS. We can see that for $\epsilon=10^{-2}$, the maximum value of $K_0$ (achieved at the edge of the cell) is small and equal to $12$. Thus, the complexity of the approximated stability region can be deduced by applying the following:
\[
2^{K_0}\frac{\left(K+U\right)!}{\left(K+U-K_0\right)!}
\]
\begin{figure}[H]
\vspace{-40pt}
\centering
\includegraphics[scale=0.085]{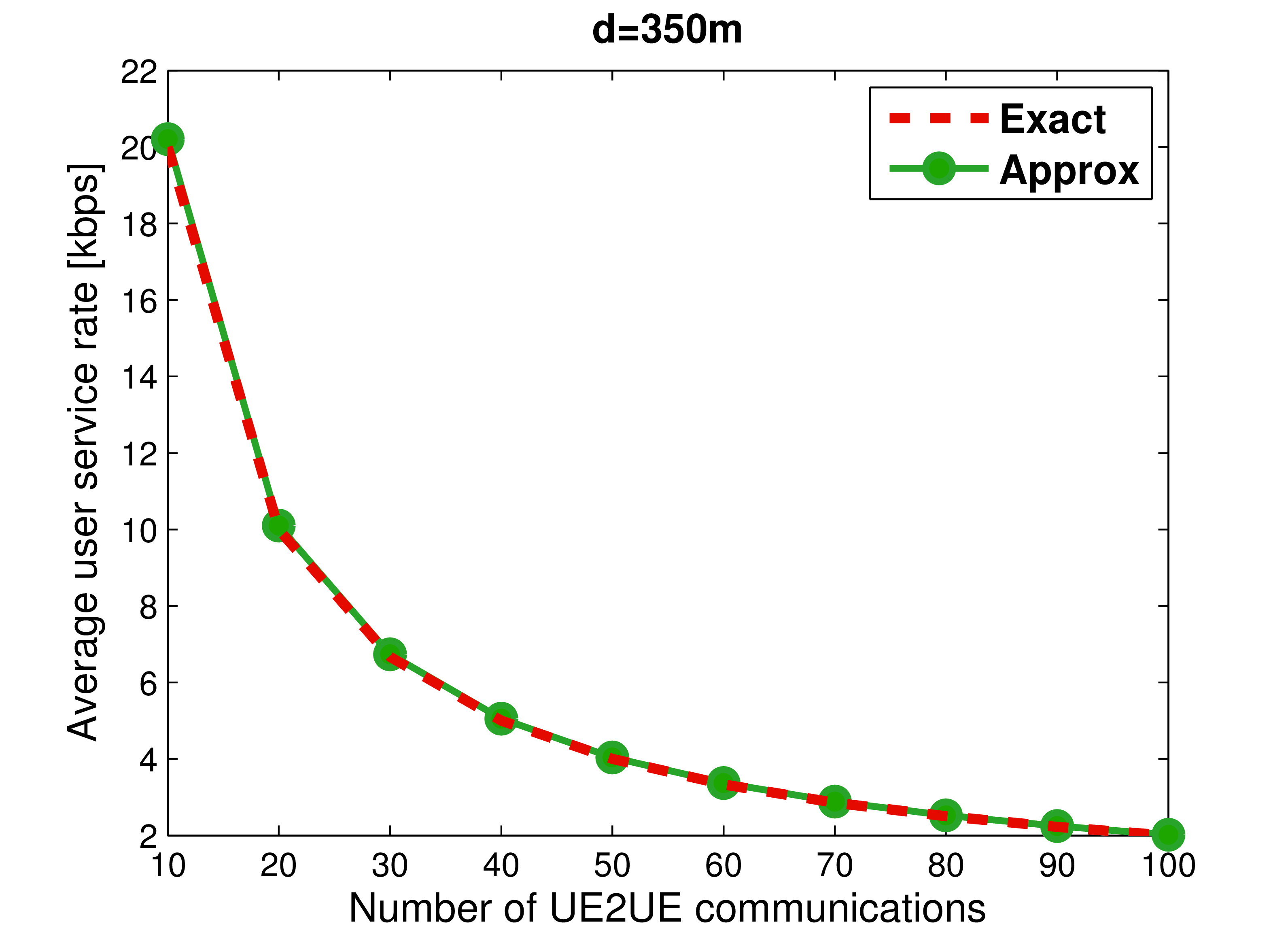}
\vspace{-3pt}
\caption{Exact vs. Approximated results}
\label{fig.AvgMu_WWO_D2D}
\end{figure}

\begin{figure}[H]
\vspace{-25pt}
\centering
\includegraphics[scale=0.1]{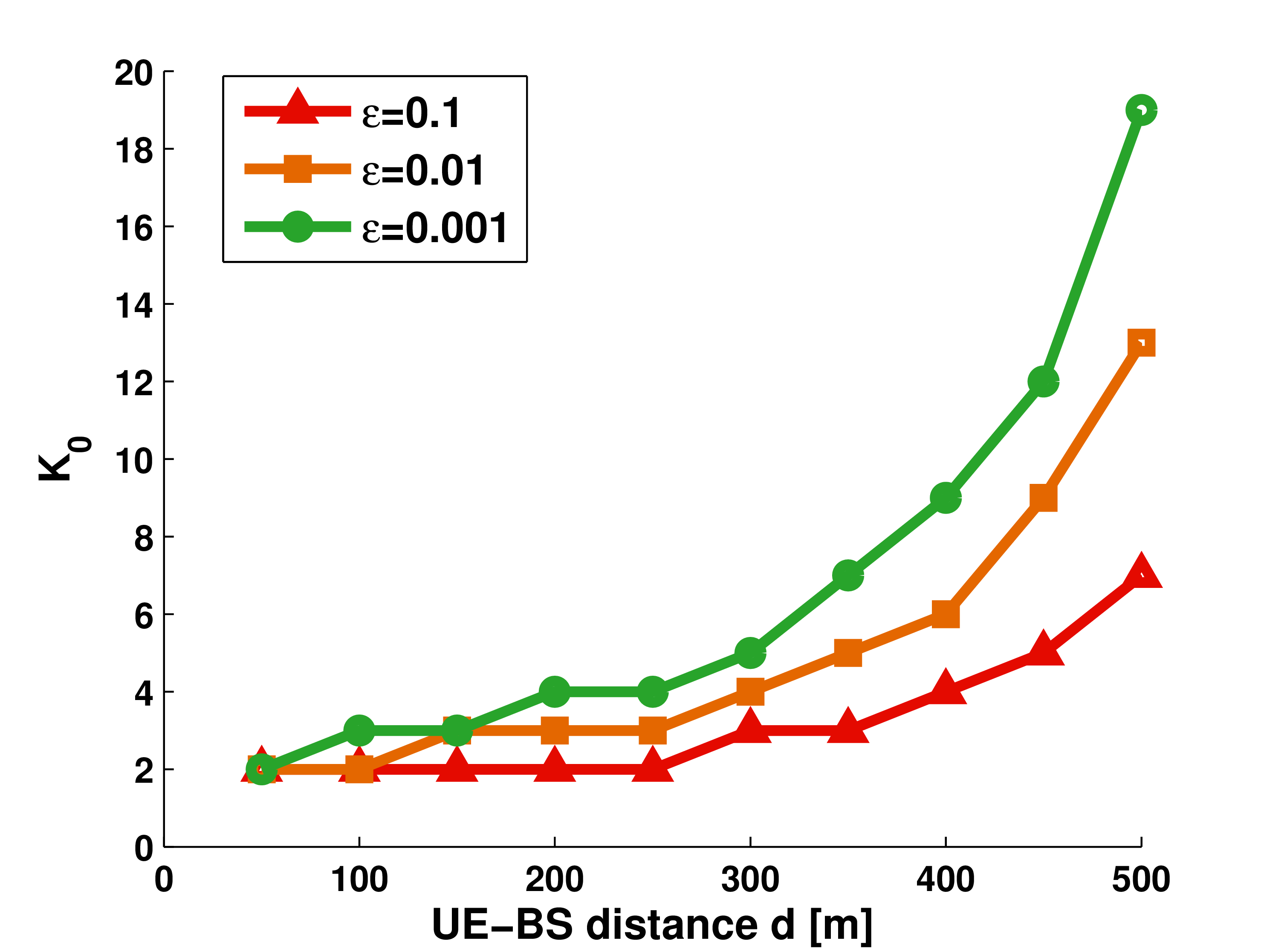}
\caption{Variation of the parameter $K_0$ as function of the distance $d$ between BS and users}
\label{fig.K0_func_d}

\end{figure}
\iffalse
{\color{red}
\textbf{\large{3. Performance of multi-UE scenario}}

3.a  performance as function of k

3.b centralized vs distributed

3.c causality for multi-UE

\textbf{\large{4. Performance of multi-UE scenario symmetric case}}

4.a  K0=f(d)

4.b performance function of d
}
\fi
\vspace{-20pt}
\section{Conclusion\label{sec:Conclusion}}
In this paper, we have carried out a queuing analysis to study a TDD network involving two types of communications: 1- \textit{UE2UE communications} passing through the BS and 2- \textit{UE2BS communications}. TDMA user scheduling and link adaptation model are the main assumptions in this work. First, we have shown the interest of our queuing theory approach by verifying that it provides a more realistic performance evaluation of the network compared to a performance analysis strictly based on physical layer that ignores the dynamic effect of the traffic pattern. This is due to the coupling between the UEs' queues and the BS queues. Second, we have given the analytic expressions of the exact stability region for the general scenario of different \textit{UE2UE} and \textit{UE2BS} \textit{communications}. Then, we have studied the complexity of such a computation. Finally, a highly accurate approximation is proposed such that simple analytic expressions of the stability regions are provided as convex polytopes with limited number of vertices (computationally feasible). The variation of these stability regions as a function of the queues and channels states is investigated. We examine the trade-off between the precision of the proposed approximations and their computational complexity. \\

\section*{APPENDIX \label{Appendix}}\renewcommand{\thesubsection}{\Alph{subsection}}

\subsection{Proof of theorem \ref{th_ss_cellular_real} \label{stab_real_twoRate}}
Here, we find the exact stability region for the three-UEs scenario. For this aim, we proceed as it follows: \textit{\textbf{step 1}} models the queue Q$_{BS}$ as a Markov chain and expresses the arrival and service probabilities of this chain,  \textit{\textbf{step 2}} computes the probability that the queue Q$_{BS}$ is empty,  \textit{\textbf{step 3}} obtains the service rate of both queues Q$_s$ and  Q$_u$, \textit{\textbf{step 4}} combines the results of the previous steps to provide the exact stability region of the three-UEs scenario, \textit{\textbf{step 5}} verifies that having a set of arrival rates within this stability region is equivalent to the stability of the system of queues. The main challenge is to solve the complicated Markov Chain that models the queue Q$_{BS}$ in order to find the probability that this queue is empty and deduce the service rate of the queues in the systems.

The coupling between the queues leads to a multidimensional Markov Chain modeling of the system of queues. Since this model complicates the study of the stability region, we approach this problem in a different way. We use the priority policies definition (in section \ref{sec:System-Model}) to transform the multidimensional model to a 1D Markov Chain (one dimensional). Hence, for a given priority policy, each queue can be modeled by a 1D Markov Chain. In order to characterize the stability region, all the priority policies should be considered or at least the priority policies that achieve the corner points of the stability region.

We recall the three-UEs scenario contains one \textit{UE2UE communication} and one \textit{UE2BS communication}. The \textit{UE2UE communication} is modeled by the cascade of Q$_{s}$ and Q$_{BS}$ and the \textit{UE2BS communication} is modeled by the queue Q$_{u}$. We consider the set of three rates $\left\lbrace r_1,r_2,r_3 \right\rbrace$ with $r_1=kr_2$ and $r_3=0$. So, if $P$ is the number of packets transmitted at $r_2$ then $kP$ is the number of packets transmitted at $r_1$.  Assuming that the two sources queues Q$_s$ and Q$_u$ are saturated, we want to characterize the stability region.

$\varGamma$ is the priority policy according to which the communications are sorted. We know that in order to characterize the stability region it is sufficient to consider the corner points of this region. These corner points correspond to the extreme policies where the priority is always given to the same communication or when the priority is always given to the communication that has the better channel state. Actually we consider two possible rate for each communication ($r_1$ and $r_2$) hence the priority policies $\Omega_{\Gamma}^{ss}$ that present the corner points of the stability region are the following:
\begin{itemize}
\item \textbf{Policy $\Gamma_1$}: \textit{UE2UE communication} is given higher priority, hence \textit{UE2BS communication} can only take place when the rate of the uplink and downlik of the \textit{UE2UE communication} are null.\\
\item \textbf{Policy $\Gamma_2$}: \textit{UE2BS communication} is given higher priority, hence the users of \textit{UE2UE communication} transmit only when the rate of the \textit{UE2BS communication} is null.\\
\item \textbf{Policy $\Gamma_3$}: At rate $r_{1}$ , the \textit{UE2UE communication} has the priority and at rate $r_{2}$, \textit{UE2BS communication} has the priority. Hence, \textit{UE2UE communication} transmits at $r_2$ only when the rate of the \textit{UE2BS communication} is null.\\
\item \textbf{Policy $\Gamma_4$}: At rate $r_{1}$, the \textit{UE2BS communication} has the priority and at rate $r_{2}$, \textit{UE2UE communication} has the priority. Hence, \textit{UE2BS communication} transmits at $r_2$ only when the rate of the uplink and downlik of the \textit{UE2UE communication} are null.\\
\item \textbf{Policy $\Gamma_5$}: Communication with highest rate transmits and when both communications have the same rate then \textit{UE2UE communication} is prioritized.\\
\item \textbf{Policy $\Gamma_6$}: Communication with highest rate transmits and when both communications have the same rate then \textit{UE2BS communication} is prioritized.\\
\end{itemize}

Thus, the set of priority policies $\Omega_{\Gamma}^{ss}$ that consists of the corner point of the stability region is given by:
\[
\Omega_{\Gamma}^{ss}=\lbrace \Gamma_1, \Gamma_2, \Gamma_3, \Gamma_4, \Gamma_5, \Gamma_6\rbrace
\]

For the cellular scenario, we denote by $q_{n}\left(S_{i},S_{j},S_{k},\Gamma\right)$ the probability to transmit over the link$_n$ (with $n\in \lbrace s, d,u \rbrace$) for a given policy $\Gamma$ and a given state of the links: link$_s$ at state $S_s$, link$_d$ is at state $S_d$ and link$_u$ is at state $S_u$ (with  $S_s$, $S_d$ and $S_u \in \left(S_{1},S_{2},S_{3}\right)$). Based on a simple analysis of the priority policies $\Gamma \in \Omega_{\Gamma}^{ss}$, these probabilities can be easily generated. Let us consider an example to clarify the procedure. We consider the priority policy $\Gamma_3$ and we give the values of the corresponding transmission probabilities in table \ref{table.q1}. Note that $q_{s}\left(S_{i},S_{j},S_{k},\Gamma_3\right)=q_{d}\left(S_{i},S_{j},S_{k},\Gamma_3\right)=q_{u}\left(S_{i},S_{j},S_{k},\Gamma_3\right)=0$ for all the the other combinations of states $\left(S_{i},S_{j},S_{k}\right)$ not shown in the table. By analogy, the transmission probabilities in the cellular scenario for all $\Gamma \in \Omega_{\Gamma}^{ss}$ are computed.\\

\begin{table}
%\begin{centering}
\begin{tabular}{|c|c|c|c|c|}
\hline 
Policy $\Gamma$ & $Q_{BS}$ State & $q_{s}(S_{i},S_{j},S_{k})$ & $q_{d}(S_{i},S_{j},S_{k})$ & $q_{u}(S_{i},S_{j},S_{k})$\tabularnewline
\hline 
\hline 
3 & Empty & $\begin{array}{cc}
q_s\,(S_{1},S_{1-2-3},S_{1-2-3})=1\\q_s\,(S_{2},S_{3},S_{1-2-3})=1
\end{array}$ & x  &$\begin{array}{cc}
q_u(S_{2-3},S_{1},S_{1-2-3})=1\\q_u(S_{2},S_{2-3},S_{1-2-3})=1
\end{array}$ \tabularnewline
\cline{3-5} 
 & Not Empty & $\begin{array}{cc}
q_s(S_{1},S_{1-2-3},S_{1})=\alpha_{1}\\
q_s(S_{1},S_{1-2-3},S_{2})=\alpha_{2}\\
q_s(S_{2},S_{1-2-3},S_{1})=\alpha_{3}\\
q_s(S_{2},S_{3},S_{2})=\alpha_{4}\\
q_s(S_{1},S_{1-2-3},S_{3})=1\\
q_s(S_{2},S_{3},S_{3})=1
\end{array}$& $\begin{array}{cc}
 q_d(S_{1},S_{1-2-3},S_{1})=1-\alpha_{1}\\
 q_d(S_{1},S_{1-2-3},S_{2})=1-\alpha_{2}\\
 q_d,(S_{2},S_{1-2-3},S_{1})=1-\alpha_{3}\\
q_d(S_{2},S_{3},S_{2})=1-\alpha_{4} \\
q_d(S_{3},S_{1-2-3},S_{1})=1\\
q_d(S_{3},S_{3},S_{2})=1
\end{array}$& $\begin{array}{cc}
q_u(S_{2-3},S_{1},S_{2-3})=1\\
q_u(S_{2},S_{2-3},S_{2})=1
\end{array}$ \tabularnewline
\hline  
\end{tabular}
%\par\end{centering}
\caption{Cell. Scenario: Transmission probabilities for the third priority policies $\Gamma_3$ and all the  possible channel states}
\label{table.q1}
\end{table}

\subsubsection*{\textbf{\textit{Step 1}}} \textit{Markov chain model of Q$_{BS}$}\\

Cellular communications are modeled as coupled processor sharing queues where the service rates of Q$_s$ (equivalent to the arrival rate of Q$_{BS}$) as well as the service rate of Q$_u$ depend on the state (empty or not) of Q$_{BS}$. Let us study the Markov chain of the queue Q$_{BS}$ in order to deduce the probability of being empty which means the probability of having at least one packet.\\

Q$_{BS}$ corresponds to the queue at the BS side that can be modeled as a Markov chain with a transition probability from a state $x_i$ to a state $x_{i+1}$ is equal to the probability of receiving $P$ packets per slot which is equivalent to the transmission probability of the queue Q$_s$ at rate $r_2$ due to the fact that the packets arriving to Q$_{BS}$ correspond to the packets transmitted by UE$_s$. Thus, the transition probability from a state $x_i$ to a state $x_{i+k}$ is equal to the probability of receiving $kP$ packets per slot which is equivalent to the transmission probability of the queue Q$_s$ at rate $r_1$. Moreover, the transition probability from a state $x_i$ to a state $x_{i-1}$ is equal to the transmission probability at the downlink BS-UE$_d$ at rate $r_2$ and the transition probability from a state $x_i$ to a state $x_{i-k}$ is equal to the transmission probability at the downlink BS-UE$_d$ at rate $r_1$. As we mentioned before, the scheduling decision depend on the state (empty or not) of Q$_{BS}$ which means that the behavior of Q$_s$ and Q$_u$ is coupled to the state of Q$_{BS}$ and the arrival probability when Q$_{BS}$ is empty (state $0$) is not equal to the arrival probability when Q$_{BS}$ is not empty (state $\neq0$). In figure \ref{fig.QBS_real_gen}, we present the discrete time Markov Chain with infinite states that describes the evolution of the queue Q$_{BS}$. 

\begin{figure}[H]
\centering
\captionsetup{justification=centering}
\includegraphics[width=0.8\textwidth]{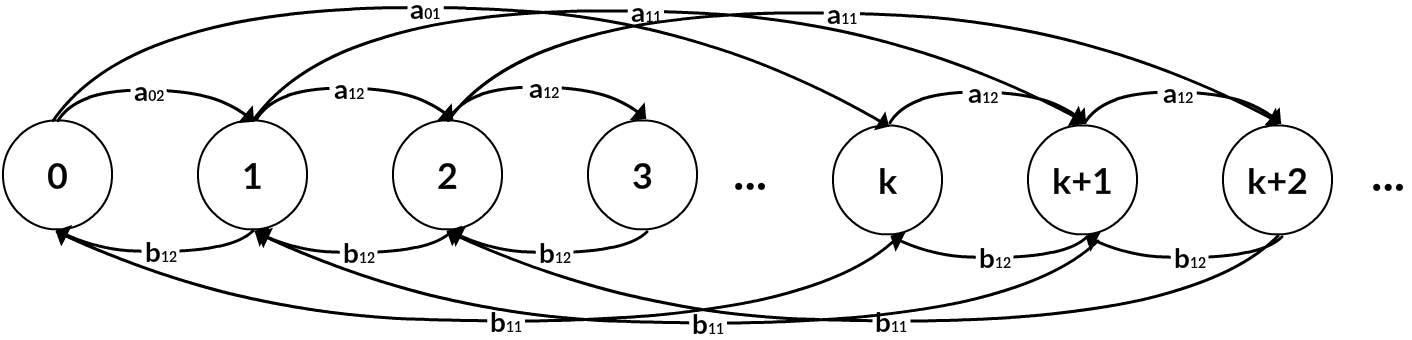}
\par
\caption{Markov Chain model of the queue Q$_{BS}$}
\label{fig.QBS_real_gen}
\end{figure}

We start by studying the queue Q$_{BS}$ in terms of probability of transmitting and probability of receiving packets in order to deduce the probability that the queue Q$_{BS}$ is empty.\\

\subsubsection*{\textbf{\textit{Step 1.a}}} \textit{Service probability of Q$_{BS}$}\\

For Q$_{BS}$, two service probabilities exist the first one b$_{11}$ corresponds to the probability of transmitting at rate $r_1$ and the second one b$_{12}$ corresponds to the probability of transmitting at rate $r_2$. 
\[
b_{11}\left(\Gamma\right)=p_d^{1}\sum_{i,j=1}^3p_s^i\,p^j_{u}\,q_{d}\left(S_{i},S_{j},S_{1},\Gamma\right)
\]
\[
b_{12}\left(\Gamma\right)=p_d^{2}\sum_{i,j=1}^3p^i_s\, p_u^{j}\, q_{d}\left(S_{i},S_{j},S_{2},\Gamma\right)
\]
Parameters $U,\,V,\,W',\,X',\,Y',\,Z',\,M$ and $N$ from table \ref{table.Notation} are used to make expressions simpler. These parameters depend on the considered priority policy $\varGamma$ and their values are deduced from that of the transmission probabilities (such that those in table \ref{table.q1} for $\Gamma_3$). The values of these parameters are specified in table \ref{table.Policy-Cell} for all $\Gamma \in \Omega_{\Gamma}^{ss}$. Hence, the service probabilities $b_{11}$ and $b_{12}$ of the queue $Q_{BS}$ can be written as it follows:
\begin{equation}
\label{QBS_b1_gen}
b_{11}=p_d^{1}\left(M-p_s^{1}W'-p_s^{2}Y'\right)=p_d^{1}\left(1-\alpha_{1}p_s^{1}-\alpha_{3}p_s^{2}\right)U
\end{equation}

\begin{equation}
\label{QBS_b2_gen}
b_{12}=p_d^{2}\left(N-p_s^{1}X'-p_s^{2}Z'\right)=p_d^{2}\left(p_s^1U+\bar{p}_s^1V-\alpha_{2}p_s^{1}U-\alpha_{4}p_s^{2}V\right)
\end{equation}

For each priority policy $\Gamma \in \Omega_{\Gamma}^{ss}$, we find the transmission probabilities of all the links ($s$, $d$ and $u$) then we compute the parameters $U,\,V,\,W',\,X',\,Y',\,Z',\,M$ and $N$ using table \ref{table.Notation}.

\subsubsection*{\textbf{\textit{Step 1.b}}} \textit{Arrival probability of Q$_{BS}$}

Due to the coupling between the queues Q$_s$ and Q$_u$ with the state (empty or not) of Q$_{BS}$, we should differ between (i) $a_{01}$ and $a_{02}$ that denote the arrival probabilities receptively at rate $r_1$ and $r_2$ when Q$_{BS}$ is empty and (ii) $a_{11}$ and $a_{12}$ that denote the arrival probabilities receptively at rate $r_1$ and $r_2$ when Q$_{BS}$ is not empty. For a given priority policy $\Gamma$, these probabilities are given by:

\[
a_{01}\left(\Gamma \right)=p_s^{1} \sum \limits_{j,k=1}^{3}p_u^{j}p_d^{k}q_{s}(S_{1},S_{j},S_{k},\varGamma|Q_{BS}(t)=0)
\]
\[
a_{02}\left(\Gamma \right)=p_s^{2}\sum   \limits_{j,k=1}^{3}p_u^{j}p_d^{k}q_{s}(S_{2},S_{j},S_{k},\varGamma|Q_{BS}(t)=0)
\]
\[
a_{11}\left(\Gamma \right)=p_{11}\sum \limits_{j,k=1}^{3}p_u^{j}p_d^{k}q_{s}(S_{1},S_{j},S_{k},\varGamma|Q_{BS}(t)>0)
\]
\[
a_{12}\left(\Gamma \right)=p_{21}\sum \limits_{j,k=1}^{3}p_u^{j}p_d^{k}q_{s}(S_{2},S_{j},S_{k},\varGamma|Q_{BS}(t)>0)
\]

\begin{table}[H]
\centering
  \begin{tabular}{|c|c|}
\hline 
Notation & Expression \tabularnewline
\hline 
\hline 
$U$ & $\sum \limits_{i=1}^{3}p_u^{i}q_{s}(S_{1},S_{i},-,\varGamma|Q_{BS}(t)=0)$ \tabularnewline
\hline 
$V$ & $\sum \limits_{i=1}^{3}p_u^{i}q_{s}(S_{2},S_{i},-,\varGamma|Q_{BS}(t)=0)$ \tabularnewline
\hline 
$W'$ & $\sum \limits_{i=1}^{3}p_u^{i}q_{s}(S_{1},S_{i},S_{1},\varGamma|Q_{BS}(t)>0)=\alpha_1U$\tabularnewline
\hline 
$X'$ & $\sum \limits_{i=1}^{3}p_u^{i}q_{s}(S_{1},S_{i},S_{2},\varGamma|Q_{BS}(t)>0)=\alpha_2U$\tabularnewline
\hline 
$Y'$ & $\sum \limits_{i=1}^{3}p_u^{i}q_{s}(S_{2},S_{i},S_{1},\varGamma|Q_{BS}(t)>0)=\alpha_3U$\tabularnewline
\hline 
$Z'$ & $\sum  \limits_{i=1}^{3}p_u^{i}q_{s}(S_{2},S_{i},S_{2},\varGamma|Q_{BS}(t)>0)=\alpha_4V$\tabularnewline
\hline 
$W$ & $\sum \limits_{i,j=1}^{3}\sum  \limits_{i=1}^{3}p_s^{i}p_d^{j}q_{u}(S_{i},S_{1},S_{j},\varGamma|Q_{BS}(t)=0)$ \tabularnewline
\hline 
$X$ & $\sum  \limits_{i,j=1}^{3}\sum  \limits_{i=1}^{3}p_s^{i}p_d^{j}q_{u}(S_{i},S_{2},S_{j},\varGamma|Q_{BS}(t)=0)$\tabularnewline
\hline 
$Y$ & $\sum  \limits_{i,j=1}^{3}\sum  \limits_{i=1}^{3}p_s^{i}p_d^{j}q_{u}(S_{i},S_{1},S_{j},\varGamma|Q_{BS}(t)>0)$\tabularnewline
\hline 
$Z$ & $\sum  \limits_{i,j=1}^{3}\sum  \limits_{i=1}^{3}p_s^{i}p_d^{j}q_{u}(S_{i},S_{2},S_{j}|Q_{BS}(t)>0$\tabularnewline
\hline 
$A$ & $\sum  \limits_{i=1}^{3}p_u^{i}q_{s}^{d}(S_{1},S_{i},\varGamma)=U$\tabularnewline
\hline 
$B$ & $\sum  \limits_{i=1}^{3}p_u^{i}q_{s}^{d}(S_{2},S_{i},\varGamma)=V$\tabularnewline
\hline 
$C$ & $\sum  \limits_{i=1}^{3}p_s^{i}q_{u}^{d}(S_{i},S_{1},\varGamma)$\tabularnewline
\hline 
$D$ & $\sum  \limits_{i=1}^{3}p_s^{i}q_{u}^{d}(S_{i},S_{2},\varGamma)$\tabularnewline
\hline 
$M$ & $p_s^{1}U+p_s^{2}U+p_s^{3}U=U$\tabularnewline
\hline 
$N$ & $p_s^{1}U+p_s^{2}V+p_s^{3}V=p_s^{1}U+\bar{p}_s^{1}V$\tabularnewline
\hline 
\end{tabular} 

\caption{Notation to make expressions simpler}
\label{table.Notation}
\end{table}

Using parameters $U$ and $V$ from table \ref{table.Notation}, the probabilities of arrival of the queue $Q_{BS}$ can be written as it follows:
\begin{equation}
\label{QBS_a01_gen}
a_{01}=p_s^{1}U
\end{equation}
\begin{equation}
\label{QBS_a02_gen}
a_{02}=p_s^{2}V
\end{equation}
\begin{equation}
a_{11}=p_s^{1}\left(\alpha_{1}p_d^{1}+\alpha_{2}p_d^{2}+p_d^{3}\right)U
\label{QBS_a11_gen}
\end{equation}
\begin{equation} 
a_{12}=p_s^{2}\left(\alpha_{3}p_d^{1}U+\alpha_{4}p_d^{2}V+p_d^{3}V\right)
\label{QBS_a12_gen}
\end{equation}

\subsubsection*{\textbf{\textit{Step 2}}\label{subsec:Pi_0_2rates}} \textit{Probability that Q$_{BS}$ empty}

Let us consider $a_0=a_{01}+a_{02}$, $a_1=a_{11}+a_{12}$ and $b_1=b_{11}+b_{12}$. We solve the balance equations in order to find the stationary distribution $\varPi$ of the Markov chain of Q$_{BS}$ and more specifically the probability that Q$_{BS}$ is empty $\varPi_0$. Recall that the stationary distribution of a Markov Chain exists if and only if  the stability condition is verified. Hence, finding the stationary distribution based on the balance equations is a necessary and sufficient condition for guaranteeing the stability of Q$_{BS}$.  As shown in figure (\ref{fig.QBS_real_gen}), the balance equations at the different states $i$ of the Markov chain are the following:

\begin{equation}
\label{eq.BE1_gen}
\text{For }i=0: \,\,\,\, \varPi_0a_{0}=\varPi_1b_{12}+\varPi_kb_{11}
\end{equation}
\begin{equation}
\label{eq.BE2_gen}
\text{For }i=1: \,\,\,\, \varPi_1\left(a_{1}+b_{12}\right)=\varPi_0a_{02}+ \varPi_2b_{12}+\varPi_{k+1}b_{11}
\end{equation}
\begin{equation}
\label{eq.BE2tok_gen}
\text{For }2\leq i\leq k-1: \,\,\,\, \varPi_i\left(a_{1}+b_{1}\right)=\varPi_{i-1}a_{12}+\varPi_{i+1}b_{12}+\varPi_{i+k}b_{11}
\end{equation}
\begin{equation}
\label{eq.BEk_gen}
\text{For }i=k: \,\,\,\, \varPi_k\left(a_{1}+b_{1}\right)=\varPi_{0}a_{01}+\varPi_{k-1}a_{12}+\varPi_{k+1}b_{12}+\varPi_{2k}b_{11}
\end{equation}
\begin{equation}
\label{eq.BalanceEquation_gen}
\text{For }i>k: \,\,\,\, \varPi_i\left(a_{1}+b_{1}\right)=\varPi_{i-k}a_{11}+\varPi_{i-1}a_{12}+\varPi_{i+1}b_{12}+\varPi_{i+k}b_{11}
\end{equation}

We look for a solution of the balance equations of the form $\varPi_n=x^n$. Then we construct a linear combination of these solutions (which also satisfying the boundary equations (\ref{eq.BE1_gen}), (\ref{eq.BE2_gen}), (\ref{eq.BE2tok_gen})) that verifies the normalization equation $\sum\limits_n\varPi_n=1$. \\
We start by substituting $\varPi_n$ by $x^n$ in (\ref{eq.BalanceEquation_gen}) and then dividing by $x^{n-k}$. This yields the following polynomial equation: 
\[
P\left(x\right)=b_{11}x^{2k}+b_{12}x^{k+1}-\left(a_{12}+a_{11}+b_{12}+b_{11}\right)x^k+a_{12}x^{k-1}+a_{11}
\]
\begin{equation}
=\left(x-1\right)\left(b_{11}\sum_{i=k+1}^{2k-1}x^i+\left(b_{11}+b_{12}\right)x^k-\left(a_{11}+a_{12}\right)x^{k-1}-a_{11}\sum_{i=0}^{k-1}x^i\right)
\label{eq.Polynom_gen}
\end{equation}
 
 The first root $x=1$ but this one is not a useful, since we must be able to normalize the solution of the equilibrium equations. Stability conditions requires that this polynomial has at least one root $x$ with $|x|<1$. Let us say we found $R$ roots such that $|x|<1$:  $x_1,x_2,x_3...x_{R}$ with $1\leq R \leq 2k-1$ (degree of the polynomial equation). Then the stationary distribution is given by the following linear combination:  
\begin{equation}
 \varPi_n=\sum^{R}_{i=1}c_ix_i^n \,\,\,\,\,\, n=k+1,k+2 
\label{eq.Pi_n_gen}
\end{equation}
Where $\varPi_n$ for $0\leq n \leq k$ as well as the coefficients $c_i$ for $0\leq i \leq R$ are computed by solving the following system of equations: (i) first $R+K$ balance equations and (ii) the normalization equation: $\sum \limits_n\varPi_n=1$.\\

For each choice of the coefficients $c_k$ the linear combination satisfies (\ref{eq.BalanceEquation_gen}). These coefficients add some freedom that can be used to also satisfy the initial balance equations for $0 \leq i \leq k$ as well as the normalization equation. 
Therefore, we have $R+k+1$ unknowns to determine in order to find the stationary distribution of the queue Q$_{BS}$, these unknowns are: the $R$ coefficients $c_i$ and the $k+1$ probabilities $\varPi_i$ for $0\leq i \leq k$. These unknowns are found by solving the following linear system $AX=B$ with $X^T=\left[ \varPi_0 ,\, \varPi_1 ,\, \varPi_2 ,\, ... ,\,\varPi_{k-1} ,\, \varPi_{k-1} ,\,c_1 ,\, c_2 ,\, c_3 ... ,\,c_{R-1},\, c_{R}\right]$ (Note that $X^T$ is the transpose of the vector $X$). 
By substituting  of (\ref{eq.Pi_n_gen}) into the first  $R+k-1$ balance equation as well as the normalization equation, we deduce for $k>2$ the matrix $A$ and $B$ of the linear system as it follows. (The special case of $k=2$ is treated afterward.
\\
%\begin{landscape}
\begin{center}

{\small{
\[
\hspace{-30pt}
A_0=\left[\begin{array}{ccccccccccc}
State & \varPi_0 & \varPi_1 & \varPi_2 & \varPi_3 & \varPi_4 & ... & \varPi_{k-2} & \varPi_{k-1} & \varPi_{k}\\
i=0 & -a_0 & b_{12} & 0 & 0 & 0 & ... & 0 & 0 & b_{11}\\
i=1 & a_{02} & -\left(a_1+b_{12}\right) & b_{12} & 0 & 0 & ... & 0 & 0 & 0\\
i=2 & 0 & a_{12} & -\left(a_1+b_{12}\right) & b_{12} & 0 & ... & 0 & 0 & 0\\
... & 0 & 0 & a_{12} & -\left(a_1+b_{12}\right) & b_{12} & 0... & 0 & 0 & 0\\
&  &  &  &  & &  ... &  &  &  \\
i=k-1 & 0 & 0 & 0 & 0 & 0 & ... & a_{12} & -\left(a_1+b_{12}\right) & b_{12}\\
i=k & a_{0k} & 0 & 0 & 0 & 0 & ... & 0 & a_{12} & -\left(a_1+b_{1}\right)\\
i=k+1 & 0 &a_{1k} & 0 & 0 & 0 & ... & 0 & 0 & a_{12}\\
i=k+2 & 0 & 0 & a_{1k} & 0 & 0 & ... & 0 & 0 & 0\\
i=k+3  & 0 & 0 & 0 &  a_{1k} & 0 & ... & 0 & 0 & 0 \\
&  &  &  &  & &  ... &  &  &  \\
i=2k & 0 & 0 & 0 &  0 & 0 & ... & 0 & 0 & a_{1k}\\
i=2k+1 & 0 & 0 &0 &  0 & 0 & ... & 0 & 0 & 0\\
... & 0 & 0 &0 &  0 & 0 & ... & 0 & 0 & 0\\
i=R+k-1 & 0 & 0 &0 &  0 & 0 & ... & 0 & 0 & 0\\
Norm & 1 & 1 & 1 & 1 & 1 &  ... & 1 & 1 & 1 \\
\end{array}\right]\]

\[\forall 1\leq i\leq R: A_i=
\left[\begin{array}{c}
c_i \\
0 \\
b_{11}x_i^{k+1}\\
b_{11}x_i^{k+2} \\
b_{11}x_i^{k+3}  \\
... \\
b_{11}x_i^{2k-1} \\
b_{12}x_i^{k+1} +b_{11}x_i^{2k}\\
-\left(a_1+b_{1}\right)x_i^{k+1} + b_{12}x_i^{k+2} +b_{11}x_i^{2k+1}\\
a_{12}x_i^{k+1}  -\left(a_1+b_{1}\right)x_i^{k+2} + b_{12}x_i^{k+3} +b_{11}x_i^{2k+2}\\
a_{12}x_i^{k+2} -\left(a_1+b_{1}\right)x_i^{k+3} + b_{12}x_i^{k+4} +b_{11}x_i^{2k+3}\\
... \\
a_{12}x_i^{2k-1} -\left(a_1+b_{1}\right)x_i^{2k} + b_{12}x_i^{2k+1} +b_{11}x_i^{3k}\\
a_{1k}x_i^{k+1}+ a_{12}x_i^{2k} -\left(a_1+b_{1}\right)x_i^{2k+1} + b_{12}x_i^{2k+2} +b_{11}x_i^{3k+1}\\
a_{1k}x_i^{k+2}+ a_{12}x_i^{2k+1} -\left(a_1+b_{1}\right)x_i^{2k+2} + b_{12}x_i^{2k+3} +b_{11}x_i^{3k+2}\\
a_{1k}x_i^{R-1}+ a_{12}x_i^{R+k-2} -\left(a_1+b_{1}\right)x_i^{R+k-1} + b_{12}x_i^{R+k} +b_{11}x_i^{R+2k}\\
\frac{1}{1-x_{i}}-\sum_{j=0}^{k}x_{i}^j \\
\end{array}\right]
\]
}}
\end{center}
\[
A=\left[\begin{array}{cccccc}
 A_0 & A_1  & A_2  & ...  & A_{R-1}  & A_R \\
\end{array}\right]\]

\[
B^T=\left[\begin{array}{cccccccccccccc}
 0 & 0  & 0  & 0  & 0  & 0 &  0...0 & 0  & 0  & 0  & 0  & 0  & 0 & 1 \\
\end{array}\right]\]

Thus, solving the linear system $AX=B$ with success induces two results: (i) stability of the queue Q$_{BS}$ and (ii) the stationary distribution of Q$_{BS}$ and especially what interest us is the probability that this queue is empty $\varPi_0$.\\

The result above holds for any $k \in \mathbb{N}^+$. However, we explicit the result for the case $\left(r_1=2r_2 \Rightarrow k=2\right)$ where $k-1=1$ hence the case $2\leq i\leq k-1$ should not be taken into the account for this especial case. Thus, the set of balance equations is the following:
\newpage
For $i=0$:
\begin{equation}
\label{eq.BE1}
\varPi_{0}a_{0}=\varPi_{1}b_{12}+\varPi_2b_{11}
\end{equation}
For $i=1$:
\begin{equation}
\label{eq.BE2}
\varPi_1\left(a_{1}+b_{12}\right)=\varPi_0a_{02}+ \varPi_2b_{12}+\varPi_3b_{11}
\end{equation}
For $i=2$:
\begin{equation}
\label{eq.BE3}
\varPi_2\left(a_{1}+b_{1}\right)=\varPi_0a_{01}+ \varPi_1a_{12}+\varPi_3b_{12}+\varPi_4b_{11}
\end{equation}
For $i\geq3$:
\begin{equation}
\varPi_i\left(a_{1}+b_{1}\right)=\varPi_{i-2}a_{11}+\varPi_{i-1}a_{12} +\varPi_{i+1}b_{12}+\varPi_{i+2}b_{11}
\label{eq.BalanceEquation}
\end{equation}

By analogy to the general case of $r_1=kr_2$, we look for a solution of the balance equations of the form $\varPi_{BS}\left( n\right)=x^n$ and then we construct a linear combination of these solutions also satisfying the boundary equations (\ref{eq.BE1}), (\ref{eq.BE2}), (\ref{eq.BE3}) as well as the normalization equation $\sum_n\varPi_{BS}\left(n\right)=1$. 
 Substituting of $\varPi_{BS}\left( n\right)$ by $x^n$ into (\ref{eq.BalanceEquation}) and then dividing by $x^{n-2}$ yield the following polynomial equation: 
\[
 P\left(x\right)=b_{11}x^4+b_{12}x^3-\left(a_{12}+a_{11}+b_{12}+b_{11}\right)x^2+a_{12}x+a_{11}
 \]
\begin{equation}
=\left(x-1\right)\left(b_{11}x^3+\left(b_{11}+b_{12}\right)x^2-\left(a_{11}+a_{12}\right)x-a_{11}\right)
\label{eq.Polynom}
\end{equation}
 
Recall that the root $x=1$ is not a useful one, since we must be able to normalize the solution of the equilibrium equations. Stability conditions requires that this polynomial has at least one root $x$ with $|x|<1$ say $x_1,x_2,x_3$. We now consider the stationary distribution given by the following linear combination:  
\begin{equation}
 \varPi_n=\sum^3_{k=1}c_kx_k^n \,\,\,\,\,\, n=3,4 ...
\label{eq.Pi_n}
\end{equation}

 and $\varPi_n$ for $=n=0,1,2$ are given by the boundary equations.
 For each choice of the coefficients $c_k$ the linear combination satisfies (\ref{eq.BalanceEquation}). These coefficients add some freedom that can be used to also satisfy the equations (\ref{eq.BE1}), (\ref{eq.BE2}), (\ref{eq.BE3}) and the normalization equation. 
 Substituting  of (\ref{eq.Pi_n}) into equations (\ref{eq.BE1}), (\ref{eq.BE2}), (\ref{eq.BE3}) we deduce $\varPi_0$, $\varPi_1$, $\varPi_2$ as function of  $c_1$,$c_1$ and $c_3$. Then substituting  of (\ref{eq.Pi_n}) for $n>3$ as well as $\varPi_0$, $\varPi_1$, $\varPi_2$ as function of  $c_1$,$c_1$ and $c_3$  into the normalization equation and the balance equations (\ref{eq.BalanceEquation}) for $n=3,4$ yield a set of 3 linear equations for 3 unknowns coefficients $c_1$,$c_1$ and $c_3$.\\
 
It can be shown that the unknowns $\varPi_0$, $\varPi_1$ and $\varPi_2$, $c_1, c_2, c_3$ is found by solving the following linear problem $AX=B$ such that:\\

 {\small{
\[
\hspace{-30pt}
A=\left[\begin{array}{cccccc} 
\varPi_0 & \varPi_1 & \varPi_2 & c_1 & ... & c_3 \\
-a_0 & b_{12} & b_{11} & 0 & ... & 0 \\ 
a_{02} & -\left(a_1 +b_{1}\right) & b_{12} & b_{11}x_1^3 & ... & b_{11}x_3^3 \\
a_{01} & a_{12} & -\left(a_1 +b_{1}\right) & b_{12}x_1^3+b_{11}x_1^4 & ... & b_{12}x_3^3+b_{11}x_3^4 \\
0 & a_{11} & a_{12} & -\left(a_1 +b_{1}\right)x_1^3 + b_{12}x_1^4+b_{11}x_1^5 & ... & -\left(a_1 +b_{1}\right)x_3^3 + b_{12}x_3^4+b_{11}x_3^5 \\
0 & 0 & a_{11} & a_{12}x_1^3 -\left(a_1 +b_{1}\right)x_1^4 + b_{12}x_1^5+b_{11}x_1^6 & ... & a_{12}x_3^3 -\left(a_1 +b_{1}\right)x_3^4 + b_{12}x_3^5+b_{11}x_3^6 \\
1 & 1 & 1 & \frac{1}{1-x_1}-1-x_1-x_1^2 & ... & \frac{1}{1-x_3}-1-x_3-x_3^2 \end{array}\right]
\]
 }}\\

\[
B=\left[\begin{array}{cccccc} 
0 & 0 & 0 & 0 & 0 & 1 
 \end{array}\right]
\]
           
Solving the system of linear equations given by $AX=B$ gives the values of $\varPi_i$ for $0\leq i \leq k-1$ and of the coefficients $c_j$ for $1\leq j \leq R$. Therefore, we find the probability that the  queue BS is empty which is $\varPi_0$.\\

\iffalse
\subsubsection{Stability condition of $Q_{BS}$}
\textbf{Equivalence between balance equations and stability}
\fi

\subsubsection*{\textbf{\textit{Step 3}}\label{subsec:Pi_0_2rates}} \textit{Service rates of Q$_s$ and Q$_u$}\\

We follow the procedure below, based on queuing theory analysis of the network capacity, in order to derive the stability region study of the cellular scenario. For the 3-UEs scenario, the stability region is characterized by computing the two service rates $\mu_s$ and $\mu_u$ of the \textit{UE2UE} and \textit{UE2BS communications}.\\

\subsubsection*{\textbf{\textit{Step 3.a}}} \textit{Service rate of Q$_{s}$}\\

Let us start with the service rate of the \textit{UE2UE communications} $\mu_s$. If Q$_{BS}$ is empty then the service rate of Q$_{s}$ is denoted by $\mu_s^0$, otherwise it is denoted by $\mu_s^1$. Hence, the average service rate of Q$_{s}$ for a given priority policy $\Gamma$ is computed as it follows: 

\[
\mu_{s}\left(\Gamma\right)=\mathbf{\mathbb{P}}\left[Q_{BS}=0\right]\mu_{s}^{0}\left(\Gamma\right)+\mathbf{\mathbb{P}}\left[Q_{BS}>0\right]\mu_{s}^{1}\left(\Gamma\right)
\]

\begin{equation}
\mu_{s}\left(\alpha,\Gamma\right)=\varPi_0\mu_{s}^{0}\left(\Gamma\right)+\left(1-\varPi_0\right)\mu_{s}^{1}\left(\alpha,\Gamma\right)\label{eq:Gmu1_real}
\end{equation}
with $\mu_{s}^{0}\left(\Gamma\right)$ and $\mu_{s}^{1}\left(\alpha,\Gamma\right)$ given by:
\[
\mu_{s}^{0}\left(\Gamma\right)=\mathbb{E}\left[\text{\ensuremath{\mu}}_{s}\left(\varGamma\,,\,Q|Q_{BS}=0\right)\right]
=\sum_{i=1}^{3}r_{i}p_s^{i}\sum \limits_{j=1}^3p_u^{j}q_{s}(S_{i},S_{j},-,\varGamma|Q_{BS}=0)
\]

\[
\mu_{s}^{1}\left(\alpha,\Gamma\right)=\mathbb{E}\left[\text{\ensuremath{\mu}}_{s}\left(\varGamma\,,\,Q|Q_{BS}>0\right)\right]=\sum_{i=1}^{3}r_{i}p_s^{i}\sum \limits_{j,k=1}^3p_u^{j}p_d^{k}q_{s}(S_{i},S_{j},S_{k},\varGamma|Q_{BS}>0)
\]

In order to simplify expressions, we consider the following notation $U$ and $V$ that depend on the considered priority policy $\varGamma$ and their values for the six priority policies $\Gamma \in  \Omega_{\Gamma}^{ss}$ are specified in table \ref{table.Policy-Cell}. Hence $\mu_{s}^{0}$ and $\mu_{s}^{1}$ can be written as it follows:

\begin{equation}
\label{mu10}
\mu_{s}^{0}=r_{1}p_s^{1}U+r_{2}p_s^{2}V
\end{equation}

\begin{equation}
\label{mu11}
\mu_{s}^{1}=r_{1}p_s^{1}\left(\alpha_{1}p_d^{1}+\alpha_{2}p_d^{2}+p_d^{3}\right)U+r_{2}p_s^{2}\left(\alpha_{3}p_d^{1}U+\alpha_{4}p_d^{2}U+p_d^{3}V\right)
\end{equation}

For each priority policy $\Gamma \in \Omega_\Gamma^{ss}$, we compute the probability that Q$_{BS}$ is empty $\varPi_0\left(\alpha,\Gamma\right)$ from \textit{ step 2} then we substitute the above expressions of $\mu_{s}^{0}\left(\Gamma\right)$ and $\mu_{s}^{1}\left(\alpha,\Gamma\right)$ into the equation (\ref{eq:Gmu1_real}), we obtain (\ref{eqn_mu1c_real}) as the service rate $\mu_{s}\left(\alpha,\Gamma\right)$ of the queue Q$_s$. \\

\subsubsection*{\textbf{\textit{Step 3.b}}} \textit{Service rate of Q$_{u}$}\\

After computing $\mu_s$, let us compute the second element of the stability region which is  the service rate $\mu_u$ of the \textit{UE2UE communication}. If Q$_{BS}$ is empty then the service rate of Q$_{u}$ is denoted by $\mu_u^0$ otherwise it is denoted by $\mu_u^1$. By analogy to the Q$_s$, the service rate of Q$_{u}$ is computed as it follows: 
\begin{equation}
\mu_{u}\left(\Gamma\right)=\varPi_0\mu_{u}^{0}\left(\Gamma\right)+\left(1-\varPi_0\right)\mu_{u}^{1}\left(\alpha,\Gamma\right)
\label{eq:Gmu2_real}
\end{equation}
with $\mu_{u}^{0}\left(\Gamma\right)$ and $\mu_{u}^{1}\left(\alpha,\Gamma\right)$ given by:

\newpage
\[\mu_{u}^{0}\left(\Gamma\right)=\mathbb{E}\left[\text{\ensuremath{\mu}}_{u}\left(S_{i},S_{j},S_{k},\varGamma\,,\,|Q_{BS}=0\right)\right]=\sum_{j=1}^{3}r_{j}p_u^{j}\sum \limits_{i=1}^3p_s^{i}q_{u}(S_{i},S_{j},-,\varGamma|Q_{BS}=0)
\]
\[\mu_{u}^{1}\left(\alpha,\Gamma\right)=\mathbb{E}\left[\text{\ensuremath{\mu}}_{u}\left(S_{i},S_{j},S_{k},\varGamma\,,\,Q|Q_{BS}>0\right)\right]=\sum_{j=1}^{3}r_{j}p_u^{j}\sum \limits_{i,k=1}^3p_s^{i}p_d^{k}q_{u}(S_{i},S_{j},S_{k},\varGamma|Q_{BS}>0)\]
\iffalse
\[=r_{1}p_{12}\sum_{i=1}^{3}\sum_{j=1}^{3}p_{i1}p_{j3}*q_{2}(S_{i},S_{1},S_{j},\varGamma|Q_{BS}(t)>0)\]
\[+r_{2}p_{22}\sum_{i=1}^{3}\sum_{j=1}^{3}p_{i1}p_{j3}*q_{2}(S_{i},S_{2},S_{j}|Q_{BS}(t)>0)
\]
\fi

Parameters $W,\,X,\,Y$ and $Z$ from table \ref{table.Notation} are used to make expressions simpler. These parameters depend on the considered priority policy $\varGamma$ and their values for the six priority policies $\Gamma \in  \Omega_{\Gamma}^{ss}$ are specified in table \ref{table.Policy-Cell}. Hence, $\mu_{u}^{0}$ and $\mu_{u}^{1}$ of the queue Q$_{u}$
can be written as it follows: 
\begin{equation}
\mu_{u}^{0}=r_{1}p_u^{1}W+r_{2}p_u^{2}X
\label{mu20}
\end{equation}
\begin{equation}
\mu_{u}^{1}=r_{1}p_u^{1}Y+r_{2}p_u^{2}Z
\label{mu21}
\end{equation}

For each priority policy $\Gamma \in \Omega_\Gamma^{ss}$, we compute the probability that Q$_{BS}$ is empty $\varPi_0\left(\alpha,\Gamma\right)$ from \textit{ step 2} then we substitute the above expressions of $\mu_{u}^{0}\left(\Gamma\right)$ and $\mu_{u}^{1}\left(\alpha,\Gamma\right)$ into the equation (\ref{eq:Gmu2_real}), we obtain (\ref{eqn_mu2c_real}) as the service rate $\mu_{u}\left(\alpha,\Gamma\right)$ of the queue Q$_u$. \\

\subsubsection*{\textbf{\textit{Step 4}}} \textit{Characterization of the stability region}\\

Combining the results of the previous steps, we deduce the exact stability region of the three-UEs scenario. Supposing that the arrival and service processes of Q$_{s}$ and Q$_{u}$ are strictly stationary and ergodic then their stability which is determined using \textbf{Loyne\textquoteright s} criterion is given by the condition that the average arrival rate is smaller than the average service rate.\\

The following procedure is pursue for capturing the stability region of the scenario. We start by considering a priority policies that correspond to the corner points of the stability region ($\Gamma \in \Omega_{\Gamma}^{ss}$ with $|\Omega_{\Gamma}^{ss}|=6$ ). Then, for this priority policy, we vary $\alpha \in \left[0,1\right]$. For each value of $\alpha$, we find the probability that Q$_{BS}$ is empty in order to deduce the service rates of the queues in the system. This procedure is applied for all the priority policies $\Gamma \in \Omega_{\Gamma}^{ss}$. \\

Therefore, the stability region for the 3-UEs cellular scenario is characterized by the set of mean arrival rates $\lambda_s$ and $\lambda_u$ in$\mathcal{R}_c^{ss}$  such that :

\[
\mathcal{R}_{c}^{ss}=co\left(\bigcup\limits_{\Gamma \in \Omega^{ss}_{\Gamma}}\bigcup\limits_{\alpha \in \left[0,1 \right]^4}\lbrace \mu_s\left(\alpha,\Gamma\right) , \mu_u\left(\alpha,\Gamma\right)  \rbrace\right)
\] 

where $\mu_{s}\left(\alpha,\Gamma\right)$ and $\mu_{u}\left(\alpha,\Gamma\right)$ are respectively given by (\ref{eqn_mu1c_real}) and (\ref{eqn_mu2c_real}). \\

\subsubsection*{\textbf{\textit{Step 5}}} \textit{Arrival rates within the stability region is equivalent to the stability of the system of queues}\\

We prove that if the mean arrival rates $\lambda_s$ and $\lambda_u$ are in the region $\mathcal{R}_c^{ss}$ is equivalent to the stability of the system of queues. To do so we prove that being having $\bm{\lambda}\in\mathcal{R}_c^{ss}$ gives the stability of the queues and vice versa. We suppose the following notation:
\begin{itemize}
\item $a_i\left(t \right)$: arrival process at queue Q$_i$ for $i=s,d,u$
\item $\lambda_s$ and $\lambda_u$ the mean arrival rates at respectively the \textit{UE2UE} and \textit{UE2BS communications}. By the law of large numbers, we have with probability 1:
\[
\lim_{t\rightarrow \infty}\frac{1}{t}\sum_{\tau=0}^{t-1}a_i\left(\tau\right)=\mathbb{E}\left[a_i\left(t\right)\right]=\lambda_i
\]
\item The second moments of the arrival processes $\mathbb{E}\left[a_i\left(t\right)^2\right]$ are assumed to be finite.
\item $b_i\left(t \right)$: departure process at queue Q$_i$ for $i=s,d,u$
\item  $S\left(t \right)=\left( S_i\left(t \right),S_j\left(t \right),S_k\left(t \right) \right)$: channel state vector where each SNR state is within $\lbrace S_1, S_2, S_3\rbrace$
\item $\Gamma \left(t\right)$: policy of scheduling on slot $t$\\
\end{itemize}

Thus for each channel  $i \in \left[s,d,u \right]$ the queuing dynamics are given by:
\[
Q_i\left(t+1\right)=max\left[Q_i\left(t\right)-b_i\left(t\right)\right] + a_i\left(t\right)
\] 

where $b_i\left(t\right)$ represents the amount of service offered to channel $i$ on slot $t$ and is defined by a function $\hat{b}_i \left(S\left(t \right),Q\left(t \right),\Gamma \left(t\right)\right)$:
\[
b_i \left( t \right)=\hat{b}_i \left(S\left(t \right),Q\left(t \right),\Gamma \left(t\right)\right)=r_i q_i\left(S\left(t \right),Q\left(t \right),\Gamma \left(t\right)\right)
\]
Further, by the law of large numbers, we have with probability 1:
\[
\mu_i\left(\Gamma\right)=\lim_{t\rightarrow \infty}\frac{1}{t}\sum_{\tau=0}^{t-1}b_i\left(\tau|\Gamma\right)=\lim_{t\rightarrow \infty}\frac{1}{t}\sum_{\tau=0}^{t-1}\hat{b}_i \left(S\left(\tau \right),Q\left(\tau \right)|\Gamma \right)
\]
\[
=\lim_{t\rightarrow \infty}\frac{1}{t}\sum_{\tau=0}^{t-1}\sum_{Q^*\in \mathbb{Q}}\hat{b}_i \left(S\left(\tau \right)|Q^*,\Gamma \right)\mathbbm{1}_{\left[ Q^* \right]}
\]
\[
=\lim_{t\rightarrow \infty}\frac{1}{t}\sum_{\tau=0}^{t-1}\left[\hat{b}_i \left(S\left(\tau \right)|Q_{BS}\left(\tau \right)=0,\Gamma \right)\mathbbm{1}_{\left[ Q_{BS}\left(\tau \right)=0 \right]}+\hat{b}_i \left(S\left(\tau \right)|Q_{BS}\left(\tau \right)>0,\Gamma \right)\mathbbm{1}_{\left[ Q_{BS}\left(\tau \right)>0 \right]}\right]
\]
\[
=\sum_{\left(S_i,S_j,S_k\right) \in \mathbb{S}}p_s^{i}p_u^{j}p_d^{k}r_iq_i \left(S_i,S_j,S_k | Q_{BS}=0,\Gamma \right)\mathbb{P}\left[Q_{BS}=0\right]
\]
\[
+\sum_{\left(S_i,S_j,S_k\right) \in \mathbb{S}}p_s^{i}p_u^{j}p_d^{k}r_iq_i \left(S_i,S_j,S_k | Q_{BS}>0,\Gamma \right)\mathbb{P}\left[Q_{BS}>0\right]
\]
\[
=\mu_{i}^0\left(\Gamma\right)\mathbb{P}\left[Q_{BS}=0\right]+\mu_{i}^1\left(\Gamma\right)\mathbb{P}\left[Q_{BS}>0\right]
\]
\[
=\mu_{i}^0\left(\Gamma\right)\varPi_0+\mu_{i}^1\left(\Gamma\right)\left(1-\varPi_0\right)
\]

\iffalse
The stability region $\mathbf{R}$ is given by the convex hull of the $\left(\mu_1\left(\Gamma\right),\mu_2\left(\Gamma\right)\right)$ for the different scheduling policies $\Gamma$ that corresponds to different prioritization of the channels.
\fi

\subsubsection*{\textbf{\textit{Step 5.a}}} \textit{$\lambda \in  \mathcal{R}_c^{ss} \Rightarrow $ Stability of the queues}\\

$\lambda \in  \mathcal{R}_c^{ss} \Rightarrow$ then for each channel $i=s,d,u$ it exists a $\mu_i^*=\sum\limits_{\Gamma} \varPi_{\Gamma} \mu_i\left(\Gamma\right)$ as combination of $\mu_i\left(\Gamma\right)$ for different scheduling policies $\Gamma$ such that $\lambda_i \leq \mu_i^*$. 
We have by the law of large numbers that:
\[
\lim_{t\rightarrow \infty}\frac{1}{t}\sum_{\tau=0}^{t-1}b_i\left(\tau\right)
=\lim_{t\rightarrow \infty}\frac{1}{t}\sum_{\Gamma^* \in \Omega_{\Gamma}}\sum_{\tau=0}^{t-1}b_i\left(\tau|\Gamma=\Gamma^*\right)\mathbbm{1}_{\Gamma=\Gamma^*}=\lim_{t\rightarrow \infty}\frac{1}{t}\sum_{\Gamma^* \in \Omega_{\Gamma}}\varPi_{\Gamma^*}\sum_{\tau=0}^{t-1}b_i\left(\tau|\Gamma=\Gamma^*\right)
\]
\[
=\sum_{\Gamma^* \in \Omega_{\Gamma}}\varPi_{\Gamma^*}\mu_i\left(\Gamma^*\right)
=\mu_i^*
\]
Hence $\lambda_i \leq \mu_i^* =\lim_{t\rightarrow \infty}\frac{1}{t}\sum_{\tau=0}^{t-1}b_i\left(\tau\right)$ which gives that queue $i$ is stable for $i=s,d,u$.
We deduce that if $\lambda \in  \mathcal{R}_c^{ss}$ then the system of queues is stable.\\

\iffalse
We use the Lyapunov drift to prove the stability.  The Lyapunov function is given by:
\[
L\left(Q\left(t\right)\right)=\frac{1}{2}\left[Q_1\left(t\right)+Q_2\left(t\right)+Q_3\left(t\right)\right]
\]
\fi

\subsubsection*{\textbf{\textit{Step 5.b}}} \textit{ Stability of the queues $\Rightarrow \lambda \in  \mathcal{R}_c^{ss}$}\\

If the queues are stable then each queue $i \in \left\lbrace s,d,u\right\rbrace$ has a stationary distribution $\varPi_i$. The mean service rate is given by:
\[
\mu_i=\mathbb{E}\left[ \hat{b}_i \left(S\left(t\right),Q\left(t\right),\Gamma\left(t\right) \right)\right]
\]
\[
=\sum_{\Gamma^*\in\Omega_{\Gamma}}\varPi_{\Gamma^*}\left[\mathbb{E}\left[ \hat{b}_i \left(S\left(t\right)|Q_{BS}=0,\Gamma=\Gamma^* \right)\varPi_{BS}\left(0\right)\right]+\mathbb{E}\left[\hat{b}_i \left(S\left(t\right)|Q_{BS}>0,\Gamma=\Gamma^* \right)\bar{\varPi}_{BS}\left(0\right)\right]\right]
\]
\[
=\sum_{\Gamma^* \in \Omega_{\Gamma}}\varPi_{\Gamma^*}\mu_i\left(\Gamma^*\right)
\]
Knowing that each queue $i\in \left\lbrace s,d,u \right\rbrace$ is stable then $\lambda_i \leq \mu_i$ which gives  $\lambda_i \leq \sum \limits_{\Gamma^* \in \Omega_{\Gamma}}\varPi_{\Gamma^*}\mu_i\left(\Gamma^*\right)$ which is a combination of the limit of the stability region  $\mathcal{R}_c^{ss}$ hence $\lambda \in  \mathcal{R}_c^{ss}$.\\

\subsection{Proof of lemma \ref{th2_gen} \label{A_gen}}

Here, we propose an approximated model for the three-UEs scenario and we define the condition that the fraction vector  $\alpha$ should verify in order to satisfy the stability of the approximated BS queue $\tilde{Q}_{BS}$. For this aim, we proceed as it follows: \textit{\textbf{step 1}} proposes the approximated Markov chain model of $\tilde{Q}_{BS}$ and gives the expression of the probability that this queue is empty as function of the arrival and service probabilities, \textit{\textbf{step 2}} provides the expressions of the arrival and service probabilities of this chain, \textit{\textbf{step 3}} deduces the stability condition that the corresponding fraction vector $\tilde{\alpha}$ should verify. The main challenge is to find an approximated Markov Chain that satisfies the following criteria: simple model, with an explicit form of the stationary distribution, that reduces complexity and finally an approximation that provides close results to the exact stability region. \\

\subsubsection*{\textbf{\textit{Step 1}}} \textit{Approximated system model}\\

Cellular communications are modeled as coupled processor sharing queues where the service rates of Q$_s$ (equivalent to the arrival rate of Q$_{BS}$) as well as the service rate of Q$_u$ depend on the state (empty or not) of Q$_{BS}$. We have already studied the stability region of the cellular scenario by consider the real Markov chain representing the queue Q$_{BS}$. This real Markov chain consider the transition probabilities of both rates $r_1$ and $r_2$. We note that finding such stability region is done by varying the vector $\alpha\in\left[0,1\right]^4$ and solving for each $\alpha$ the linear system of equations (see subsection \ref{subsec:Pi_0_2rates}) in order to find the probability that Q$_{BS}$ is empty.\\
Let us present an approximation of the stability region for which we provide an explicit formula of the points $\left(\mu_s,\mu_u\right)$ as well as the exact fraction vector $\tilde{\alpha}^*$ for which the limit of the stability region is achieved. In other terms, this explicit formula avoid us the consideration of all the $\alpha \in \left[0,1 \right]^4$ as well as the resolving of the system of equations.

To do so, we consider the average number of packets received and transmitted per time slot at the BS side in order to find the probability of receiving and transmitting a packet per time slot at the BS level. Thus, the goal is to have an approximated queue $\tilde{Q}_{BS}$ that can be modeled by a simple birth and death Markov chain.
We remark first that the average number of packets received by the BS is given by $Pa_{02}+kPa_{01}$ when the queue is empty and this average is equal to $Pa_{12}+kPa_{11}$ otherwise, we note second that the average number of transmitted packets by the BS is equal to $Pb_{12}+kPb_{11}$. In other terms, with probability $a_{12}+ka_{11}$ the BS is receiving $P$ packets per time slot and with probability $b_{12}+kb_{11}$ the BS is transmitting $P$ packets per time slot.

The approximation is based on the simplification of the Markov Chain model that helps us to reduce the complexity of the problem. Therefore, the approximated Markov Chain still take into account the multiple rate model at the level of the probabilities of transition. Indeed, the multiple rate model is integrated in such a way that the probability of transmitting $P$ packets at rate $r_1$ is given by the sum of the two following terms: (i) $k$ times the probability of having a transmission at a rate $r_1$ and (ii) the probability of having a transmission at a rate $r_2$.\\

\subsubsection*{\textbf{\textit{Step 1.a}}} \textit{Approximated Markov Chain  $\tilde{Q}_{BS}$}

As we mentioned before, the scheduling decision depend on the state (empty or not) of Q$_{BS}$ which means that the behavior of Q$_s$ and Q$_u$ is coupled to the state of Q$_{BS}$. Thus, the arrival probability when Q$_{BS}$ is empty (state $0$ of Markov chain) is not equal to the arrival probability when Q$_{BS}$ is not empty (state $\neq0$ of Markov chain). 
 The approximated queue $\tilde{Q}_{BS}$ corresponds to the queue at the BS side for which the transition probability from a state $x_i$ to a state $x_{i+1}$ is equivalent to the probability of transmitting $P$ packets per time slot by Q$_{s}$ which is equal to the probability of receiving $P$ packets per time slot by $\tilde{Q}_{BS}$ which means a probability of $a_{02}+ka_{01}$ when $\tilde{Q}_{BS}$ is empty and $a_{12}+ka_{11}$ otherwise. Moreover, the transition probability from a state $x_i$ to a state $x_{i-1}$ is equivalent to the probability of transmitting $P$ packets per time slot at the downlink BS-UE$_d$ which means a probability of $b_{12}+kb_{11}$. In figure \ref{fig.QBS}, we present the discrete time Markov Chain with infinite states which describes the evolution of the approximated queue $\tilde{Q}_{BS}$. \\
 
\begin{figure}[H]
\begin{centering}
\includegraphics[width=0.6\textwidth]{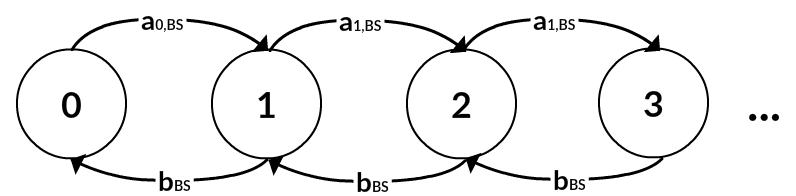}
\par
\caption{Approximated Markov Chain model of the BS queue: $\tilde{Q}_{BS}$}
\label{fig.QBS}
\end{centering}
\end{figure}

The transition probabilities of the Markov chain corresponding to $\tilde{Q}_{BS}$ are given by:
\begin{itemize}
\item Service probability: $b_{BS}\left( \Gamma\right)=b_{12}\left( \Gamma\right)+kb_{11}\left( \Gamma\right)$.
\item Arrival probability when $\tilde{Q}_{BS}=0$: $a_{BS}^0\left( \Gamma\right)=a_{02}\left( \Gamma\right)+ka_{01}\left( \Gamma\right)$.
\item Arrival probability when $\tilde{Q}_{BS}>0$: $a_{BS}^1\left(\alpha, \Gamma\right)=a_{12}\left( \alpha,\Gamma\right)+ka_{11}\left( \alpha,\Gamma\right)$.\\
\end{itemize}

We see that these probabilities depend on the priority policy $\Gamma$, but for clarity reasons of the expressions we may use the notation $b_{BS}$, $a^0_{BS}$ and $a^1_{BS}$ instead of $b_{BS}\left( \Gamma\right)$, $a^0_{BS}\left( \Gamma\right)$ and $a^1_{BS}\left(\alpha, \Gamma\right)$. For the approximated Markov chain we consider the fraction vector $\tilde{\alpha}=\left[\tilde{\alpha}_1, \tilde{\alpha}_2, \tilde{\alpha}_3, \tilde{\alpha}_4\right]$.\\

\subsubsection*{\textbf{\textit{Step 1.b}}} \textit{Probability that $\tilde{Q}_{BS}$ is empty}\\

The probability that $\tilde{Q}_{BS}$ is empty is deduced from the stationary distribution $\tilde{\varPi}$ of the approximated Markov Chain of $\tilde{Q}_{BS}$ (figure \ref{fig.QBS}) which is computed by applying the balance equations
as follows:  
\[
\tilde{\varPi}_1=\frac{a_{BS}^0}{b_{BS}}\tilde{\varPi}_0
;
\tilde{\varPi}_2=\frac{a_{BS}^0a_{BS}^1}{b_{BS}}\tilde{\varPi}_0
\]

\[
\tilde{\varPi}_n=\frac{a_{BS}^0\left(a_{BS}^1\right)^{n-1}}{\left(b_{BS}\right)^{n}}\tilde{\varPi}_0
\]
Knowing that the stationary distribution verifies $\sum_{n=1}^{\infty}  \tilde{\varPi}_n=1$; hence the probability that $\tilde{Q}_{BS}$ is empty for a given priority policy $\Gamma$ is given by:
\begin{equation}
\tilde{\varPi}_0=\mathbf{\mathbb{P}}\left[\tilde{Q}_{BS}=0\right]
=\frac{b_{BS}-a_{BS}^1}{b_{BS}-a_{BS}^1+a_{BS}^0}
\label{eq.Pi_0_gen}
\end{equation}

\[
iff \,\, a_{BS}<b_{BS}.
\]

\iffalse
For the cellular scenario, we derive the following analytical expressions of the service rates for a given policy $\Gamma$:
\fi

\subsubsection*{\textbf{\textit{Step 2}}} \textit{Service and arrival probabilities of $\tilde{Q}_{BS}$}\\

In order to find the stability conditions, we need to find the expressions of transitions probabilities of the approximated Markov chain of $\tilde{Q}_{BS}$.

\subsubsection*{\textbf{\textit{Step 2.a}}} \textit{Service probability of $\tilde{Q}_{BS}$}\\

Using parameters $U,\,V$ and $N$ from table \ref{table.Notation} and based on equations (\ref{QBS_b1_gen}) and (\ref{QBS_b2_gen}), the service probability $b_{BS}$ can be written as it follows:

\begin{equation}
b_{BS}=kp_d^{1}\left(1-\tilde{\alpha}_{1}p_s^{1}-\tilde{\alpha}_{3}p_s^{2}\right)U+p_d^{2}\left(N-\tilde{\alpha}_{2}p_s^{1}U-\tilde{\alpha}_{4}p_s^2V\right)
\label{Pmu_BS_gen}
\end{equation}

\subsubsection*{\textbf{\textit{Step 2.b}}} \textit{Arrival probability of $\tilde{Q}_{BS}$}\\

We should differ between $a_{BS}^0$: arrival probability when $\tilde{Q}_{BS}$ is empty and $a_{BS}^1$: arrival probability when $\tilde{Q}_{BS}$ is not empty. Using parameters $U,\,V$ and $N$ from table \ref{table.Notation} and based on equations (\ref{QBS_a01_gen}), (\ref{QBS_a02_gen}), (\ref{QBS_a11_gen}) and (\ref{QBS_a12_gen}), the probabilities of arrival $a_{BS}^0$ and $a_{BS}^1$ of the approximated queue $\tilde{Q}_{BS}$ respectively when it is empty or not empty can be written as it follows:

\begin{equation}
\label{Pmu10_gen}
a_{BS}^0=kp_s^{1}U+p_s^{2}V
\end{equation}
\begin{equation}
\label{Pmu11_gen}
a_{BS}^1=kp_s^{1}\left(\tilde{\alpha}_{1}p_d^{1}+\tilde{\alpha}_{2}p_d^{2}+p_d^{3}\right)U+p_s^{2}\left(\tilde{\alpha}_{3}p_d^{1}U+\tilde{\alpha}_{4}p_d^{2}V+p_d^{3}V\right)
\end{equation}

The values of the parameter above $U$, $V$ and $N$ are given in table \ref{table.Policy-Cell} for all the considered priority policies $\Gamma$ that correspond to the corner point of the stability region ($\Gamma \in \Omega_{\Gamma}^{ss}$). Since the arrival probability $a_{BS}$ at the queue $\tilde{Q}_{BS}$ is given by: 
\[\mathbf{\mathbb{P}}\left[\tilde{Q}_{BS}=0\right]a_{BS}^{0}+\mathbf{\mathbb{P}}\left[\tilde{Q}_{BS}>0\right]a_{BS}^{1} = \tilde{\varPi}_0a_{BS}^{0}+\left( 1-\tilde{\varPi}_0\right)a_{BS}^{1}
\] 
\begin{equation}
\label{eq.a_bs_gen}
\Rightarrow a_{BS}=\frac{b_{BS}a_{BS}^0}{b_{BS}-a_{BS}^1+a_{BS}^0}
\end{equation}

\subsubsection*{\textbf{\textit{Step 3}}} \textit{Stability condition of $\tilde{Q}_{BS}$}\\

Supposing that the arrival and
service processes of $\tilde{Q}_{BS}$ are strictly stationary and ergodic
then the stability of $\tilde{Q}_{BS}$ can be determined using \textbf{Loyne's} criterion
which states that the queue is stable if and only if the average arrival
rate is strictly less than the average service rate. Then the stability
condition that the fraction vector $\tilde{\alpha}=\left(\tilde{\alpha}_{1},\tilde{\alpha}_{2},\tilde{\alpha}_{3},\tilde{\alpha}_{4}\right)$
should verify in order to satisfy the stability of the queue $\tilde{Q}_{BS}$
is given by:
\[
a_{BS}<b_{BS}\Leftrightarrow\frac{b_{BS}a_{BS}^0}{b_{BS}-a_{BS}^1+a_{BS}^0}<b_{BS} \Leftrightarrow a_{BS}^1<b_{BS}
\]
Since $a_{BS}^1$ and $b_{BS}$ are given by (\ref{Pmu11_gen}) and (\ref{Pmu_BS_gen}) then the queue $\tilde{Q}_{BS}$ is stable if the fraction vector $\tilde{\alpha}$ verifies (\ref{eq:QBSstabConstraint_gen}).

%%%%%%%%%%%%%%%%%%%%%%%%%GENERAL%%%%%%%%%%%%%%%%%%%%%%%%%%%%%%%%%%%%%%%%%%%%%%%%%
%\newpage
\subsection{Proof of theorem \ref{th3_gen} \label{B_gen}}

Here, we study the approximated model for the three-UEs scenario and we derive its corresponding stability region. For this aim, we proceed similarly to the proof in appendix \ref{stab_real_twoRate}. However, the first two steps for the approximated model (the modeling of the BS queue by an approximated Markov Chain as well as the computation of the probability that this queue is empty) are already proved in appendix \ref{A_gen}. Hence, based on these results, we pursue the procedure by the following: \textit{\textbf{step 1}} computes the service rate of both queues Q$_s$ and  Q$_u$, \textit{\textbf{step 2}} provides the optimal fraction vector that depends on the priority policy $\Gamma \in \Omega_\Gamma^{ss}$ and that achieves the corner point corresponding to each priority policy $\Gamma$, \textit{\textbf{step 3}} combines the previous results and provides the simple and explicit form of the approximated stability region and \textit{\textbf{step 4}}, similarly to the step 5 of the appendix \ref{stab_real_twoRate}, verifies that having a set of arrival rates within this stability region is equivalent to the stability of the approximated system of queues. Here, there are two main challenges: (i) finding the approximated Markov Chain model of the queue Q$_{BS}$ that will help us to achieve an explicit and simple form of the stability region and (ii) the computation of the optimal fraction vector that will reduce the the complexity of the problem. \\

\subsubsection*{\textbf{\textit{Step 1}}} \textit{Service rate of Q$_{s}$ and Q$_{u}$}\\

We compute the approximated stability region of the 3-UEs cellular scenario. For that, we compute the service rate of both \textit{UE2UE} and \textit{UE2BS communications} for a given priority policy $\Gamma$. We follow the procedure below, based on queuing theory analysis of the network capacity, in order to derive the performance study of cellular system.\\

\subsubsection*{\textbf{\textit{Step 1.a}}} \textit{Service rate of Q$_{s}$}\\

If $\tilde{Q}_{BS}$ is empty then the service rate of Q$_{s}$ is denoted by $\mu_s^0$ and by $\mu_s^1$ otherwise. Since the service rate of Q$_{s}$ is computed by 
\[\tilde{\mu}_{s}\left(\tilde{\alpha},\Gamma\right)=\mathbf{\mathbb{P}}\left[\tilde{Q}_{BS}=0\right]\mu_{s}^{0}\left(\Gamma\right)+\mathbf{\mathbb{P}}\left[\tilde{Q}_{BS}>0\right]\mu_{s}^{1}\left(\tilde{\alpha},\Gamma\right)
\]
\[
=\frac{\mu_{s}^{0}\left(\Gamma\right)\left(b_{BS}\left(\Gamma\right)-a_{BS}^1\left(\tilde{\alpha},\Gamma\right)\right)+\mu_{s}^{1}\left(\tilde{\alpha},\Gamma\right)a_{BS}^{0}\left(\Gamma\right)}{b_{BS}\left(\Gamma\right)-a_{BS}^{1}\left(\tilde{\alpha},\Gamma\right)+a_{BS}^{0}\left(\Gamma\right)}
 \]
\begin{equation}
\Rightarrow \tilde{\mu}_{s}=\frac{\mu_{s}^{0}\left(b_{BS-}a_{BS}^1\right)+\mu_{s}^{1}a_{BS}^{0}}{b_{BS}-a_{BS}^{1}+a_{BS}^{0}}\label{eq:Gmu1_gen}
\end{equation}
$\mu_{s}^{0}\left(\Gamma\right)$ and $\mu_{s}^{1}\left(\alpha,\Gamma\right)$ are respectively given by (\ref{mu10}) and( \ref{mu11}),  we obtain (\ref{eqn_mu1c_1_gen}) as the service rate \textbf{$\tilde{\mu}_{1}\left(\tilde{\alpha},\Gamma\right)$} of the queue Q$_s$.

\[
\resizebox{1\hsize}{!}{$\tilde{\mu}_{s}\left(\tilde{\alpha},\Gamma\right)=\frac{\left(r_{1}p_s^{1}U+r_{2}p_s^{2}V\right)\left(-2k\tilde{\alpha}_{1}p_s^{1}p_d^{1}U-\left(k+1\right)\tilde{\alpha}_{2}p_s^{1}p_d^{2}U-\left(k+1\right)\tilde{\alpha}_{3}p_s^{2}p_d^{1}U-\tilde{\alpha}_{4}p_s^{2}p_d^{2}V+kp_d^{1}+p_d^{2}N-\left(kp_s^{1}U+p_s^{2}V\right)p_d^{3}\right)}{-2k\tilde{\alpha}_{1}p_s^{1}p_d^{1}U-\left(k+1\right)\tilde{\alpha}_{2}p_s^{1}p_d{2}U-\left(k+1\right)\tilde{\alpha}_{3}p_s^{2}p_d^{1}U-2\tilde{\alpha}_{4}p_s^{2}p_d^{2}V+\left(kp_d^{1}+p_d^{2}N\right)+\left(kp_s^{1}U+p_s^{2}V\right)\left(1-p_d^{3}\right)}$}
\]
\[
\resizebox{1\hsize}{!}{$+\frac{\left(kp_s^{1}U+p_s^{2}V\right)\left(r_{1}\tilde{\alpha}_{1}p_s^{1}p_d^{1}W+r_{1}\tilde{\alpha}_{2}p_s^{1}p_d^{2}U+r_{2}\tilde{\alpha}_{3}p_s^{2}p_d^{1}U+r_{2}\tilde{\alpha}_{4}p_s^{2}p_d^{2}V+\left(r_{1}p_s^{1}U+r_{2}p_s^{2}V\right)p_d^{3}\right)}{-2k\tilde{\alpha}_{1}p_s^{1}p_d^{1}U-\left(k+1\right)\tilde{\alpha}_{2}p_s^{1}p_d^{2}U-\left(k+1\right)\tilde{\alpha}_{3}p_s^{2}p_d^{1}U-2\tilde{\alpha}_{4}p_s^{2}p_d^{2}V+\left(kp_d^{1}+p_d^{2}N\right)+\left(kp_s^{1}U+p_s^{2}V\right)\bar{p}_d^{3}}$}
\]
$=\left(r_{1}p_s^{1}U+r_{2}p_s^{2}V\right)+\left(kp_s^{1}U+p_s^{2}V\right)\times$
\[
\resizebox{1\hsize}{!}{$\frac{\left(-r_{1}\tilde{\alpha}_{1}p_s^{1}p_d^{1}U-r_{1}\tilde{\alpha}_{2}p_s^{1}p_d^{2}U-r_{2}\tilde{\alpha}_{3}p_s^{2}p_d^{1}U-r_{2}\tilde{\alpha}_{4}p_s^{2}p_d^{2}V+\left(r_{1}p_s^{1}U+r_{2}p_s^{2}V\right)\bar{p}_d^3\right)}{2k\tilde{\alpha}_{1}p_s^{1}p_d^{1}U+\left(k+1\right)\tilde{\alpha}_{2}p_s^{1}p_d^{2}U+\left(k+1\right)\tilde{\alpha}_{3}p_s^{2}p_d^{1}U+2\tilde{\alpha}_{4}p_s^{2}p_d^{2}V-\left(kp_d^{1}+p_d^{2}N\right)-\left(kp_s^{1}U+p_s^{2}V\right)\bar{p}_d^{3}}$}
\]
$=\left(r_{1}p_s^{1}U+r_{2}p_s^{2}V\right)+r_{2}\left(kp_s^{1}U+p_s^{2}V\right)\times$
\[
\resizebox{1\hsize}{!}{$\frac{2\left(-k\tilde{\alpha}_{1}p_s^{1}p_d^{1}U-k\tilde{\alpha}_{2}p_s^{1}p_d^{2}U-\tilde{\alpha}_{3}p_s^{2}p_d^{1}U-\tilde{\alpha}_{4}p_s^{2}p_d^{2}V+\left(kp_s^{1}U+p_s^{2}V\right)\bar{p}_d^3\right)}{2\left( 2k\tilde{\alpha}_{1}p_s^{1}p_d^{1}U+\left(k+1\right)\tilde{\alpha}_{2}p_s^{1}p_d^{2}U+\left(k+1\right)\tilde{\alpha}_{3}p_s^{2}p_d^{1}U+2\tilde{\alpha}_{4}p_s^{2}p_d^{2}V-\left(kp_d^{1}+p_d^{2}N\right)-\left(kp_s^{1}U+p_s^{2}V\right)\bar{p}_d^{3}\right)}$}
\]
$=\left(r_{1}p_s^{1}U+r_{2}p_s^{2}V\right)\times$
\[
\resizebox{1\hsize}{!}{$\left(\frac{1}{2}+\frac{\left(1-k\right)\tilde{\alpha}_2 p_s^{1}p_d^{2}U+\left(k-1\right)\tilde{\alpha}_{3}p_s^{2}p_d^{1}U-\left(kp_d^{1}+p_d^{2}N\right)+\left(kp_s^{1}U+p_s^{2}V\right)\bar{p}_d^3}{2\left(2k\tilde{\alpha}_{1}p_s^{1}p_d^{1}U+\left(k+1\right)\tilde{\alpha}_{2}p_s^{1}p_d^{2}U+\left(k+1\right)\tilde{\alpha}_{3}p_s^{2}p_d^{1}U+2\tilde{\alpha}_{4}p_s^{2}p_d^{2}V-\left(kp_d^{1}+p_d^{2}N\right)-\left(kp_s^{1}U+p_s^{2}V\right)\bar{p}_d^{3}\right)}\right)$}
\]
$=\frac{1}{2}\left(r_{1}p_s^{1}U+r_{2}p_s^{2}V\right)\times$
\[
\resizebox{1\hsize}{!}{$\left(1+\frac{\left(1-k\right)\tilde{\alpha}_2p_s^{1}p_d^{2}U+\left(k-1\right)\tilde{\alpha}_{3}p_s^{2}p_d^{1}U-\left(kp_d^{1}+p_d^{2}N\right)+\left(kp_s^{1}U+p_s^{2}V\right)\bar{p}_d^3}{2k\tilde{\alpha}_{1}p_s^{1}p_d^{1}U+\left(k+1\right)\tilde{\alpha}_{2}p_s^{1}p_d^{2}U+\left(k+1\right)\tilde{\alpha}_{3}p_s^{2}p_d^{1}U+2\tilde{\alpha}_{4}p_s^{2}p_d^{2}V-\left(kp_d^{1}+p_d^{2}N\right)-\left(kp_s^{1}U+p_s^{2}V\right)\bar{p}_d^{3}}\right)$}
\]

\subsubsection*{\textbf{\textit{Step 1.b}}} \textit{Service rate of Q$_{u}$}\\

If $\tilde{Q}_{BS}$ is empty then the service rate of Q$_{u}$ is denoted by $\mu_u^0$ otherwise it is denoted by $\mu_u^1$. By analogy to the Q$_s$, the service rate of Q$_{u}$ is computed as it follows: 

\[\tilde{\mu}_{u}\left(\tilde{\alpha},\Gamma\right)=\mathbf{\mathbb{P}}\left[\tilde{Q}_{BS}=0\right]\mu_{u}^{0}\left(\Gamma\right)+\mathbf{\mathbb{P}}\left[\tilde{Q}_{BS}>0\right]\mu_{u}^{1}\left(\tilde{\alpha},\Gamma\right)
\]
\[
=\frac{\mu_{u}^{0}\left(\Gamma\right)\left(b_{BS}\left(\Gamma\right)-a_{BS}^1\left(\tilde{\alpha},\Gamma\right)\right)+\mu_{u}^{1}\left(\tilde{\alpha},\Gamma\right)a_{BS}^{0}\left(\Gamma\right)}{b_{BS}\left(\Gamma\right)-a_{BS}^{1}\left(\tilde{\alpha},\Gamma\right)+a_{BS}^{0}\left(\Gamma\right)}
 \]
\begin{equation}
\Rightarrow \tilde{\mu}_{u}=\frac{\mu_u^{0}\left(b_{BS-}a_{BS}^1\right)+\mu_{u}^{1}a_{BS}^{0}}{b_{BS}-a_{BS}^{1}+a_{BS}^{0}}\label{eq:Gmu2_gen}
\end{equation}
With $\mu_{u}^{0}\left(\Gamma\right)$ and $\mu_{u}^{1}\left(\alpha,\Gamma\right)$  respectively given by (\ref{mu20}) and( \ref{mu21}),  we obtain (\ref{eqn_mu2c_1_gen}) as the service rate \textbf{$\tilde{\mu}_{2}\left(\tilde{\alpha},\Gamma\right)$} of the queue Q$_u$.

% ensure that we have normalsize text

\[
\resizebox{1\hsize}{!}{$\tilde{\mu}_{u}\left(\tilde{\alpha},\Gamma\right)=\frac{\left(r_{1}p_u^{1}W+r_{2}p_u^{2}X\right)\left(-2k\tilde{\alpha}_{1}p_s^{1}p_d^{1}U-\left(k+1\right)\tilde{\alpha}_{2}p_s^{1}p_d^{2}U-\left(k+1\right)\tilde{\alpha}_{3}p_s^{2}p_d^{1}U-2\tilde{\alpha}_{4}p_s^{2}p_d^{2}V+kp_d^{1}+p_d^{2}N-\left(kp_s^{1}U+p_s^{2}V\right)p_d^{3}\right)}{-2k\tilde{\alpha}_{1}p_s^{1}p_d^{1}U-\left(k+1\right)\tilde{\alpha}_{2}p_s^{1}p_d^{2}U-\left(k+1\right)\tilde{\alpha}_{3}p_s^{2}p_d^{1}U-2\tilde{\alpha}_{4}p_s^{2}p_d^{2}V+\left(kp_d^{1}+p_d^{2}N\right)+\left(kp_s^{1}U+p_s^{2}V\right)\bar{p}_d^{3}}$}
\]
\[
\resizebox{1\hsize}{!}{$+\frac{\left(r_{1}p_u^{1}Y+r_{2}p_u^{2}Z\right)\left(kp_s^{1}U+p_s^{2}V\right)}{-2k\tilde{\alpha}_{1}p_s^{1}p_d^{1}U-\left(k+1\right)\tilde{\alpha}_{2}p_s^{1}p_d^{2}U-\left(k+1\right)\tilde{\alpha}_{3}p_s^{2}p_d^{1}U-2\tilde{\alpha}_{4}p_s^{2}p_d^{2}V+\left(kp_d^{1}+p_d^{2}N\right)+\left(kp_s^{1}U+p_s^{2}V\right)\bar{p}_d^{3}}$}
\]
$=\left(r_{1}p_u^{1}W+r_{2}p_u^{2}X\right)+$
\[
\resizebox{1\hsize}{!}{$\frac{\left(r_{1}p_u^{1}W+r_{2}p_u^{2}X\right)\left(-kp_s^{1}U+p_s^{2}V\right)+\left(r_{1}p_u^{1}Y+r_{2}p_u^{2}Z\right)\left(kp_s^{1}U+p_s^{2}V\right)}{-2k\tilde{\alpha}_{1}p_s^{1}p_d^{1}U-\left(k+1\right)\tilde{\alpha}_{2}p_s^{1}p_d^{2}U-\left(k+1\right)\tilde{\alpha}_{3}p_s^{2}p_d^{1}U-2\tilde{\alpha}_{4}p_s^{2}p_d^{2}V+\left(kp_d^{1}+p_d^{2}N\right)+\left(kp_s^{1}U+p_s^{2}V\right)\bar{p}_d^{3}}$}
\]
$=\left(r_{1}p_u^{1}W+r_{2}p_u^{2}X\right)$
\[
\resizebox{1\hsize}{!}{$+\frac{\left(r_{1}p_u^{1}\left(W-Y\right)+r_{2}p_u^{2}\left(X-Z\right)\right)\left(kp_s^{1}U+p_s^{2}V\right)}{2k\tilde{\alpha}_{1}p_s^{1}p_d^{1}U+\left(k+1\right)\tilde{\alpha}_{2}p_s^{1}p_d^{2}U+\left(k+1\right)\tilde{\alpha}_{3}p_s^{2}p_d^{1}U+2\tilde{\alpha}_{4}p_s^{2}p_d^{2}V-\left(kp_d^{1}+p_d^{2}N\right)-\left(kp_s^{1}U+p_s^{2}V\right)\bar{p}_d^{3}}$}
\]

$\tilde{\mu}_s\left(\tilde{\alpha},\Gamma\right)=\frac{1}{2}\left(r_{1}p_s^{1}U+r_{2}p_s^{2}V\right)\times$
\begin{equation}
\label{eqn_mu1c_1_gen}
\resizebox{0.9\hsize}{!}{$\left(1+\frac{\left(1-k\right)\tilde{\alpha}_2 p_s^{1}p_d^{2}U+\left(k-1\right)\tilde{\alpha}_{3}p_s^{2}p_d^{1}U-\left(kp_d^{1}+p_d^{2}N\right)+\left(kp_s^{1}U+p_s^{2}V\right)\bar{p}_d^3}{2k\tilde{\alpha}_{1}p_s^{1}p_d^{1}U+\left(k+1\right)\tilde{\alpha}_{2}p_s^{1}p_d^{2}U+\left(k+1\right)\tilde{\alpha}_{3}p_s^{2}p_d^{1}U+2\tilde{\alpha}_{4}p_s^{2}p_d^{2}V-\left(kp_d^{1}+p_d^{2}N\right)-\left(kp_s^{1}U+p_s^{2}V\right)\bar{p}_d^{3}}\right)$}
\end{equation}
$\tilde{\mu}_{u}\left(\tilde{\alpha},\Gamma\right)=\left(r_{1}p_u^{1}W+r_{2}p_u^{2}X\right)+$
\begin{equation}
\label{eqn_mu2c_1_gen}
\resizebox{0.9\hsize}{!}{$\frac{\left(r_{1}p_u^{1}\left(W-Y\right)+r_{2}p_u^{2}\left(X-Z\right)\right)\left(kp_s^{1}U+p_s^{2}V\right)}{2k\tilde{\alpha}_{1}p_s^{1}p_d^{1}U+\left(k+1\right)\tilde{\alpha}_{2}p_s^{1}p_d^{2}U+\left(k+1\right)\tilde{\alpha}_{3}p_s^{2}p_d^{1}U+2\tilde{\alpha}_{4}p_s^{2}p_d^{2}V-\left(kp_d^{1}+p_d^{2}N\right)-\left(kp_s^{1}U+p_s^{2}V\right)\bar{p}_d^{3}}$}
\end{equation}

%\hrulefill
% The spacer can be tweaked to stop underfull vboxes.

\subsubsection*{\textbf{\textit{Step 2}}} \textit{The optimal fraction vector $\alpha^*$}\\

 We know that finding the corner points of approximated stability region are sufficient to characterize this region. Hence, we intend to reduce more the complexity by avoiding the consideration of all the  $\tilde{\alpha} \in \left[0,1\right]^4$ and finding the set of optimum $\alpha^*$ that achieve the corner points of the approximated stability region. So, for each priority policy $\Gamma \in \Omega_{\Gamma}^{ss}$, we compute the optimum $\tilde{\alpha}^*=\left( \tilde{\alpha}^*_{1},\tilde{\alpha}^*_{2},\tilde{\alpha}^*_{3},\alpha^*_{4}\right)$ that maximizes $\mu_s$ and $\mu_u$ such a border point of the stability region is achieved. In other terms, finding $\tilde{\alpha}^*$ for each priority policy $\Gamma \in \Omega_{\Gamma}^{ss} $ allows us to find explicitly the corresponding corner point and to avoid the need to vary $\tilde{\alpha} \in \left[0,1\right]^4$ for obtaining this point of the stability region.\\

We start by the service rate $\tilde{\mu}_{u}$. We can see that for all the priority policies, $\tilde{\mu}_{u}\left(\tilde{\alpha},\Gamma\right)$ is inversely proportional to $\tilde{\alpha}_{1},\tilde{\alpha}_{2},\tilde{\alpha}_{3}$ and $\tilde{\alpha}_{4}$ then $\tilde{\alpha}^{*}$=$\left(0,0,0,0\right)$ maximizes $\tilde{\mu}_{u}\left(\tilde{\alpha}, \Gamma\right)$ for all $\Gamma \in \Omega_{\Gamma}^{ss}$ . It remains to maximize $\tilde{\mu}_{s}$ by solving the following optimization problem for each priority policy:

\[
\max \limits_{\tilde{\alpha}}\:\tilde{\mu}_{s}\left(\tilde{\alpha},\Gamma\right)
\]
\begin{equation}
s.t.\:C_{0}=\begin{cases}
C_{01}:\,\left( \ref{eq:QBSstabConstraint_gen} \right) is\,verified\\
C_{02}:\,0\leq\tilde{\alpha}_{1}\leq1,\:0\leq\tilde{\alpha}_{2}\leq1\\
\:0\leq\tilde{\alpha}_{3}\leq1\,and\,0\leq\tilde{\alpha}_{4}\leq1
\end{cases}\label{eq:constraint0_gen}
\end{equation}

The optimization problem above is a linear-fractional programming (LFP) where the objective function is the ratio of two linear functions with a set of linear constraints. This LFP has a bounded and non-empty feasible region and can be transformed into an equivalent linear problem (LP)  which the solution (solved by any LP solution method such that simplex algorithm) yields that of the original LFP problem.

The solution of the optimization problem above belongs to the scenarios where at least one of the two constraints ($C_{01}$,$C_{02}$ ) is reached which corresponds to different cases:  (i) constraint $C_{02}$ is achieved ($\left(\tilde{\alpha}_1,\tilde{\alpha}_2,\tilde{\alpha}_3,\tilde{\alpha}_4\right) \in \lbrace 0,1\rbrace^4$) while the stability condition $C_{01}$ is verified without equality and (ii) the stability condition $C_{01}$ is achieved with $\left(\tilde{\alpha}_1,\tilde{\alpha}_2,\tilde{\alpha}_3,\tilde{\alpha}_4\right) \in \left[ 0,1\right]^4$.

These cases will be studied in order to come up with the set, called $\mathbb{S}_\alpha$, of potential solutions of the optimization problem (\ref{eq:constraint0_gen}) . Several cases will be detailed where from each one we can deduce an element of the set $\mathbb{S}_\alpha$. \\

\textbf{\textit{Case 1:}} For the first case, we start by supposing that at the optimal solution the constraint $C_{01}$ is verified but not achieved while the second constraint $C_{02}$ is reached. The optimum in this case is the combination of $\tilde{\alpha}=\tilde{\alpha}_{1},\tilde{\alpha}_{2},\tilde{\alpha}_{3},\tilde{\alpha}_{4} \in \left\{ 0,1\right\} ^{4}$ that maximizes $\tilde{\mu}_{s}\left(\tilde{\alpha},\Gamma\right)$ while verifying the stability condition $C_{01}$.\\

\textbf{\textit{Case 2:}} For the second case, we suppose that the stability condition $C_{01}$ is achieved at the optimum. Hence, the optimization problem is converted to the problem (\ref{eq:constraint1_gen}) with $F_{1}=M$ and $F_{2}=F_{1}-2kp_s^{1}p_d^{1}U-2p_s^{2}p_d^{2}V$. Recall that $M=\left(k p_d^{1}+p_d^{2}N\right)-\left(k p_s^{1}U+p_s^{2}V\right)p_d^{3}$.\\

 By analogy to (\ref{eq:constraint0_gen}), we solve the problem (\ref{eq:constraint1_gen}) by considering that at least one of the linear constraints is reached which corresponds to the two following cases: (i) the first one assumes that $C_{12}$ constraint is achieved and (ii) the other one considers that $C_{11}$ constraint is achieved. We study both cases then the optimum corresponds to the case that maximizes $\tilde{\mu}_{s}\left(\tilde{\alpha},\Gamma\right)$.\\ 

We start by supposing that at the optimal solution the constraint $C_{11}$ is not achieved then the optimum in this case is the combination of $\tilde{\alpha}_{2},\tilde{\alpha}_{3} \in \left\{ 0,1\right\} ^{2}$ that maximizes $\tilde{\mu}_{s}\left(\tilde{\alpha},\Gamma\right)$ and for which it exists $\tilde{\alpha}_1$ and $\tilde{\alpha}_4 \in \left[ 0,1\right]$ that verify (\ref{eq:constraint1_gen}). \\

\begin{figure*}
% ensure that we have normalsize text
\normalsize
\resizebox{0.9\hsize}{!}{$
\max\limits_{\tilde{\alpha}_2,\tilde{\alpha}_3}\:\frac{1}{2}\left(r_{1}p_s^{1}U+r_{2}p_s^{2}V\right)\left(1+\frac{\left( k-1\right)\tilde{\alpha}_2p_s^{1}p_d^{2}U+\left(1-k\right)\tilde{\alpha}_{3}p_s^{2}p_d^{1}U+\left(k p_d^{1}+p_d^{2}N\right)-\left(k p_s^{1}U+p_s^{2}V\right)\bar{p}_d^{3}}{kp_s^{1}U+p_s^{2}V}\right)
$}
\begin{equation}
s.t.\:C_{1}=\begin{cases}
C_{11}: F_{2}\leq \left(k+1\right)\tilde{\alpha}_{2}p_s^{1}p_d^{2}U+\left(k+1\right)\tilde{\alpha}_{3}p_s^{2}p_d^{1}U\leq F_{1}\\
C_{12}: 0\leq \tilde{\alpha}_{2}\leq1\:and\,0\leq\tilde{\alpha}_{3}\leq1
\end{cases}\label{eq:constraint1_gen}
\end{equation}
% Restore the current equation number.
%\setcounter{equation}{\value{MYtempeqncnt}}
% IEEE uses as a separator
\hrulefill
% The spacer can be tweaked to stop underfull vboxes.
\vspace*{4pt}
\end{figure*}

\textbf{\textit{Case 3:}} For this case, we suppose that the constraint $C_{11}$ is active at the optimum by achieving either the lower bound $F_{1}$ (for $\tilde{\alpha}_1=\tilde{\alpha}_4=0$) or the upper bound $F_{2}$ (for $\alpha_1=\alpha_4=1$). For each bound $F_1$ and $F_2$ respectively, $\tilde{\alpha}_2$ and $\tilde{\alpha}_3$ are chosen in a way to verify respectively the following equations: $\left(k+1\right)\tilde{\alpha}_{2}p_s^{1}p_d^{2}U+\left(k+1\right)\tilde{\alpha}_{3}p_s^{2}p_d^{1}U=F_i$ for $i=1,2$. For this case, we suppose that the constraint $C_{11}$ is active at the optimum then the optimization problem can be written as in (\ref{eq:constraint2_gen}).\\

The LFP problem (\ref{eq:constraint2_gen}) depends only from $\tilde{\alpha}_2$ where the objective function is linearly depending on $\tilde{\alpha}_2$ with a positive coefficient. Hence, the objective function is maximized by increasing as much as possible $\tilde{\alpha}_2$ such that the constraints (\ref{eq:constraint2_gen}) are verified. \\

Combining the results of the three cases above, for each priority policy $\Gamma$, we find the set $\mathbb{S}_\alpha$ to which corresponds the optimum $\alpha^*$ that maximizes $\tilde{\mu}_s\left(\tilde{\alpha},\Gamma\right)$ and achieves the corresponding corner point of the stability region. Hence, we hugely reduce the complexity by finding the explicit value of $ \alpha^*$ that gives the corner points that are sufficient to characterize the approximated stability region. \\

\begin{figure*}
% ensure that we have normalsize text
\normalsize
\[
\max\limits_{\tilde{\alpha}_2}\:\frac{1}{2}\left(r_{1}p_s^{1}U+r_{2}p_s^{2}V\right)\left(1+\frac{2\left(k-1\right)\tilde{\alpha}_2p_s^{1}p_d^{2}U+\frac{\left(1-k\right)}{\left(1+k\right)}F_i+\left(k p_d^{1}+p_d^{2}N\right)-\left(k p_s^{1}U+p_s^{2}V\right)\bar{p}_d^{3}}{kp_s^{1}U+p_s^{2}V}\right)
\]
\begin{equation}
s.t.\:C_{1}=\begin{cases}
C_{21}: \left(k+1\right)\tilde{\alpha}_{2}p_s^{1}p_d^{2}U+\left(k+1\right)\tilde{\alpha}_3p_s^{2}p_d^{1}U= F_{i} \\
C_{22}: 0\leq \tilde{\alpha}_{2}\leq1 \\
C_{23}: \,\,for\,\, F_1:\,\, \tilde{\alpha}_1=\tilde{\alpha}_4=0 \,\,;\,\, for\,\, F_2:\,\, \tilde{\alpha}_1=\tilde{\alpha}_4=1
\end{cases}\label{eq:constraint2_gen}
\end{equation}
% Restore the current equation number.
%\setcounter{equation}{\value{MYtempeqncnt}}
% IEEE uses as a separator
\hrulefill
% The spacer can be tweaked to stop underfull vboxes.
\vspace*{4pt}
\end{figure*}

\[\mathbb{S}_{\alpha}=
\left(\begin{array}{c}
\tilde{\alpha}_1=0,\tilde{\alpha}_2=1, \tilde{\alpha}_3=\frac{F_1-\left(k+1\right)p_s^{1}p_d^{2}}{\left(k+1\right)p_d^{1}p_s^{2}U},\tilde{\alpha}_4=0 \\
\hspace{5pt} \\
\tilde{\alpha}_1=1, \tilde{\alpha}_2=1, \tilde{\alpha}_3=\frac{F_2-\left(k+1\right)p_s^{1}p_d^{2}U}{\left(k+1\right)p_d^{1}p_s^{2}U},\tilde{\alpha}_4=1 \\
\hspace{5pt} \\
\tilde{\alpha}_1=0, \tilde{\alpha}_2=\frac{F_1}{\left(k+1\right)p_s^{1}p_d^{2}U}, \tilde{\alpha}_3=0,\tilde{\alpha}_4=0 \\
\hspace{5pt} \\
\tilde{\alpha}_1=0, \tilde{\alpha}_2=\frac{F_1-\left(k+1\right)p_d^{1}p_s^{2}U}{\left(k+1\right)p_s^{1}p_d^{2}U}, \tilde{\alpha}_3=1,\tilde{\alpha}_4=0 \\
\hspace{5pt} \\
\tilde{\alpha}_1=1, \tilde{\alpha}_2=\frac{F_2}{\left(k+1\right)p_s^{1}p_d^{2}U}, \tilde{\alpha}_3=0,\alpha_4=1\\
\hspace{5pt} \\
\tilde{\alpha}_1=1, \tilde{\alpha}_2=\frac{F_2-\left(k+1\right)p_d^{1}p_s^{2}U}{\left(k+1\right)p_s^{1}p_d^{2}U}, \tilde{\alpha}_3=1,\tilde{\alpha}_4=1 
\end{array}\right)
\]

Finally, for each priority policy within $\Gamma \in  \Omega_{\Gamma}^{ss}$, the optimum $\tilde{\alpha}^{*}=\left(\tilde{\alpha}_{1}^{*},\tilde{\alpha}_{2}^{*},\tilde{\alpha}_{3}^{*},\tilde{\alpha}_{4}^{*}\right)$, that reaches the corner point of the approximated stability region corresponding to $\Gamma$, to the finite set (\ref{Set_alpha_gen}) denoted by $\mathbb{S}_{\alpha}$.\\

\subsubsection*{\textbf{\textit{Step 3}}} \textit{Characterization of the approximated stability region}

Combining the results of the previous steps, we deduce the approximation of the stability region for the three-UEs scenario. Supposing that the arrival and service processes of $Q_{s}$ and $Q_{u}$ are strictly stationary and ergodic then the stability of $\tilde{Q}_{BS}$ can be determined using \textbf{Loyne's} criterion which states that the queue is stable if and only if the average arrival
rate is strictly less than the average service rate.

We proceed as it follows for characterizing the approximated stability region. We start by considering a priority policy that corresponds to a corner point of the region ($\Gamma \in  \Omega^{ss}_{\Gamma}$ with $ |\Omega^{ss}_{\Gamma}|=6$). Then, for the considered $\Gamma$, we find the corresponding optimum fraction vector $\alpha^*$ within the finite set (\ref{Set_alpha_gen}), . We deduce the service rate of the system of queues. Finally, these steps are applied over all the priority policies $\Gamma \in  \Omega^{ss}_{\Gamma}$. Thus, the approximated stability region $\tilde{\mathcal{R}}_c^{ss}$ for the 3-UEs scenario is characterized by the set of the mean arrival rates $\left(\lambda_s,\lambda_u\right)$ in $\tilde{\mathcal{R}}_{c}^{ss}$ such that:
\[
\tilde{\mathcal{R}}_{c}^{ss}=co \left( \bigcup\limits_{\Gamma \in \Omega^{ss}_{\Gamma}}\lbrace  \tilde{ \mu}_s\left(\tilde{\alpha}^*,\Gamma\right) , \tilde{ \mu}_u\left(\tilde{\alpha}^*,\Gamma\right)  \rbrace \right)
\] 
where $\tilde{\mu}_{s}\left(\tilde{\alpha}^*,\Gamma\right)$ and $\tilde{\mu}_{u}\left(\tilde{\alpha}^*,\Gamma\right)$ are respectively given by (\ref{eqn_mu1c_gen}) and (\ref{eqn_mu2c_gen}).

\vspace{5pt}
\subsubsection*{\textbf{\textit{Step 4}}} \textit{Arrival rates within the stability region is equivalent to the stability of the system of queues} \\

Proving that having a set of mean arrival rates $\left(\lambda_s,\lambda_u\right)$ in $\tilde{\mathcal{R}}_{c}^{ss}$  is equivalent to the stability of the system queue is identical to the demonstration done in step 5 of the appendix \ref{stab_real_twoRate}. For brevity, we avoid repetition and remove this proof.
%%%%%%%%%%%%%%%%%%%%%%%%%%%%%%%%%%%%%%%%%%%%%%%%%%%%%%%%%%%%%%%%%%%%%%%%%%%%%%%%

\vspace{-5pt}
%\newpage
\subsection{Proof of theorem \ref{th_comparison_gen} \label{th_comp_gen}}

Here, we compare between the real and approximated stability region for the three-UEs scenario. For this aim, we proceed similarly to the proof in appendix \ref{stab_real_twoRate}. For a given priority policy $\Gamma \in \Omega_\Gamma^{ss}$, we proceed as it follows : in \textit{\textbf{step 1}}, we compare between the two Markov Chain model by computing the absolute error between the two values $\varPi_0$ and $\tilde{\varPi}_0$ that respectively correspond to the probabilities that Q$_{BS}$ and $\tilde{Q}_{BS}$ are empty and \textit{\textbf{step 2}} provides the absolute error between the service rates $\mu_s$ and $\tilde{\mu}_s$ in order to veridy that the approximation is an upper bound of the stability region. In \textit{\textbf{step 3}} we deduce the distance between these two regions based on the maximum relative error between their six corner points that correspond to the set of the border priority policies $\Omega_\Gamma^{ss}$. Finally, in \textit{\textbf{step 4}} we find the upper and the lower bounds of the stability region as function of its approximation.   \\

\subsubsection*{\textbf{\textit{Step 1}}} \textit{Comparison between the Markov chains of Q$_{BS}$ and $\tilde{Q}_{BS}$}

We start by presenting the used notation for facilitating the readability of the expressions:

\begin{itemize}
\item Let us keep in mind that the comparison between the two Markov chain models (exact and approximation) is done for each priority policy $\Gamma \in \Omega_\Gamma^{ss}$, hence to improve the clarity of the expressions we use the following: $a_{01}=a_{01}\left( \alpha,\Gamma\right)$, $a_{02}=a_{02}\left( \alpha,\Gamma\right)$, $a_{11}=a_{11}\left( \alpha,\Gamma\right)$, $a_{12}=a_{12}\left( \alpha,\Gamma\right)$, $b_{11}=b_{11}\left( \alpha,\Gamma\right)$, $b_{12}=b_{12}\left( \alpha,\Gamma\right)$\\
\item $a_0=a_{01}+a_{02}$, $a_1=a_{11}+a_{12}$, $b_1=b_{11}+b_{12}$\\
\item $\varPi_n=\varPi_n\left( \alpha,\Gamma\right)$ probability that Q$_{BS}$ is at state $n$ of the Markov chain which means that the queue contains n P-packets\\
\item $\tilde{\varPi_n}=\tilde{\varPi}_n\left( \alpha,\Gamma\right)$ probability that $\tilde{Q}_{BS}$ is at state $n$ of the approximated Markov chain which means that the approximated queue $\tilde{Q}_{BS}$ contains n P-packets.\\
\end{itemize} 

The main challenges of the exact stability region computation lies on the following to facts: (i) this region has no closed form expression since the final solution requires solving a linear system and (ii) the complexity of the exact stability region is high. Therefore, the advantage of the proposed approximation is to have a close upper bound that has an simple and explicit solution. 

The Markov Chain modeling is restrictively used for computing the probability that the corresponding queues are empty. Thus, the comparison between the exact and the approximated Markov Chains is done by finding the absolute error between the probability that Q$_{BS}$ is empty ($\varPi_0$) and the probability that $\tilde{Q}_{BS}$ is empty ( ($\tilde{\varPi}_0$) ). For a given priority policy $\Gamma$ and a given fraction vector $\alpha$ this error is defined as it follows: 
\[
\delta_{\varPi}\left(\alpha,\Gamma\right):=\varPi_0\left(\alpha,\Gamma\right)-\tilde{\varPi}_0\left(\alpha,\Gamma\right)
\]

Let us consider the Markov chain of the queue Q$_{BS}$ given in figure (\ref{fig.QBS_real_gen}) and its corresponding balance equations given by (\ref{eq.BE1_gen}), (\ref{eq.BE2_gen}), (\ref{eq.BE2tok_gen}) and (\ref{eq.BalanceEquation_gen}). The generating function corresponding to this Markov chain is given by:
\[
G\left(z\right)=\sum_{n=0}^{\infty}\varPi_nz^n
\] 
We deduce $G\left(z\right)$ as function of $\varPi_0$ and $\varPi_1$ by multiplying the $n^{th}$ balance equation by $z^n$ and by adding all these equations ($for \, n=0,...,\infty$) as it follows:

\[
\varPi_{0}a_{0}+\sum_{n=1}^{\infty}\varPi_{n}\left(a_{1}+b_{12}\right)z^{n}+b_{11}\sum_{n=k}^{\infty}\varPi_{n}z^{n}
\]\[
=a_{11}\sum_{n=k+1}^{\infty}\varPi_{n-k}z^{n}+a_{12}\sum_{n=2}^{\infty}\varPi_{n-1}z^{n}+b_{12}\sum_{n=0}^{\infty}\varPi_{n+1}z^{n}+b_{11}\sum_{n=0}^{\infty}\varPi_{n+k}z^{n}+a_{02}\varPi_{0}z+a_{01}\varPi_{0}z^{k}
\]

\[
\Rightarrow \sum_{n=0}^{\infty}\varPi_{n}\left(a_{1}+b_{1}\right)z^{n}+\varPi_{0}\left(a_{0}-a_{1}-b_{1}\right)-b_{11}\sum_{n=1}^{k-1}\varPi_{n}z^{n}
\]\[
=a_{11}\sum_{n=1}^{\infty}\varPi_{n}z^{n+k}+a_{12}\sum_{n=1}^{\infty}\varPi_{n}z^{n+1}+b_{12}\sum_{n=1}^{\infty}\varPi_{n}z^{n-1}+b_{11}\sum_{n=k}^{\infty}\varPi_{n}z^{n-k}+a_{02}\varPi_{0}z+a_{01}\varPi_{0}z^{k}
\]

\[
\Rightarrow \left(a_{1}+b_{1}\right)G\left(z\right)+\varPi_{0}\left(a_{0}-a_{1}-b_{1}-a_{02}z-a_{01}z^{k}\right)-b_{11}\sum_{n=1}^{k-1}\varPi_{n}z^{n}
\]\[
=a_{11}z^{k}G\left(z\right)-a_{11}\varPi_{0}z^{\text{k}}+a_{12}zG\left(z\right)-a_{12}\varPi_{0}z+b_{12}z^{-1}G\left(z\right)-b_{12}\varPi_{0}z^{-1}+b_{11}z^{-k}G\left(z\right)-b_{11}\sum_{n=0}^{k-1}\varPi_{n}z^{n-k}
\]

\[
\Rightarrow G\left(z\right)\left(a_{1}+b_{1}-a_{11}z^{k}-a_{12}z-b_{12}z^{-1}-b_{11}z^{-k}\right)
\]\[
=\varPi_{0}\left(-a_{0}+a_{1}+b_{1}+a_{02}z+a_{01}z^{k}-a_{11}z^{\text{k}}-a_{12}z-b_{12}z^{-1}-b_{11}z^{-k}\right)+b_{11}\left(1-z^{-k}\right)\sum_{n=1}^{k-1}\varPi_{n}z^{n}
\]
 
\[
\Rightarrow G\left(z\right)\left(a_{11}z^{2k}+a_{12}z^{k+1}-\left(a_{1}+b_{1}\right)z^{k}+b_{12}z^{k-1}+b_{11}\right)
\]\[
=\varPi_{0}\left(\left(a_{11}-a_{01}\right)z^{\text{2k}}+\left(a_{12}-a_{02}\right)z^{k+1}+\left(a_{0}-a_{1}-b_{1}\right)z^{k}+b_{12}z^{k-1}+b_{11}\right)-b_{11}\left(z^{k}-1\right)\sum_{n=1}^{k-1}\varPi_{n}z^{n}
\]

\[
\Rightarrow G\left(z\right)\left(a_{11}\sum_{n=2k-1}^{k+1}z^{n}+a_{1}z^{k}-b_{1}z^{k-1}-b_{11}\sum_{n=k-2}^{0}z^{n}\right)
\]\[
=\varPi_{0}\left(\left(a_{11}-a_{01}\right)\sum_{n=2k-1}^{k+1}z^{n}+\left(a_{1}-a_{0}\right)z^{k}-b_{1}z^{k-1}-b_{11}\sum_{n=k-2}^{0}z^{n}\right)-b_{11}\left(\sum_{i=k-1}^{0}z^{i}\right)\sum_{n=1}^{k-1}\varPi_{n}z^{n}
\]

Knowing that $G\left(1\right)=1$ due to the normalization constraint then we deduce $\varPi_0$ as function of $\tilde{\varPi}_0$:

\[
a_{11}\left(k-1\right)+a_{1}-b_{1}-b_{11}\left(k-1\right)=\varPi_{0}\left(\left(a_{11}-a_{01}-b_{11}\right)\left(k-1\right)+a_{1}-a_{0}-b_{1}\right)-kb_{11}\sum_{n=1}^{k-1}\varPi_{n}z^{n}
\]

\[
\Rightarrow ka_{11}+a_{12}-b_{12}-kb_{11}=\varPi_{0}\left(a_{12}+ka_{11}-a_{02}-ka_{01}-b_{12}-kb_{11} \right)-kb_{11}\sum_{n=1}^{k-1}\varPi_{n}
\]

\[
\Rightarrow \varPi_{0}=\frac{ka_{11}+a_{12}-b_{12}-kb_{11}}{a_{12}+ka_{11}-a_{02}-ka_{01}-b_{12}-kb_{11}}+\frac{kb_{11}\sum\limits_{n=1}^{k-1}\varPi_{n}}{a_{12}+ka_{11}-a_{02}-ka_{01}-b_{12}-kb_{11}}
\]

\[
\Rightarrow \varPi_{0}=\frac{b_{12}+kb_{11}-a_{12}-ka_{11}}{b_{12}+kb_{11}-a_{12}-ka_{11}+a_{02}+ka_{01}}-\frac{kb_{11}\sum\limits_{n=1}^{k-1}\varPi_{n}}{b_{12}+kb_{11}-a_{12}-ka_{11}+a_{02}+ka_{01}}
\]

$\tilde{\varPi}_0$ is given by (\ref{eq.Pi_0_gen}) therefore we can write the following:
\[
\tilde{\varPi}_0=\frac{b_{1}+\left(k-1\right)b_{11}-a_{1}-\left(k-1\right)a_{11}}{b_{1}+\left(k-1\right)b_{11}-a_{1}-\left(k-1\right)a_{11}+a_{0}+\left(k-1\right)a_{01}}
\]

\[
=\frac{b_{12}+kb_{11}-a_{12}-ka_{11}}{b_{12}+kb_{11}-a_{12}-ka_{11}+a_{02}+ka_{01}}
\]
We deduce $\delta_{\varPi}\left( \alpha,\Gamma\right)$ as it follows:

\begin{equation}
\label{eq.Pi_0_comparison_gen}
\delta_{\varPi}\left( \alpha,\Gamma\right)=\varPi_0\left( \alpha,\Gamma\right)-\tilde{\varPi}_0\left( \alpha,\Gamma\right)=-\frac{kb_{11}\sum_{n=1}^{k-1}\varPi_{n}\left( \alpha,\Gamma\right)}{b_{12}+kb_{11}-a_{12}-ka_{11}+a_{02}+ka_{01}}
\end{equation}

\subsubsection*{\textbf{\textit{Step 2}}} \textit{$\tilde{\mathcal{R}}_c^{ss}$ an upper bound of ${\mathcal{R}}_c^{ss}$}\\

We prove that this approximation is a close upper bound for the stability region of the 3-UEs cellular scenario by finding the absolute difference between ${\mu}_{s}\left( {\alpha},\Gamma\right)$  and $\tilde{\mu}_{s}\left( {\alpha},\Gamma\right)$ as defined below:
\[
\delta_\mu\left( {\alpha},\Gamma\right)=\tilde{\mu}_{s}\left( {\alpha},\Gamma\right) -{\mu}_{s}\left( {\alpha},\Gamma\right)
\]

We know that ${\tilde{\mu}}_{s}\left( {\alpha},\Gamma\right)$ is deduced from the service rate of the queue Q$_s$when $\tilde{Q}_{BS}$ is empty ($\tilde{\mu}_{s}^{0}\left(\Gamma\right)$) and when $\tilde{Q}_{BS}$ is not empty ($\tilde{\mu}_{s}^{1} \left( \alpha,\Gamma\right)$) as it follows:

\[
\tilde{\mu}_{s}\left( \tilde{\alpha},\Gamma\right)=\tilde{\varPi}_0\left(  {\alpha},\Gamma\right)\tilde{\mu}_{s}^{0}\left(\Gamma\right)+\left(1-\tilde{\varPi}_0\left( \alpha,\Gamma\right)\right)\tilde{\mu}_{s}^{1} \left( \alpha,\Gamma\right)
\]
\begin{equation}
\label{eq.mu1_comp_gen}
 =\frac{\left(b_{1}+\left(k-1\right)b_{11}-a_{1}-\left(k-1\right)a_{11}\right)\tilde{\mu}_{s}^{0}\left(\Gamma\right)+\left(a_{0}+\left(k-1\right)a_{01}\right)\tilde{\mu}_{s}^{1}\left( \tilde{\alpha},\Gamma\right)}{b_{1}+\left(k-1\right)b_{11}-a_{1}-\left(k-1\right)a_{11}+a_{0}+\left(k-1\right)a_{01}}
\end{equation}

Similarly, ${{\mu}}_{s}\left( {\alpha},\Gamma\right)$  can be written as it follows:
\[
\mu_{s}\left( \alpha,\Gamma\right)=\varPi_0\left(\alpha,\Gamma\right)\mu_{s}^{0}\left(\Gamma\right)+\left(1-\varPi_0\left(\alpha,\Gamma\right)\right)\mu_{s}^{1} \left(\alpha,\Gamma\right)
\]
Given that $\varPi_0$ is related to $\tilde{\varPi}_0$ by the equation (\ref{eq.Pi_0_comparison_gen}) then:

\[
\mu_{s}\left( \alpha,\Gamma\right)=\left( \tilde{\varPi}_0\left(\alpha,\Gamma\right)+\delta_{\varPi}\left(\alpha,\Gamma\right)\right)\mu_{s}^{0}\left(\Gamma\right)+\left(1-\tilde{\varPi}_0\left(\alpha,\Gamma\right)-\delta_{\varPi}\left(\alpha,\Gamma\right) \right)\mu_{s}^{1}\left(\alpha,\Gamma\right)
\]

\[
=\tilde{\varPi}_0\left(\alpha,\Gamma\right)\mu_{s}^{0}\left(\Gamma\right)+\left(1-\tilde{\varPi}_0\left(\alpha,\Gamma\right)\right)\mu_{s}^{1}\left(\alpha,\Gamma\right)+\delta_{\varPi}\left(\alpha,\Gamma\right)\left(\mu_{s}^{0}\left(\Gamma\right) -\mu_{s}^{1}\left(\alpha, \Gamma\right)\right)
\]

\[
=\tilde{\mu}_s\left( \alpha,\Gamma\right)-\frac{kb_{11}\sum\limits_{n=1}^{k-1}\varPi_{n}\left(\alpha,\Gamma\right)}{b_{12}+kb_{11}-a_{12}-ka_{11}+a_{02}+ka_{01}}\left(\mu_{s}^{0}\left(\Gamma\right) -\mu_{s}^{1}\left(\alpha,\Gamma\right)\right)
\]

\[
\Rightarrow \delta_\mu\left( {\alpha},\Gamma\right)=\frac{kb_{11}\sum\limits_{n=1}^{k-1}\varPi_{n}\left(\alpha,\Gamma\right)}{b_{12}+kb_{11}-a_{12}-ka_{11}+a_{02}+ka_{01}}\left(\mu_{s}^{0}\left(\Gamma\right) -\mu_{s}^{1}\left(\alpha,\Gamma\right)\right)\geq 0\,\,\,\, \forall \alpha, \forall\Gamma
\]
\[
\Rightarrow {\mu}_{s}\left( {\alpha},\Gamma\right) \leq \tilde{\mu}_{s}\left( {\alpha},\Gamma\right)  \,\,\,\, \forall \alpha \in \left[0,1\right], \forall \,\, \Gamma \in \Omega_\Gamma^{ss}
\]
The highest transmission rate of Q$_s$, under a given priority policy $\Gamma$, corresponds to the case where Q$_{BS}$, therefore we can write $\mu_s^1\left(\alpha,\Gamma\right) \leq \mu_s^0\left(\Gamma\right)\,\, \forall\Gamma, \forall \alpha$. Hence, we prove that  $\delta_\mu\left( {\alpha},\Gamma\right)\geq 0$ for all $\alpha \in \left[ 0,1\right]^4$ and for all $\Gamma \in \Omega_{\Gamma}$. Thus, for all the fraction vectors $\alpha$ and all the priority policies $\Gamma$, the approximated service rate is always higher than the exact service rate. We deduce that $\tilde{\mathcal{R}}_c^{ss}$ is an upper bound of ${\mathcal{R}}_c^{ss}$.\\

\subsubsection*{\textbf{\textit{Step 3}}} \textit{Relative distance between ${\mathcal{R}}_c^{ss}$ and $\tilde{\mathcal{R}}_c^{ss}$}\\

We define the relative distance between $\tilde{\mathcal{R}}_c^{ss}$ and ${\mathcal{R}}_c^{ss}$ as the maximum relative error between the service rate $\mu_s$ and $\tilde{\mu}_s$ over all the priority policies $\Gamma \in \Omega_\Gamma^{ss}$ (with $|\Gamma \in \Omega_\Gamma^{ss}|=6$).
We denote by $\epsilon_{\mu}\left(\alpha,\Gamma\right)$ the following expression:\\
\[
\epsilon_{\mu}\left(\alpha,\Gamma\right):=\frac{\tilde{\mu_s}\left(\alpha, \Gamma\right)-\mu_s\left(\alpha,\Gamma\right)}{\tilde{\mu_s}\left(\alpha,\Gamma\right)}=\frac{\frac{kb_{11}\sum\limits_{n=1}^{k-1}\varPi_{n}\left(\alpha,\Gamma\right)}{b_{12}+kb_{11}-a_{12}-ka_{11}+a_{02}+ka_{01}}\left(\mu_{s}^{0}\left(\Gamma\right)-\mu_{s}^{1}\left(\alpha,\Gamma\right)\right)\varPi_1\left(\alpha,\Gamma\right)}{\frac{\left(b_{12}+kb_{11}-a_{12}-ka_{11}\right)\mu_{s}^{0}\left(\Gamma\right)+\left(a_{02}+ka_{01}\right)\mu_{s}^{1} \left(\alpha,\Gamma\right)}{b_{12}+kb_{11}-a_{12}-ka_{11}+a_{02}+ka_{01}}}
\]
Using the notation introduced here for simplifying equation, we can write:
\[
\mu_s^0=r_1a_{01}+r_2a_{02}=r_2\left(ka_{01}+a_{02}\right) \text{ and }
\mu_s^1=r_1a_{11}+r_2a_{12}=r_2\left(ka_{11}+a_{12}\right)
\]
Thus $\epsilon_{\mu}\left(\alpha,\Gamma\right)$ can be written as it follows:
{\small{
\[
\epsilon_{\mu}\left(\alpha,\Gamma\right)=\frac{kb_{11}\sum \limits_{n=1}^{k-1}\varPi_{n}\left(\alpha,\Gamma\right)\left(a_{02}+ka_{01}-a_{12}-ka_{11} \right)}{\left(b_{12}+kb_{11}-a_{12}-ka_{11}\right)+\left(a_{02}+ka_{01}\right)+\left(a_{02}+ka_{01}\right)\left(b_{12}+kb_{11}-a_{12}-ka_{11}+a_{02}+ka_{01}\right)}
\]}}
\begin{equation}
\label{eq.delta_rel_alpha_gen}
\Rightarrow \epsilon_{\mu}\left(\alpha,\Gamma\right)=\frac{kb_{11}\sum\limits_{n=1}^{k-1}\varPi_{n}\left(\alpha\right)\left(a_{02}+ka_{01}-a_{12}-ka_{11} \right)}{\left(b_{12}+kb_{11}\right)\left(a_{02}+ka_{01}\right)}
\end{equation}

We recall that all the parameters in the equations above depend on the fraction vector $\alpha$ and on the priority policy $\Gamma$.
\\

The transmission rate of Q$_s$, under a given priority policy $\Gamma$, achieved its highest value when Q$_{BS}$, therefore we can write $\mu_s^1\left(\alpha,\Gamma\right) \leq \mu_s^0\left(\Gamma\right)\,\, \forall\Gamma, \forall \alpha$. Hence, $r_2\left(ka_{11}\left(\alpha\right)+a_{12}\left(\alpha\right)\right) \leq r_2\left(ka_{01}+a_{02}\right)$ thus $\left(ka_{11}\left(\alpha\right)+a_{12}\left(\alpha\right)\right) \leq \left(ka_{01}+a_{02}\right)$. Hence, we prove that $\epsilon_{\mu}\left(\alpha,\Gamma\right)\geq 0$ for all $\alpha \in \left[ 0,1\right]^4$ and for all $\Gamma \in \Omega_{\Gamma}^{ss}$ which validate the fact that the approximated stability region is an {upper bound} of the exact stability region of the cellular scenario. \\

Consider that a corner point of the stability region is achieved at $\alpha^*$ and that the limit of the approximated stability region is achieved at $\tilde{\alpha}^*$ then:
\[
\mu_s\left(\alpha,\Gamma\right)\leq \mu_s\left(\alpha^*,\Gamma\right) \text{ and } \tilde{\mu_s}\left(\alpha,\Gamma\right)\leq \tilde{\mu}_s\left(\tilde{\alpha}^*,\Gamma\right)\text{   } \forall \alpha \in \left[0,1\right]^4\]
We deduce that the service rate relative error $\epsilon_{\mu}\left(\alpha,\Gamma\right)$ for a given priority policy $\Gamma$ is maximized as it follows:
\[
\epsilon_{\mu}\left(\alpha,\Gamma\right)\leq \frac{|\mu_s\left(\alpha^*,\Gamma \right)-\tilde{\mu}_s\left(\tilde{\alpha}^*,\Gamma \right)|}{\tilde{\mu}_s\left(\tilde{\alpha}^*,\Gamma \right)}
= \frac{|\tilde{\mu}_s\left(\tilde{\alpha}^*,\Gamma \right)-\mu_s\left(\alpha^* ,\Gamma\right)|}{\tilde{\mu}_s\left(\tilde{\alpha}^* ,\Gamma\right)}
\]
\[=\left| \frac{kb_{11}\sum\limits_{n=1}^{k-1}\varPi_{n}\left(\Gamma\right)\left(a_{12}+ka_{11}-a_{02}-ka_{01} \right)}{\left(b_{12}+kb_{11}\right)\left(a_{02}+ka_{01}\right)}\right|_{\tilde{\alpha}^*}:=\epsilon^*\left(\Gamma\right)
\] 
\iffalse
Thus, 
\[
\delta_{\mu}\left(\Gamma\right)= \left| \frac{\mu_s\left(\alpha^* \right)-\tilde{\mu}_s\left(\tilde{\alpha}^* \right)}{{\mu}_s\left({\alpha}^* \right)} \right| \leq \left| \frac{\epsilon^*\left(\Gamma\right)}{1-\epsilon^*\left(\Gamma\right)}\right|
\]
\fi
We define the relative distance $\epsilon_R^*$ between ${\mathcal{R}}_c^{ss}$ and $\tilde{\mathcal{R}}_c^{ss}$ as the maximum of the relative service rate error $\epsilon^*\left(\Gamma\right)$ over all the priority policies $\Gamma \in \Omega_{\Gamma}^{ss}$. Considering all the corner points of the stability region which corresponds to all $\Gamma \in \Omega_{\Gamma}^{ss}$, we define the   $\epsilon_R^*$ as it follows :
\[
\epsilon_R^*:= \max \limits_{\Gamma \in  \Omega^{ss}}\epsilon^*\left(\Gamma\right)= \max \limits_{\Gamma \in  \Omega^{ss}}\left| \frac{kb_{11}\sum\limits_{n=1}^{k-1}\varPi_{n}\left(\Gamma\right)\left(a_{12}+ka_{11}-a_{02}-ka_{01} \right)}{\left(b_{12}+kb_{11}\right)\left(a_{02}+ka_{01}\right)}\right|_{\tilde{\alpha}^*}
\]
Thus, we deduce that the exact service rate of ${\mu_s}$ is limited by its approximation $\tilde{\mu_s}$ as it follows:
\[
\epsilon_{\mu}\left(\alpha,\Gamma\right)=\frac{\tilde{\mu_s}\left(\alpha, \Gamma\right)-\mu_s\left(\alpha,\Gamma\right)}{\tilde{\mu_s}\left(\alpha,\Gamma\right)}\leq \epsilon_R^* \,\,\,\,\,\,\,\, \forall \alpha \in \left[0,1\right] \text{ and } \forall \,\,\Gamma \in \Omega_\Gamma^{ss}
\]
\[
\Rightarrow {\mu_s}\left(\alpha,\Gamma\right) \geq  \left( 1- \epsilon_R^*\right) \tilde{\mu_s}\left(\alpha,\Gamma\right) \,\,\,\,\,\,\,\, \forall \alpha \in \left[0,1\right] \text{ and } \forall \,\,\Gamma \in \Omega_\Gamma^{ss}
\]
\iffalse
\[
\delta_R:= \max \limits_{\Gamma \in  \Omega^{ss}}\delta_{\mu}\left(\Gamma\right)
\]
We deduce the following:
\[
\delta_R \leq \max \limits_{\Gamma \in  \Omega^{ss}_{\Gamma}}\frac{\epsilon^*\left(\Gamma\right)}{1-\epsilon^*\left(\Gamma\right)}
\]
\fi
\subsubsection*{\textbf{\textit{Step 4}}} \textit{Upper and lower bound of ${\mathcal{R}}_c^{ss}$}\\

We proved in step 2 that the approximated service rate is an upper bound of the exact service rate hence we can write: ${\mu_s}\left(\alpha,\Gamma\right) \leq \tilde{\mu_s}\left(\alpha,\Gamma\right); \,\,\,\,\,\,\,\, \forall \alpha \in \left[0,1\right] \text{ and } \forall \,\,\Gamma \in \Omega_\Gamma^{ss}$. We proved in step 3 that the relative error between the exact and approximated service rate of Q$_s$ has a maximum value equal to $\epsilon_R^*$ hence we can write: ${\mu_s}\left(\alpha,\Gamma\right) \geq  \left( 1- \epsilon_R^*\right) \tilde{\mu_s}\left(\alpha,\Gamma\right); \,\,\,\,\,\,\,\, \forall \alpha \in \left[0,1\right] \text{ and } \forall \,\,\Gamma \in \Omega_\Gamma^{ss}$.\\

Let us verify that the exact stability region is bounded as it follows:
\[
 \left(1-\epsilon_R^*\right)\tilde{\mathcal{R}}_c^{ss} \subseteq  {\mathcal{R}}_c^{ss} \subseteq \tilde{\mathcal{R}}_c^{ss}
\]
\iffalse
Similar analysis to the step 5 of appendix \ref{stab_real_twoRate} is applied in order to verify that having $\left(\lambda_s,\lambda_u\right) \in \left(\lambda_s,\lambda_u\right) \in \tilde{\mathcal{R}}_c^{ss}$  is equivalent to the stability of the system of queues in the  approximated model.
following statements:\
\fi
\newpage
These bounds are verified by proving the following two statements:
\begin{itemize}
\item Having a set of mean arrival rate $\left(\lambda_s,\lambda_u\right) \in {\mathcal{R}}_c^{ss}$ (which means the stability of the exact model of the queues) gives that $\left(\lambda_s,\lambda_u\right) \in {\tilde{\mathcal{R}}}_c^{ss}$ (which means the stability of the approximated model of the queues). \\
\item Having a set of mean arrival rate $\left(\lambda_s,\lambda_u\right) \in \left(1-\epsilon_R^* \right)\tilde{\mathcal{R}}_c^{ss}$ gives that $\left(\lambda_s,\lambda_u\right) \in {{\mathcal{R}}}_c^{ss}$. \\
\end{itemize}

\iffalse
\begin{itemize}
\item Having a set of mean arrival rate $\left(\lambda_s,\lambda_u\right) \in \left(1-\epsilon_R^*\right){\tilde{\mathcal{R}}}_c^{ss}$ (which means $\in  {\mathcal{R}}_c^{ss}$ and $\in  \tilde{\mathcal{R}}_c^{ss}$) is identical to the stability validation of the system of queues in both models (exact and approximated). \\
\item Having a set of mean arrival rate $\left(\lambda_s,\lambda_u\right) \in {\mathcal{R}}_c^{ss}$ (which means $\in  \tilde{\mathcal{R}}_c^{ss}$) is identical to the stability validation of the system of queues in both models (exact and approximated). \\
\end{itemize}

 For brevity, we avoid repetition and remove this proof.

We prove that if the mean arrival rates $\lambda_s$ and $\lambda_u$ are in the region $\mathcal{R}_c^{ss}$ with $ \left(1-\epsilon_R^*\right)\tilde{\mathcal{R}}_c^{ss} \subseteq  {\mathcal{R}}_c^{ss} \subseteq \tilde{\mathcal{R}}_c^{ss}$ is equivalent to the stability of the system of queues. To do so we prove that being having $\bm{\lambda}\in\mathcal{R}_c^{ss}$ with $\left(1-\epsilon_R^*\right)\tilde{\mathcal{R}}_c^{ss} \subseteq  {\mathcal{R}}_c^{ss} \subseteq \tilde{\mathcal{R}}_c^{ss}$ gives the stability of the queues and vice versa. \fi
To do so, we use the same notation as in step 5 of appendix \ref{stab_real_twoRate} and we recall that $\mu_i\left(\Gamma\right)=\mu_{i}^0\left(\Gamma\right)\varPi_0+\mu_{i}^1\left(\Gamma\right)\left(1-\varPi_0\right)$. \\

\subsubsection*{\textbf{\textit{Step 4.a}}} \textit{$\bm{\lambda} \in  \mathcal{R}_c^{ss} \Rightarrow \bm{\lambda} \in  \tilde{\mathcal{R}}_c^{ss}$}\\

$\bm{\lambda} \in  \mathcal{R}_c^{ss} \Rightarrow$ for each channel $i=s,d,u$ it exists a $\mu_i^*=\sum\limits_{\Gamma} \varPi_{\Gamma} \mu_i\left(\Gamma\right)$ as combination of $\mu_i\left(\Gamma\right)$ for different scheduling policies $\Gamma$ such that $\lambda_i \leq \mu_i^*$. Given that in \textbf{\textit{step 2}} we verified that $\mu \left( \alpha,\Gamma \right) \leq \tilde{\mu} \left( \alpha,\Gamma \right)$ for all $\alpha \in \left[0,1\right]$ and all $\Gamma \in \Omega_{\Gamma}^{ss}$, then it exists $\tilde{\mu}_i^* \in \tilde{\mathcal{R}}_c^{ss}$ such that ${\mu}_i^* \leq \tilde{\mu}_i^*$. Thus, we deduce that $\lambda_i \leq \mu_i^* \leq \tilde{\mu}_i^*$ which means that $\bm{\lambda} \in \tilde{ \mathcal{R}}_c^{ss}$ and that the queues are stable in the approximated model.\\

\iffalse
We have by the law of large numbers that:
\[
\lim_{t\rightarrow \infty}\frac{1}{t}\sum_{\tau=0}^{t-1}b_i\left(\tau\right)
=\lim_{t\rightarrow \infty}\frac{1}{t}\sum_{\Gamma^* \in \Omega_{\Gamma}}\sum_{\tau=0}^{t-1}b_i\left(\tau|\Gamma=\Gamma^*\right)\mathbbm{1}_{\Gamma=\Gamma^*}=\lim_{t\rightarrow \infty}\frac{1}{t}\sum_{\Gamma^* \in \Omega_{\Gamma}}\varPi_{\Gamma^*}\sum_{\tau=0}^{t-1}b_i\left(\tau|\Gamma=\Gamma^*\right)
\]
\[
=\sum_{\Gamma^* \in \Omega_{\Gamma}}\varPi_{\Gamma^*}\mu_i\left(\Gamma^*\right)
=\mu_i^*
\]
Equation (\ref{eq.mu_muapp}) gives that $\left(1-\epsilon_R^*\right)\tilde{\mu}_i^* \leq \lim\limits_{t\rightarrow \infty}\frac{1}{t}\sum\limits_{\tau=0}^{t-1}b_i\left(\tau\right) \leq \tilde{\mu}_i^* $

Hence, $\lambda_i \leq \mu_i^* =\lim\limits_{t\rightarrow \infty}\frac{1}{t}\sum\limits_{\tau=0}^{t-1}b_i\left(\tau\right)$ which gives that queue $i$ is stable for $i=s,d,u$.
We deduce that if $\lambda \in  \mathcal{R}_c^{ss}$ with $\left(1-\epsilon_R^*\right)\tilde{\mathcal{R}}_c^{ss} \subseteq  {\mathcal{R}}_c^{ss} \subseteq \tilde{\mathcal{R}}_c^{ss}$ then the system of queues is stable.\\

We use the Lyapunov drift to prove the stability.  The Lyapunov function is given by:
\[
L\left(Q\left(t\right)\right)=\frac{1}{2}\left[Q_1\left(t\right)+Q_2\left(t\right)+Q_3\left(t\right)\right]
\]
\fi

\subsubsection*{\textbf{\textit{Step 4.b}}} \textit{$\bm{\lambda} \in  \left(1-\epsilon_R^*\right) \tilde{\mathcal{R}}_c^{ss} \Rightarrow \bm{\lambda} \in  {\mathcal{R}}_c^{ss}$}\\
\iffalse
If the queues are stable then each queue $i \in \left\lbrace s,d,u\right\rbrace$ has a stationary distribution $\varPi_i$. The mean service rate is given by:
\[
\mu_i=\mathbb{E}\left[ \hat{b}_i \left(S\left(t\right),Q\left(t\right),\Gamma\left(t\right) \right)\right]
\]
\[
=\sum_{\Gamma^*\in\Omega_{\Gamma}}\varPi_{\Gamma^*}\left[\mathbb{E}\left[ \hat{b}_i \left(S\left(t\right)|Q_{BS}=0,\Gamma=\Gamma^* \right)\varPi_{BS}\left(0\right)\right]+\mathbb{E}\left[\hat{b}_i \left(S\left(t\right)|Q_{BS}>0,\Gamma=\Gamma^* \right)\bar{\varPi}_{BS}\left(0\right)\right]\right]
\]
\[
=\sum_{\Gamma^* \in \Omega_{\Gamma}}\varPi_{\Gamma^*}\mu_i\left(\Gamma^*\right)
\]
\fi

$\bm{\lambda} \in \left(1-\epsilon_R^* \right) \tilde{\mathcal{R}}_c^{ss}$ then for each channel $i=s,d,u$ it exists a $\tilde{\mu}_i^*=\sum\limits_{\Gamma} \varPi_{\Gamma} \tilde{\mu}_i\left(\Gamma\right)$ as combination of $\tilde{\mu}_i\left(\Gamma\right)$ for different scheduling policies $\Gamma$ such that $\lambda_i \leq \left(1-\epsilon_R^* \right) \mu_i^*$. Given that in \textbf{\textit{step 3}} we verified that $\mu \left( \alpha,\Gamma \right) \geq \left(1-\epsilon_R^* \right) \tilde{\mu} \left( \alpha,\Gamma \right)$ for all $\alpha \in \left[0,1\right]$ and all $\Gamma \in \Omega_{\Gamma}^{ss}$, then it exists ${\mu}_i^* \in {\mathcal{R}}_c^{ss}$ such that ${\mu}_i^* \geq \left(1-\epsilon_R^* \right) \tilde{\mu}_i^*$. Thus, we deduce that $\lambda_i \leq \left(1-\epsilon_R^* \right) \tilde{\mu}_i^* \leq {\mu}_i^*$ which means that the queues are stable in the exact model and that $\bm{\lambda} \in { \mathcal{R}}_c^{ss}$.\\
\iffalse
We proved respectively in \textit{\textbf{step 2}} and \textbf{\textit{step 3}} that $\mu_i\left(\alpha, \Gamma \right) \leq \tilde{\mu}_i\left(\alpha,\Gamma \right)$ and $\frac{ \tilde{\mu}_i\left(\alpha,\Gamma \right)-\mu_i\left(\alpha, \Gamma \right)}{ \tilde{\mu}_i\left(\alpha,\Gamma \right)}\leq \epsilon_R^*$ for all $\alpha \in \left[0,1 \right]^4$ and for all $\Gamma \in \Omega_\Gamma$. Hence, it exists $\tilde{\mu}_i$ such that $\left(1-\epsilon_R^*\right)\tilde{\mu}_i\leq \mu_i=\varPi_{\Gamma^*}\mu_i\left(\Gamma^*\right)\leq \tilde{\mu}_i$. \\

Knowing that each queue $i\in \left\lbrace s,d,u \right\rbrace$ is stable then $\lambda_i \leq \mu_i= \sum \limits_{\Gamma^* \in \Omega_{\Gamma}}\varPi_{\Gamma^*}\mu_i\left(\Gamma^*\right)$ with $\mu_i$ given as a combination of the limit of the stability region $\mathcal{R}_c^{ss}$ and with $\left(1-\epsilon_R^*\right)\tilde{\mu}_i\leq \mu_i\leq \tilde{\mu}_i$ hence $\lambda \in  \mathcal{R}_c^{ss}$ with $\left(1-\epsilon_R^*\right)\tilde{\mathcal{R}}_c^{ss} \subseteq  {\mathcal{R}}_c^{ss} \subseteq \tilde{\mathcal{R}}_c^{ss}$.\\
\fi

 From \textbf{\textit{Step 4.a}} and \textbf{\textit{Step 4.b}}, we deduce the following bounds of the exact stability region:
\[
\left(1-\epsilon_R^*\right)\tilde{\mathcal{R}}_c^{ss} \subseteq  {\mathcal{R}}_c^{ss} \subseteq \tilde{\mathcal{R}}_c^{ss}
\]

%\newpage
\subsection{Proof of theorem \ref{th1_oneRate_MU} \label{proof_oneRate_MU}}

Here, we find the exact stability region for the multi-UE scenario. For this aim, we pursue the same procedure as the three-UEs scenario in appendix \ref{stab_real_twoRate}. For a given $i^{th}$ \textit{UE2UE communication} and $j^{th}$ \textit{UE2BS communication} we proceed as it follows: \textit{\textbf{step 1}} models the queue Q$_{i,BS}$ as a Markov chain and computes the stability condition of the queue Q$_{i,BS}$ as well as the probability that this queue is empty,  \textit{\textbf{step 2}} obtains the service rate of both queues Q$_{i,s}$ and  Q$_{j,u}$ (these results are applied for all $1\leq i\leq K, 1\leq j\leq U$), \textit{\textbf{step 3}} combines the results of the previous steps to provide the exact stability region of the multi-UE scenario, \textit{\textbf{step 4}} verifies that having a set of arrival rates within this stability region is equivalent to the stability of the system of queues. The main challenge is to find the set of mean arrival rates $\bm{\lambda}$ that verify the stability of the users queues while guaranteeing the stability conditions of the queues $Q_{i,BS}$ at the BS level (for $1\leq i\leq K$). The computation of the exact stability region in this section proves the complexity of such procedure and motivates us to search for less complex solutions.

We recall that we have two types of communications: \textit{UE2UE communication} and \textit{UE2BS communications}. The \textit{UE2UE communications} pass thorough the BS such that each \textit{UE2UE communication} $i$ is modeled as the cascade of the queues Q$_{i,s}$ at the UE$_{i,s}$ level and Q$_{i,BS}$ at the BS level (see Fig.\ref{fig.Queues_MU}). We consider two rates $\left\lbrace r_1,r_2 \right\rbrace$ with $r_2=0$. Assuming that the sources queues are saturated, we want to characterize the stability region of the multi-UE scenario. $\Gamma$ is the priority policy according to which the users are prioritized for transmission with $\Omega_{\Gamma}$ the set of all the possible policies (it means $\Gamma \in \Omega_{\Gamma}$). \\

\subsubsection*{\textbf{\textit{Step 1}}} \textit{Markov chain model of Q$_{i,BS}$}\\

Cellular communications are modeled as coupled processor sharing queues where the service rates of Q$_{i,s}$ (equivalent to the arrival rate of Q$_{i,BS}$) as well as the service rate of Q$_{j,u}$ depend on the state (empty or not) of all the queues Q$_{k,BS}$ for $1\leq k \leq K$. Each queue Q$_{i,BS}$ can be modeled by a simple birth and death Markov chain (see Fig.\ref{fig.QBS_MU}) which describes the evolution of the queue Q$_{i,BS}$. The transition probabilities of the Markov chain corresponding to Q$_{i,BS}$ depend on: (i) the fraction vector $\alpha=\left(\alpha_1,...,\alpha_K\right)$ of all the queues Q$_{i,BS}$, (ii) the priority policy $\Gamma$ under which the queues are prioritized, and (iii) the state empty of not of the queues Q$_{i,BS}$ (for $1 \leq i\leq K$). These probabilities are denoted as it follows:
\begin{itemize}
\item Service probability: $b_{i,BS}\left(\alpha,\Gamma\right)$.\\
\item Arrival probability when Q$_{i,BS}=0$: $a_{i,BS}^0\left(\alpha,\Gamma\right)$.\\
\item Arrival probability when Q$_{i,BS}>0$: $a_{i,BS}^1\left(\alpha,\Gamma\right)$.\\
\end{itemize}

\begin{figure}[H]
\vspace{-20pt}
\begin{centering}
\includegraphics[width=0.7\textwidth]{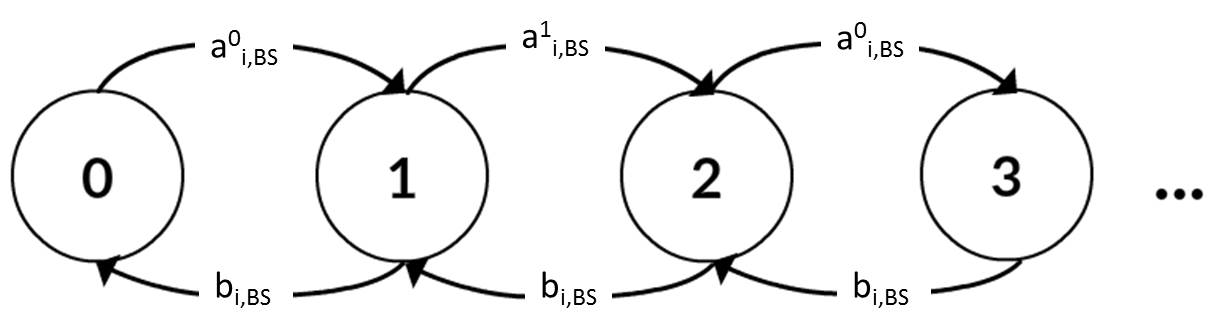}
\captionsetup{justification=centering}
\caption{Markov chain model of  each Q$_{i,BS}$ in the multi-UE scenario}
\label{fig.QBS_MU} 
\end{centering}
\end{figure}

\subsubsection*{\textbf{\textit{Step 1.a}}} \textit{Arrival and service probabilities of Q$_{i,BS}$}\\

We should differ between $a_{i,BS}^0$: arrival probability when Q$_{i,BS}$ is empty and $a_{i,BS}^1$: arrival probability when Q$_{i,BS}$ is not empty and these probabilities of arrival are respectively given by: (for all $1\leq i \leq K$)

\begin{equation}
\label{Pmu10_oneRate_MU}
a_{i,BS}^0\left(\alpha,\Gamma\right)=p_{i,1}X_i\left(\alpha,\Gamma\right)
\end{equation}

\begin{equation}
\label{Pmu11_oneRate_MU}
a_{i,BS}^1\left(\alpha,\Gamma\right)=p_s^{i}\left(\alpha_ip_d^{3}+\bar{p}_d^{i}\right)X_i\left(\alpha,\Gamma\right)
\end{equation}

with $X_i\left(\alpha,\Gamma\right)$ the probability that the resources are allocated to the $i^{th}$ \textit{UE2UE communication} for a given priority policy $\Gamma$ and a given fraction vector $\alpha$:

\[
X_i\left(\alpha,\Gamma\right)=\prod_{m \in \mathbf{U}_{\Gamma}\lbrace i \rbrace}\bar{p}_u^{m}\prod_{n \in \mathbf{K}_{\Gamma}\lbrace i \rbrace}\bar{p}_s^{n}\left[ \varPi_0^n+\bar{p}_d^{n}\bar{\varPi}^n_0 \right]
\]
with ${\varPi}^n_0$ the probability that the queue Q$_{n,BS}$ is empty and $\bar{\varPi}^n_0$ the probability that this queue is not empty. It will be computed \textbf{Step 1.c}.\\

The service probability $b_{i,BS}$ of $Q_{i,BS}$ can be written as it follows: (for $1\leq i \leq K$)
\begin{equation}
b_{i,BS}=p_d^{i}\left(1-\alpha_i p_s^{i}\right)X_i\left(\alpha,\Gamma\right)
\label{Pmu_BS_OneRate_MU}
\end{equation}

\subsubsection*{\textbf{\textit{Step 1.b}}} \textit{Stability condition of $Q_{i,BS}$}\\

Supposing that the arrival and
service processes of Q$_{i,BS}$ are strictly stationary and ergodic
then the stability of Q$_{i,BS}$ can be determined using \textbf{Loyne's} criterion
which states that the queue is stable if and only if the average arrival
rate is strictly less than the average service rate. Then the stability
condition that the fraction parameter $\alpha_i$
should verify in order to satisfy the stability of the queue $Q_{i,BS}$
is given by:
\[
a_{i,BS}<b_{i,BS}
\]
Since the arrival probability $a_{i,BS}$ at the queue Q$_{i,BS}$ is given by (\ref{eq.a_bs_gen}) then this condition can be written: 
\[
{b_{i,BS}a_{i,BS}^0}{b_{i,BS}-a_{i,BS}^1+a_{i,BS}^0}<b_{i,BS} \Rightarrow a_{i,BS}^1<b_{i,BS}
\]
Since $a_{i,BS}^1$ and $b_{i,BS}$ are given by (\ref{Pmu11_oneRate_MU}) and (\ref{Pmu_BS_OneRate_MU}) then the queue $Q_{i,BS}$ is stable if the fraction parameter $\alpha_i$ verifies:
\[
p_s^{i}X_i\left(\alpha_i p_d^{i}+1-p_d^{i}\right) \leq p_d^{i}\left(1-\alpha_i p_s^{i}\right)X_i
\]
\[
 2\alpha_i p_s^{i}p_d^{i} \leq p_d^{i} - p_s^{i}+p_s^{i}p_d^{i}
\]
\begin{equation}
\label{eq.alpha_opt_MUE}
 \alpha_i \leq \frac{p_d^{i} - p_s^{i}+p_s^{i}p_d^{i}}{2p_s^{i}p_d^{i}}:=\alpha_i^{*}
\end{equation}

We define the maximum fraction vector $\alpha^*=\left(\alpha_1^*,\alpha_2^*, ..., \alpha_K^*\right)$ as the vector containing the maximum fraction parameter $\alpha_i^*$ of all the \textit{UE2UE communications} for which the system is stable. Thus, the stability region is defined by varying the fraction vector $\alpha$ within $\left[ 0,\alpha^*\right]$: $\alpha \prec \alpha^*$ ($\prec$ is a component wise inequality).\\

\subsubsection*{\textbf{\textit{Step 1.c}}} \textit{Probability that Q$_{i,BS}$ is empty}\\

The probability that $Q_{i,BS}$ is empty is deduced from the stationary distribution ${\varPi}^{i}$ of the Markov Chain of $Q_{i,BS}$ (Fig. \ref{fig.QBS_MU}) which is deduced from the balance equations. As demonstrated in (\ref{eq.Pi_0_gen}), iff $a_{i,BS}<b_{i,BS}$ then the probability that $Q_{i,BS}$ is empty  is given by: (for all $1\leq i \leq K$)
\begin{equation}
{\varPi}^i_{0}\left(\alpha,\Gamma\right)=\mathbf{\mathbb{P}}\left[Q_{i,BS}=0\right]
=\frac{b_{i,BS}-a_{i,BS}^1}{b_{i,BS}-a_{i,BS}^1+a_{i,BS}^0}
\label{eq.Pi_0_oneRate_MU}
\end{equation}
\[=\frac{b_{i,BS}\left(\alpha,\Gamma\right)-a_{i,BS}^1\left(\alpha,\Gamma\right)}{b_{i,BS}\left(\alpha,\Gamma\right)-a_{i,BS}^1\left(\alpha,\Gamma\right)+a_{i,BS}^0\left(\alpha,\Gamma\right)}
=\frac{p_d^{i}\left(1-\alpha_i p_s^{i}\right)-p_s^{i}\left(\alpha_ip_d^{i}+\bar{p}_d^{i}\right)}{p_d^{i}\left(1-\alpha_i p_s^{i}\right)-p_s^{i}\left(\alpha_ip_d^{i}+\bar{p}_d^{i}\right)+p_s^{i}}
\]

$\Rightarrow$
\begin{equation}
\label{pi_MU}
{\varPi}_0^i=\frac{-2\alpha_ip_s^{i}p_d^{i}+p_d^{i}-p_s^{i}\bar{p}_d^{i}}{-2\alpha_ip_s^{i}p_d^{i}+p_d^{i}+p_s^{i}p_d^{i}}
\end{equation}
 
Which gives
\begin{equation}
\label{XN_MU}
X_i\left(\alpha,\Gamma\right)=\prod_{m \in \mathbf{U}_{\Gamma}\lbrace i \rbrace}\bar{p}_u^{m}\prod_{n \in \mathbf{K}_{\Gamma}\lbrace i \rbrace}\bar{p}_s^{n}\left[1+ \frac{{p}_s^{n}p_d^{n}}{2\alpha^*_n p_s^{n}p_d^{n}-\left(1+p_s^{n}\right)p_d^{n}} \right]
\end{equation}

\subsubsection*{\textbf{\textit{Step 2}}} \textit{Service rates of Q$_{i,s}$ and Q$_{j,u}$}\\

We compute the exact stability region of the multi-UE cellular scenario. For that, we compute the service rate of each \textit{UE2UE} and \textit{UE2BS communications} for a given priority policy $\Gamma$. For brevity, we present the results for the $i^{th}$ \textit{UE2UE communication} and $j^{th}$ \textit{UE2BS communication}. These results are applied for all $1 \leq i \leq K; 1 \leq j \leq U$. We follow the procedure below, based on queuing theory analysis of the network capacity, in order to derive the performance of the system of queues.\\
\subsubsection*{\textbf{\textit{Step 2.a}}} \textit{Service rate of Q$_{i,s}$}\\

If Q$_{i,BS}$ is empty then the service rate of Q$_{i,s}$ is denoted by $\mu_{i,s}^0$ and by $\mu_{i,s}^1$ otherwise then the service rate of Q$_{s,1}$ is computed as it follows:
\[
\mu_{i,s}\left(\alpha,\Gamma\right)=\mathbf{\mathbb{P}}\left[Q_{i,BS}=0\right]\mu_{i,s}^{0}\left(\alpha,\Gamma\right)+\mathbf{\mathbb{P}}\left[Q_{i,BS}>0\right]\mu_{i,s}^{1}\left(\alpha,\Gamma\right)
\] 
\begin{equation}
=\varPi_0^i\mu_{i,s}^{0}\left(\alpha,\Gamma\right)+\left[1-\varPi_0^i\right]\mu_{i,s}^{1}\left(\alpha,\Gamma\right)
\label{eq:Gmu1_real_oneRate_MU}
\end{equation}
with $\mu_{i,s}^{0}$ and $\mu_{i,s}^{1}$ given by:
\[
\mu_{i,s}^{0}\left(\alpha, \Gamma \right)=\mathbb{E}\left[\text{\ensuremath{\mu}}_{i,s}\left(S\left( t \right),\varGamma\,,\,Q(t)|Q_{i,BS}(t)=0\right)\right]
\]

\[
\mu_{i,s}^{1}\left( \alpha,\Gamma \right)=\mathbb{E}\left[\text{\ensuremath{\mu}}_{i,s}\left(S\left( t \right),\varGamma\,,\,Q(t)|Q_{i,BS}(t)>0\right)\right]
\]
Based on (\ref{XN_MU}), $\mu_{i,s}^{0}\left(\alpha,\Gamma\right)$ and $\mu_{i,s}^{1}\left(\alpha,\Gamma\right)$ can be written as it follows:

\begin{equation}
\label{mu10_oneRate_MU}
\mu_{i,s}^{0}\left(\alpha,\Gamma\right)=r_{1}p_s^{i}\prod_{m \in \mathbf{U}_{\Gamma}\lbrace i \rbrace}\bar{p}_u^{m}\prod_{n \in \mathbf{K}_{\Gamma}\lbrace i \rbrace}\bar{p}_s^{n}\left[1+ \frac{{p}_s^{n}p_d^{n}}{2\alpha_n p_s^{n}p_d^{n}-\left(1+p_s^{n}\right)p_d^{n}} \right]
\end{equation}
\begin{equation}
\label{mu11_oneRate_MU}
\mu_{i,s}^{1}\left(\alpha,\Gamma\right)=r_{1}p_s^{i}\left(\alpha_i p_d^{1}+1-p_d^{1}\right)\prod_{n \in \mathbf{K}_{\Gamma}\lbrace i \rbrace}\bar{p}_s^{n}\left[1+ \frac{{p}_s^{n}p_d^{n}}{2\alpha_n p_s^{n}p_d^{n}-\left(1+p_s^{n}\right)p_d^{n}} \right]
\end{equation}

Substituting all the $\varPi^i_{0}$ by their values from equation (\ref{pi_MU}) and based on the above expressions of $\mu_{i,s}^{0}$ and $\mu_{i,s}^{1}$, equation (\ref{eq:Gmu1_real_oneRate_MU}) gives that the formula (\ref{eqn_mu1c_oneRate_MU}) that describes the service rate $\mu_{i,s}\left(\alpha,\Gamma\right)$ of the $i^{th}$ \textit{UE2UE communications} in the multi-UE cellular scenario for a given fraction vector $\alpha$ and a given priority policy $\Gamma$. Not that $1\leq i \leq K$.\\

\subsubsection*{\textbf{\textit{Step 2.b}}} \textit{Service rate of Q$_{j,u}$}\\ 

For a given policy $\Gamma$, the service rate of each Q$_{j,u}$ with $1\leq j\leq U$ depends also on the state (empty or not) of all the queues Q$_{m,BS}$ for $m \in \mathbf{K}_\Gamma\lbrace j+K \rbrace$ and is computed as it follows:
 
\[\mu_{j,u}=\mathbb{E}\left[\text{\ensuremath{\mu}}_{j,u}\left(S(t),\varGamma\,,\,Q(t)|Q_{BS}(t)\right)\right]
= r_{1}p_{j,u}X_{j+K}\left(\alpha,\Gamma\right) 
\]
\[= r_{1}p_{j,u}\prod_{m \in \mathbf{U}_{\Gamma}\lbrace j+K\rbrace}\bar{p}_u^{m}\prod_{n \in \mathbf{K}_{\Gamma}\lbrace j+K \rbrace}\bar{p}_s^{n}\left[ \varPi_0^n+\bar{p}_d^{n}\bar{\varPi}_0^n\right]\]
\iffalse
\[=\sum_{j=1}^{2}r_{i}\sum_{i,k}p_{i1}*p_{j2}*p_{k3}*q_{2}(S_{i},S_{j},S_{k},\varGamma|Q_{i,BS}(t)=0)
\]

\[\mu_{i,2}^{1}=\mathbb{E}\left[\text{\ensuremath{\mu}}_{i,2}\left(S(t),\varGamma\,,\,Q(t)|Q_{i,BS}(t)>0\right)\right]
\]

\[\sum_{j=1}^{2}r_{j}\sum_{i,k}p_{i1}*p_{j2}*p_{k3}*q_{2}(S_{i},S_{j},S_{k},\varGamma|Q_{BS}(t)>0)\]

\[=r_{1}p_{12}\sum_{i=1}^{3}\sum_{j=1}^{2}p_{i1}p_{j3}*q_{2}(S_{i},S_{1},S_{j},\varGamma|Q_{BS}(t)>0)\]
\[+r_{2}p_{22}\sum_{i=1}^{3}\sum_{j=1}^{3}p_{i1}p_{j3}*q_{2}(S_{i},S_{2},S_{j}|Q_{BS}(t)>0)
\]

Based on (\ref{XN_MU}), $\mu_{i,2}^{0}$ and $\mu_{i,2}^{1}$ of the queue Q$_{i,2}$ for $1 \leq i \leq U$ can be written as it follows: 
\fi

Substituting all $\varPi_0^i$ by their values from equation (\ref{pi_MU}) we can verify that the service rate $\mu_{i,s}\left(\alpha,\Gamma\right)$ of the $j^{th}$ \textit{UE2BS communications} in the Multi-UE cellular scenario is given by (\ref{eqn_mu1c_oneRate_MU}) for a given fraction vector $\alpha$ and a given priority policy $\Gamma$. Note that $1 \leq j \leq U$.

\vspace{5pt}
 
 \subsubsection*{\textbf{\textit{Step 3}}} \textit{Characterization of the multi-UE stability region}

\vspace{5pt}
 
Combining the results of the previous steps, we deduce the exact stability region of the multi-UE scenario.Supposing that the arrival and service processes of the queues are strictly stationary and ergodic then their stability which is determined using \textbf{Loyne\textquoteright s} criterion is given by the condition that the average arrival rate is smaller than the average service rate.

The following procedure is pursue for capturing the stability region of the scenario. We consider the set of all the possible priority policies $ \Omega_{\Gamma}$ with $|\Omega_{\Gamma}|=\left(K+U\right)!$ ). Then, for each priority policy, we vary $\bm{\alpha} \in \left[0,\bm{\alpha}^*\right]$ in order to guarantee the stability of the Q$_{i,BS}$ queues (with $1\leq i\leq K$). For each value of $\alpha$, we find the probabilities that Q$_{i,BS}$ (with $1\leq i\leq K$) are empty in order to deduce the service rates of the queues in the system. This procedure is applied for all the priority policies $\Gamma \in \Omega_{\Gamma}$.

Therefore, the stability region for the multi-UE cellular scenario is characterized by the set of mean arrival rates $\bm{\lambda} \in \mathcal{R}_c$ such that:
\[
\mathcal{R}_{c}=co\left( \bigcup\limits_{\Gamma \in \Omega_{\Gamma}}\bigcup\limits_{\alpha \in \left[0,\alpha^* \right]}\lbrace \bm{\mu}\left(\bm{\alpha},\Gamma \right) \rbrace\right)
\]

where the elements of $\bm{\mu} \left( \Gamma \right)$ which are  $\mu_{i,s}\left( \Gamma\right)$ and $\mu_{j,u}\left( \Gamma\right)$ (for $1 \leq i\leq K$, $1 \leq j\leq U$) are respectively given by (\ref{eqn_mu1c_oneRate_MU}) and (\ref{eqn_mu2c_oneRate_MU}).

\subsubsection*{\textbf{\textit{Step 4}}} \textit{Arrival rates within the stability region is equivalent to the stability of the system of queues}

Proving that having a set of mean arrival rates $\bm{\lambda}$ in $\tilde{\mathcal{R}}_{c}$  is equivalent to the stability of the system queue is identical to the demonstration done in step 5 of the appendix \ref{stab_real_twoRate}. For brevity, we remove this proof to avoid repetition.

%\vspace{-15pt}
%\newpage
\subsection{Proof of theorem \ref{th2_oneRate_MU} \label{proof_oneRate_MU_th2}}
In order to reduce the computation complexity of the stability region for the cellular scenario of the multi-UE scenario, we start by reducing the set of $\alpha$ that is studied for characterizing this stability region.  We prove that considering only the border values of each fraction parameter $\alpha_i$ (which corresponds to $0$ and $\alpha_i^*$ given by equation (\ref{eq.alpha_opt_MUE})) is sufficient for characterizing the stability region $\mathcal{R}_c$. Hence, the procedure for computing the stability region  $\mathcal{R}_{c}$, as shown in appendix \ref{proof_oneRate_MU}, remains unaffected thus for brevity only the modifications will be presented in this proof.

Here, we prove that there is no need to vary each $\alpha_i \in \left[0,\alpha_i^* \right]$ for finding the stability region $\mathcal{R}_{c}$; but it is ample to consider only the border values of $\alpha_i \in \lbrace0,\alpha_i^*\rbrace$. For each $\alpha_i$ (for all $1\leq i \leq K$), only these two values $\lbrace 0, \alpha_i^*\rbrace$  have to be taken into account, thus the computation complexity of the stability region $\mathcal{R}_{c}$ decreases largely due to the fact that there is no need anymore to vary $\alpha_i$ within $\in \left[0,\alpha_i^*\right]$. Let us verify that for each fraction parameter $\alpha_i$ (that corresponds to the $i^{th}$ \textit{UE2UE communication}), it is sufficient to consider these two values of $\alpha_i$: $\lbrace 0,\alpha_i^*\rbrace$ for characterizing the exact stability region $\mathcal{R}_{c}$. To do so, we prove that, for all $1 \leq i\leq K$, the service rate vector $\bm{\mu}\left(\alpha_1,..,\alpha_i, ... , \alpha_K, \Gamma \right)$ can be written as a convex combination of $\bm{\mu}\left(\alpha_1,..,0, ... , \alpha_K, \Gamma \right)$  and $\bm{\mu}\left(\alpha_1,..,\alpha_i^*, ... , \alpha_K, \Gamma \right)$. \\

The following approach is applied separately over all the fraction parameters $\alpha_i$ (with $1 \leq i\leq K$). The mathematical form of the convex combination consists of verifying that for each $\alpha_i \in \left[0,\alpha_i^*\right]$ it exists a $\gamma_i \in \left[ 0 , 1\right]$ such that:
\begin{equation}
\label{eq.ConvComMu}
\bm{\mu}\left(\alpha_1,..,\alpha_i, ... , \alpha_K, \Gamma \right)=\gamma_i\bm{\mu}\left(\alpha_1,..,0, ... , \alpha_K, \Gamma \right)+(1-\gamma_i)\bm{\mu}\left(\alpha_1,..,\alpha_i^*, ... , \alpha_K, \Gamma \right)
\end{equation}

Expressing $\bm{\mu}\left(\alpha_1,..,\alpha_i, ... , \alpha_K, \Gamma \right)$ as a convex combination of $\bm{\mu}\left(\alpha_1,..,0, ... , \alpha_K, \Gamma \right)$ and \\ $\bm{\mu}\left(\alpha_1,..,\alpha_i^*, ... , \alpha_K, \Gamma \right)$ means that, for all $\alpha_i \in \left[0,\alpha_i^* \right]$, the service vector $\bm{\mu}\left(\alpha_1,..,\alpha_i, ... , \alpha_K, \Gamma \right)$ is an interior point of the convex hull $co \lbrace \bm{\mu}\left(\alpha_1,..,0, ... , \alpha_K, \Gamma \right),\bm{\mu}\left(\alpha_1,..,\alpha_i^*, ... , \alpha_K, \Gamma \right) \rbrace$. Hence, these two corner points $\bm{\mu}\left(\alpha_1,..,0, ... , \alpha_K, \Gamma \right)$ and $\bm{\mu}\left(\alpha_1,..,\alpha_i^*, ... , \alpha_K, \Gamma \right)$ are sufficient for describing the convex hull over all the points $\bm{\mu}\left(\alpha_1,..,\alpha_i, ... , \alpha_K, \Gamma \right)$ with $\alpha_i$ varying within $\left[0,\alpha_i^* \right]$.

In order to find the value of $\gamma_i$, we start by analyzing the expressions of $\bm{\mu}\left(\alpha_1,..,\alpha_i, ... , \alpha_K, \Gamma \right)$ and especially how the parameter $\alpha_i$ affects this formula. Indeed, the service rate of each communication in the system is affected by the parameter $\alpha_i$ in one of the two following forms: \iffalse appears in two forms inside $\bm{\mu} \left(\alpha,\Gamma\right)$:\fi 
\[
\text{\textbf{First form}: }  1+ \frac{\left(1-p_{i,s}\right)p_{i,d}}{-2\alpha_ip_{i,s}p_{i,d}+\left(1+p_{i,s}\right)p_{i,d}}:= F_1\left(\alpha_i\right)
\]
\[
\text{\textbf{Second form}:  } 1- \frac{p_{i,s}p_{i,d}}{-2\alpha_ip_{i,s}p_{i,d}+\left(1+p_{i,s}\right)p_{i,d}}:= F_2\left(\alpha_i\right)
\]
The first form $F_1\left(\alpha_i\right)$ appears only at the level of the service rate of the $i^{th}$ \textit{UE2UE communication} ($\mu_{1,i} \left( \bm{\alpha}, \Gamma \right)$) while the second form $F_2\left(\alpha_i\right)$ appears at the level of the service rate of all the other communications.  $\bm{\mu} \left(\bm{\alpha},\Gamma\right)$ can be written as in equation (\ref{eq.ConvComMu}) if and only if both forms ($F_1\left(\alpha_i\right)$ and $F_2\left(\alpha_i\right)$) respectively verified the equation (\ref{eq.ConvComMu}). In other terms, we need to verify that $F_1\left(\alpha_i\right)$ (respectively $F_2\left(\alpha_i\right)$) can be written as a convex combination of $F_1\left(0\right)$ and $F_1\left(\alpha^*_i\right)$ (respectively  $F_2\left(0\right)$ and $F_2\left(\alpha^*_i\right)$) with $\gamma_i$ as coefficient:
\[
F_1\left(\alpha_i\right)=\gamma_i F_1\left(0\right)+\left(1-\gamma_i\right)F_1\left(\alpha_i^*\right)
\]
\[
F_2\left(\alpha_i\right)=\gamma_i F_2\left(0\right)+\left(1-\gamma_i\right)F_2\left(\alpha_i^*\right)
\]
Knowing that $\alpha \in \left[ 0,\alpha_i^*\right]$, we can write $\alpha_i=\lambda\alpha_i^*$ with $\lambda \in \left[ 0,1\right]$. Let us consider the following notation:
\begin{equation}
\label{eq.gamma_i}
\gamma_i=\frac{\left(1+p_{i,s}\right)p_{i,d}\left(1-\lambda\right)}{\left(1+p_{i,s}\right)p_{i,d}-\lambda\left(p_{i,s}p_{i,d}+p_{i,d}-p_{i,s}\right)}
\end{equation}

\[
\Rightarrow 1-\gamma_i=\frac{p_{i,s}\left(1-\lambda\right)}{\left(1+p_{i,s}\right)p_{i,d}-\lambda\left(p_{i,s}p_{i,d}+p_{i,d}-p_{i,s}\right)}
\]
We show that $0 \leq \gamma_i \leq 1$. First, both numerator and denominator are positives then $\gamma_i \geq 0$. Second, 
\[
-\lambda\left(1+p_{i,s}\right)p_{i,d}\leq-\lambda\left(1+p_{i,s}\right)p_{i,d}+\lambda p_{i,s}
\]
\[
\Rightarrow\left(1+p_{i,s}\right)p_{i,d}\left(1-\lambda\right)\leq\left(1+p_{i,s}\right)p_{i,d}-\lambda\left(p_{i,d}+p_{i,s}p_{i,d}-p_{i,s}\right) \Rightarrow\gamma_i\leq1
\]
Knowing that $\gamma_i$ (given by (\ref{eq.gamma_i})) is within $\left[ 0,1\right]$, we can easily verify the following:
\[
F_1\left(\alpha_i\right)=\gamma_i F_1\left(0\right)+\left(1-\gamma_i\right) F_1\left(\alpha_i^*\right)
\]
\[
F_2\left(\alpha_i\right)=\gamma_i F_2\left(0\right)+\left(1-\gamma_i\right) F_2\left(\alpha_i^*\right)
\]
where $F_1\left(0\right)=1+\frac{\bar{p}_{i,s}}{1+p_{i,s}}$, $F_1\left(\alpha_i^*\right)=1+\frac{\bar{p}_{i,s}p_{i,d}}{p_{i,s}}$, $F_2\left(0\right)=1-\frac{p_{i,s}}{1+p_{i,s}}$ and $F_1\left(\alpha_i^*\right)=\bar{p}_{i,d}$.\\

Hence, for all $1 \leq i \leq K$ the service rate vector $\bm{\mu}$ can be written as in equation (\ref{eq.ConvComMu}) with $\gamma_i$ given by (\ref{eq.gamma_i}). Therefore, we show that taking the convex hull of the vector $\bm{\mu}$ over all $\alpha_i \in \left[0,\alpha_i^*\right]$ is equivalent to considering the convex hull over $\alpha_i \in \lbrace 0,\alpha_i^* \rbrace $. This remains true for all the $1 \leq i \leq K$. \\

The same procedure, as in appendix \ref{proof_oneRate_MU}, is pursued for capturing the less complex stability region of the scenario: (i) we consider the set of all the possible priority policies $ \Omega_{\Gamma}$ with $|\Omega_{\Gamma}|=\left(K+U\right)!$ ), (ii) for each priority policy, we vary $\alpha_i \in \lbrace 0, \alpha_i^*\rbrace$ (with $1\leq i\leq K$), (iii) for each value of $\alpha$, we find the probabilities that Q$_{i,BS}$ (with $1\leq i\leq K$) are empty in order to deduce the service rates of the queues in the system, (iv) this procedure is applied for all the priority policies $\Gamma \in \Omega_{\Gamma}$. The only difference with that in appendix \ref{proof_oneRate_MU} is at the level of the step (ii) where we consider a limited set of values of the fraction vector $\bm{\alpha}$ compared to that in appendix \ref{proof_oneRate_MU}. We see that the complexity in this case is reduced to: $2^K\left(K+U\right)!$.\\

Thus, for the multi-UE cellular scenario, the stability region $\mathcal{R}_c$ can be simplified as it follows:
\[
\mathcal{R}_{c}=co\left(\bigcup\limits_{\Gamma \in \Omega_{\Gamma}}\bigcup\limits_{\alpha \in \mathbb{S}_{\alpha}}\lbrace \bm{\mu}\left(\bm{\alpha},\Gamma \right) \rbrace\right)
\]
with 
\[
\mathbb{S_{\alpha}}=\lbrace \bm{\alpha} \, | \, \alpha_i \in \lbrace 0, \alpha_i^*=\frac{p_{i,d}-p_{i,s}+p_{i,s}p_{i,d}}{2p_{i,s}p_{i,d}}  \rbrace \,\,\,\forall\, 1 \leq i \leq K \rbrace
\]
where the elements of $\bm{\mu} \left( \Gamma \right)$ which are  $\mu_{i,s}\left( \Gamma\right)$ and $\mu_{j,u}\left( \Gamma\right)$ (for $1 \leq i\leq K$, $1 \leq j\leq U$) are respectively given by (\ref{eqn_mu1c_oneRate_MU}) and (\ref{eqn_mu2c_oneRate_MU}).\iffalse
{{\color{red} ADD that $\alpha^*$ is sufficient + figures for k=2 and k=3}}
\fi

%\newpage
\subsection{Proof of theorem \ref{th3_oneRate_MU} \label{proof_oneRate_MU_th3}}

Similarly to the proof in appendix \ref{proof_oneRate_MU_th2}, the demonstration here follows the same steps as in the appendix \ref{proof_oneRate_MU}. The only difference is at the level of the set of the considered priority policies that will be studied for characterizing the stability region. Actually, in appendix \ref{proof_oneRate_MU} we consider all the possible priority policies $\Gamma \in \Omega_{\Gamma}$ with $|\Omega_{\Gamma}|=\left(K+U\right)!$ while in this proof we reduce the subset of the considered priority policies to $\Omega_{\Gamma}^{K_0} \subset \Omega_{\Gamma}$ where $|\Omega_{\Gamma}^{K_0}|=\frac{\left(K+U\right)!}{\left(K+U-K_0\right)!}$ with $K_0 <<<< K+U$.We denote by $p_s=p_{i,s}=p_{j,u}$ and $p_d=p_{i,d}$ overall  $1 \leq i \leq K$ and $1 \leq j  \leq U$.\\

From appendix \ref{proof_oneRate_MU_th2} we can see that finding the stability region of the cellular symmetric case involves the consideration of $|\Omega_{\Gamma}|=\left(K+U\right)!$ policies and $2^{K}$ values of the fraction vector. \iffalse $\alpha$ ($\alpha_K^*$ is sufficient for the less prioritized user and no need to consider $\alpha_K=0$)\fi On the aim of reducing this  $2^{K}\left(K+U\right)!$ complexity, we remark that it exists a $K_0 \in \mathbb{N}$ such than when the user has a priority less than $K_0$ then the improvement in terms of performance (i.e. service rate of its queue) is small. Hence, we avoid the search among all the priority policies $\Omega_{\Gamma}$ and we limit the search to the set $\Omega_{\Gamma}^{K_0}$ which consists of the priority policies of all the subsets of $K_0$ elements within the $K+U$ communications. Thus, the computational complexity is reduced to $2^{K_0}\frac{\left(K+U\right)!}{\left(K+U-K_0\right)!}$.\\

Considering $\epsilon$ as the preciseness parameter, we pursue the following procedure for finding the $\epsilon-$ approximation of the stability region of the multi-UE symmetric case: \textbf{step 1} gives the best service rate that a user can have when it is at the $k^{th}$ priority level, \textbf{step 2} deduces the priority level $K_0\left(\epsilon\right)$ (as function of the percentage of preciseness $\epsilon$) such that beyond this priority level the error of the users' service rate is bounded by $\epsilon$; this enables us to deduce the set of the priority policies $\Omega_{\Gamma}^{K_0}$ that is sufficient for describing the $\epsilon$-approximation of the stability region, \textit{\textbf{step 3}} combines the results of the previous steps to provide the $\epsilon$ close upper bound of the stability region of the symmetric case of the multi-UE scenario. \\

\subsubsection*{\textbf{\textit{Step 1 }}} \textit{The best service rate of a user at a priority level $k$}\\

We consider the symmetric case of the multi-UE scenario with $K$ \textit{UE2UE communications} and  $U$ \textit{UE2BS communications} ($2K+U$ users). A user with a priority level $k$ has, at the best case, a service rate equals to:
\begin{equation}
\mu^k=r_1p_s\bar{p}_s^{k-1}
\label{eq.BestCasemu}
\end{equation}

This service rate corresponds to the most optimistic case where the $k^{th}$ user as well as all the more prioritized UEs ($1 \leq i \leq k$) have \textit{UE2BS communications}. Actually, considering that the user at priority level $k$ has $K_1$ more prioritized \textit{UE2UE communications} and $K_2$ more prioritized \textit{UE2BS communications} then from (\ref{eqn_mu1c_oneRate_MU}) and (\ref{eqn_mu2c_oneRate_MU}) we deduce the following best service rates:\\
\begin{itemize}
\item If the $k^{th}$ priority level corresponds to a \textit{UE2UE communication} then its most optimist service rate corresponds to $\alpha_k=\alpha_k^*$ and $\alpha_i=0$ for all $1 \leq i \leq K_1$ :
\[
\mu^k_{s}=\frac{1}{2}r_1{p}_s\left(1+\frac{\bar{p_s}p_d}{p_s}\right)\bar{p}_s^{K_1+K_2}\left(\frac{1}{1+p_s}\right)^{K_1}
\]

\item If the $k^{th}$ priority level corresponds to a \textit{UE2BS communication} then its most optimist service rate corresponds to $\alpha_i=0$ for all $1 \leq i \leq K_1$ :
\[
\mu^k_{u}=r_1{p_s}\bar{p}_s^{K_1+K_2}\left(\frac{1}{1+p_s}\right)^{K_1}
\]
\end{itemize}

\newpage
Furthermore, a more optimist case corresponds to the case where all the more prioritized users having \textit{UE2BS communications}, which means $K_1=k-1$ and $K_2=0$. Then:
\begin{itemize}
\item If the $k^{th}$ is participating to a \textit{UE2UE communication} then its most optimist service rate is:
\[
\mu^k_{s}=\frac{1}{2}r_1{p_s}\left(1+\frac{\bar{p_s}p_d}{p_s}\right)\bar{p}_s^{k-1}
\]

\item If the $k^{th}$ is participating to a \textit{UE2BS communication} then its most optimist service rate is:
\[
\mu^k_{u}=r_1{p_s}\bar{p}_s^{k-1}
\]
\end{itemize}
It is obvious that $\mu^k_{u}\geq \mu^k_{s}$ since $\frac{1}{2}\left(1+\frac{\bar{p_s}p_d}{p_s}\right) \leq 1$. Hence, the service rate of a user at a priority level $k$ (denoted by $\mu^k$) is maximized as it follows: 
\[
\mu^k\leq r_1{p_s}\bar{p}_s^{k-1}:=\mu^{k*}
\]
We deduce that the service rate of any user, at priority level $k$, is lower or equal to $\mu^{k*} :=r_1{p_s}\bar{p}_s^{k-1}$. Our approximation is based on finding the priority level $K_0$ such that the maximum service rate $\mu^{K_0*}= r_1{p_s}\bar{p}_s^{K_0-1}$ at this priority level is lower than $\epsilon$. Then, we assume that the service rates of the communications that have a higher priority level (greater than $K_0$) are equal to zero.

\subsubsection*{\textbf{\textit{Step 2 }}} \textit{The priority policies subset $\Omega_\Gamma^{K_0}$ as function of $\epsilon$}\\

The level $K_0$ corresponds to the priority policy level at which the service rate is negligible which means lower than $\epsilon$. Knowing that the service rate of a user at a priority level $k$ is bounded by equation (\ref{eq.BestCasemu})), hence for the priority level $K_0$ we can write:
\[
\mu^{K_0*}=r_1p_s\bar{p}_s^{K_0-1} \leq \epsilon \Rightarrow  K_0=\left\lceil 1+\frac{\log\left(\frac{\epsilon}{r_1p_s}\right)}{\log\left(\bar{p}_s\right)}\right\rceil
\]
We should be aware that $K_0$ is limited by the total number of communications: ($K+U$). Hence,
\begin{equation}
\label{eq.K0_2}
K_0=min\left\{K+U,\left\lceil 1+\frac{\log\left(\frac{\epsilon}{r_1p_s}\right)}{\log\left(\bar{p}_s\right)}\right\rceil\right\}
\end{equation}

Thus, we define  $\Omega_\Gamma^{K_0}$ as it follows:
\begin{itemize}
\item We consider all the possible sets of $K_0$ communications within the all existing $K+U$ communications. In total, there is $\frac{\left(K+U\right)!}{\left(K+U-K_0\right)!K_0!}$ subsets of $K_0$ elements. \\
\item Then for each of the above subsets, we find all the possible priority sorting of these $K_0$ communications. We know that $K_0!$ is the number of the possible permutation of the $K_0$ elements of each subset.\\
\end{itemize}
We deduce that $\Omega_\Gamma^{K_0}$ , the set of all the possible priority policies of all the subset of $K_0$ elements within the total of $K+U$ communications, consists of $\frac{\left(K+U\right)!}{\left(K+U-K_0\right)!}$ elements.\\
 
\subsubsection*{\textbf{\textit{Step 3 }}} \textit{Characterizing the approximated stability region}\\

The same procedure, as in appendix \ref{proof_oneRate_MU_th2}, is pursued for capturing the $\epsilon$-approximated stability region of multi-UE symmetric case. However, here the set of the studied priority policies is restricted in order to reduce the complexity and proposing an $\epsilon$ close upper bound of the exact stability region. We proceed as it follows : (i) for a given $\epsilon$ (for a given precision), we find $K_0$ from equation (\ref{eq.K0_2})  and we deduce the subset $ \Omega_{\Gamma}^{K_0}$ of the required priority policies for characterizing the $\epsilon$ upper bound of the exact stability region (see \textbf{\textit{step 2}} where $|\Omega_{\Gamma}^{K_0}|=\frac{\left(K+U\right)!}{\left(K+U-K_0\right)!}$ ), (ii) for each priority policy $\Gamma \in \Omega_{\Gamma}^{K_0}$, we vary $\alpha_i \in \lbrace 0, \alpha_i^*\rbrace$ (with $1\leq i\leq K$), (iii) for each value of $\alpha$, we find the probabilities that the queues Q$_{i,BS}$ (with $1\leq i\leq K$) are empty in order to deduce the service rates of the queues in the system. The main difference between this procedure and the one in appendix \ref{proof_oneRate_MU_th2} is at the level of the step (i) where we consider a limited set of priority policies $\Gamma \in \Omega_{\Gamma}^{K_0}$ and not all the possible priority policies $\Omega_{\Gamma}^{K_0}$. For these priority policies we consider that the service rates of the communications, that have a priority level higher than $K_0$, are equal to zero.\\

We define $\mu^k_{i}$ as the service rate of the $i^{th}$ communication (with $1 \leq i \leq K+U$) when it has a priority level $k$. Hence, the approximated model (with $\tilde{\bm{\mu}}$ denotes the service rate vector) consists of considering the following: \\
\begin{itemize}
\item $\tilde{\mu}_i^k=\mu_i^k$ for $1 \leq i \leq K+U$ and for all $k \leq K_0$.

\item $\tilde{\mu}_i^k=0$ for $1 \leq i \leq K+U$ and for all $k \geq K_0$.\\
\end{itemize}
Thus, we can deduce that the service rate of all the communications  ($1 \leq i \leq K+U$) for any priority policy $\Gamma \in \Omega_\Gamma ^{K_0}$ is bounded as it follows:
\begin{equation}
\label{eq.boundsmu}
\tilde{\mu_i} \leq {\mu_i}  \leq \tilde{\mu_i} + \epsilon
\end{equation}

The above expression can be proved by verifying that it remains true for both cases: (i) priority levels $k \leq K_0$ and priority levels $k \geq K_0$.\\

A similar analysis as the one done in \textbf{\textit{step 4}} of appendix  \ref{th_comp_gen}, let us deduce the bounds of the stability region based on the bounds of the service rates (equation (\ref{eq.boundsmu})). Thus, for the symmetric case of the multi-UE cellular scenario , the exact stability region can be bounded as it follows:
\[
\tilde{\mathcal{R}}_{c} \subseteq  {\mathcal{R}}_{c} \subseteq \tilde{\mathcal{R}}_{c}+\epsilon
\]
Where the $\epsilon$ close upper bound of the exact stability region $\tilde{\mathcal{R}}_{c}$ is given by the set of $ \bm{\lambda} \in\tilde{{\mathcal{R}}}_c$ such that:
\[
\tilde{\mathcal{R}}_{c}=co\left(\bigcup\limits_{\Gamma \in \Omega_{\Gamma}^{K_0}}\bigcup\limits_{\alpha \in \mathbb{S}_{\alpha}^{K_0}}\lbrace \bm{\mu}\left(\bm{\alpha},\Gamma \right) \rbrace
\right)
\]
with 
\[
K_0=min\left\{K+U,\left\lceil 1+\frac{\log\left(\frac{\epsilon}{r_1p_s}\right)}{\log\left(\bar{p}_s\right)}\right\rceil\right\}
\]
where $\Omega_{\Gamma}^{K_0} $ is the set of the feasible priority policies of $K_0$ communications among all the $K+U$ communications, $\mathbb{S}_{\alpha}^{K_0}$ is the set of the fraction vector $\bm{\alpha}$ where each element $\alpha_i \in \lbrace 0,\alpha_i^* \rbrace $ and the elements of $\bm{\mu} \left( \Gamma \right)$ which are  $\mu_{i,s}\left( \Gamma\right)$ and $\mu_{j,u}\left( \Gamma\right)$ (for $1 \leq i\leq K$, $1 \leq j\leq U$) are respectively given by (\ref{eqn_mu1c_oneRate_MU}) and (\ref{eqn_mu2c_oneRate_MU}).
\newpage
\bibliography{article1_Bib}
\end{document}